%% file: ApproxShort.tex
\documentclass[a4paper]{article}
\usepackage{fullpage}
\usepackage{mathtools}
\usepackage[pass, margin=1in]{geometry}
\usepackage{times}
\usepackage{amsthm,amsmath}
\usepackage{amsfonts}
\usepackage{relsize}
\usepackage{algorithm}
\usepackage{authblk}
\usepackage[noend]{algpseudocode}
\usepackage{caption}
\usepackage{float}
\usepackage{graphicx}
\usepackage[bottom]{footmisc}
\usepackage[toc,page]{appendix}
\usepackage{hyperref}
\usepackage{authblk}
\newtheorem{theorem}{Theorem}[section]
\newtheorem{lemma}{Lemma}[section]
\newtheorem{corollary}{Corollary}[section]
\newtheorem{definition}{Definition}[section]
\DeclarePairedDelimiter\ceil{\lceil}{\rceil}
\DeclarePairedDelimiter\floor{\lfloor}{\rfloor}

\newcommand{\osubrange}[2]{\ensuremath{ \left( ( 1 + #1)^{#2},~( 1 + #1)^{#2 + 1}  \right] }}

\newcommand{\hatter}[1]{\ensuremath{ \hat{d}(#1)} }
\newcommand{\limitedapprox}[2]{\ensuremath{d^{(2\beta +1)}_{ G^{(k-1) } } (#1, #2 )} }
\newcommand{\hopgraph}{ \ensuremath{ G^{(k-1) } } }
\algnewcommand\algorithmicforeach{\textbf{for each}}
\algdef{S}[FOR]{ForEach}[1]{\algorithmicforeach\ #1\ \algorithmicdo}
\algnewcommand{\LeftComment}[1]{\Statex \(\triangleright\) #1}

\newcommand{\arrow}[1]{\ensuremath{ \overset{\rightarrow}{#1} }}
\hypersetup{
  colorlinks   = true, 
  urlcolor     = blue, 
  linkcolor    = blue, 
  citecolor   = blue 
}
\title{$(1+\epsilon)$-Approximate Shortest Paths in Dynamic Streams}
\author[1]{Michael Elkin \thanks{This research was supported by ISF grant No. 2344/19.}} 
\author[2]{Chhaya Trehan}
\affil[1]{Department of Computer Science, Ben-Gurion University of the Negev, Beer-Sheva, Israel.}
\affil[2]{Department of Mathematics, London School of Economics and Political Science, London, England.}
\affil[1]{Email: elkinm@cs.bgu.ac.il}
\affil[2]{Email: c.trehan@lse.ac.uk}
\date{}

\begin{document}
\begin{titlepage}
\maketitle
\abstract
Computing approximate shortest paths in the dynamic streaming setting is a
fundamental challenge that has been intensively studied during the last decade.
Currently existing solutions for this problem either build a sparse multiplicative spanner
of the input graph and compute shortest paths in the spanner offline, or compute 
an exact single source BFS tree.

Solutions of the first type are doomed to incur a stretch-space tradeoff of $2\kappa-1$ versus
$n^{1+1/\kappa}$, for an integer parameter $\kappa$. (In fact, existing solutions also incur an extra factor of $1+\epsilon$ in the stretch for 
weighted graphs, and an additional factor of $\log^{O(1)}n$ in the space.)
The only existing solution of the second type uses $n^{1/2 - O(1/\kappa)}$ passes over the stream (for space $O(n^{1+1/\kappa})$),
and applies only to unweighted graphs.

In this paper we show that $(1+\epsilon)$-approximate single-source shortest paths can be computed in this setting with
$\tilde{O}(n^{1+1/\kappa})$ space using just \emph{constantly} many passes in unweighted graphs,
and polylogarithmically many passes in weighted graphs (assuming $\epsilon$ and $\kappa$ are constant).
Moreover, in fact, the same result applies for multi-source shortest paths, as long as the number of sources is $O(n^{1/\kappa})$.

We achieve these results by devising efficient dynamic streaming constructions of $(1 + \epsilon, \beta)$-spanners 
and hopsets. We believe that these constructions are of independent interest.
\end{titlepage}

\input{IntroductionNew}

\input{Preliminaries}

\input{BFSDynamic}

\input{WeightedBFS}

\section{Construction of Near-Additive Spanners in the Dynamic Streaming Model}\label{alg:algorithmMain}
\subsection{Overview}\label{sec:Over}
We use the \emph{superclustering and interconnection} approach introduced by Elkin and Peleg in~\cite{ElkinPeleg}, 
which was later refined by Elkin and Neiman~\cite{ElkinNeiman19} (randomized version) and Elkin and Matar~\cite{ElkinShaked19} (deterministic version). 
Specifically, we adapt the randomized algorithm of~\cite{ElkinNeiman19} to work in the dynamic streaming setting. 
The main ingredient of both the superclustering and interconnection steps is a set of BFS explorations up to a given depth in the input graph from a set of chosen vertices.  
As was shown in~\cite{ElkinNeiman19}, their algorithm for constructing near-additive spanners can be easily modified to work with the insertion-only streaming model. 
This is done by identifying the edges spanned by each of the BFS explorations of depth $\delta$ (for an integer parameter $\delta \ge 1$)  by making $\delta$ passes through the stream. 
Other parts of the spanner construction, such as identifying the vertices of the graph from which to perform BFS explorations 
and subsequently adding a subset of edges spanned by these explorations to the spanner, can be performed offline. 
Given parameters $\epsilon > 0,~\kappa =  1,2,\ldots$ and $1/\kappa \le \rho < 1/2$, 
the basic version of their streaming algorithm constructs a spanner with the same stretch and size as their centralized algorithm, 
using $O(n^{1 + \rho} \cdot \log n)$ space whp and $O(\beta)$ passes through the stream. 
Recall that $\beta = \beta(\epsilon, \kappa)$ is defined as $\beta = O(\frac{\log \kappa}{\epsilon})^{\log \kappa}$ (See also Section~\ref{sec:intro}).
 They also provide a slightly different variant of their streaming algorithm which allows one to trade space for the number of passes.
 This variant uses only $O(n\log n + n^{1+ \frac{1}{\kappa}})$ expected space, but it requires $O((n^{\rho}/{\rho}) \cdot \log n \cdot \beta)$ passes. 
 
 We devise a technique to perform BFS traversals up to a given depth from a set of chosen vertices in the graph in the dynamic streaming setting, 
 and as in~\cite{ElkinNeiman19}, perform the rest of the work offline.
 The algorithm for creating a BFS forest starting from a subset of vertices in the graph is described in Section~\ref{sec:BFSForest}. 
 We use the algorithm for creating a BFS forest from a subset of vertices as a subroutine in the superclustering step of our main algorithm. 
 An even bigger challenge we face is during the interconnection step, 
 where each vertex in the graph needs to identify all the BFS explorations it is a part of, and find its path to the source of each such exploration. 
 Due to the dynamic nature of the stream, 
 a given vertex may find itself on a lot more explorations than it finally ends up belonging to. 
 We deal with this problem by combining a delicate encoding/decoding scheme for the IDs of exploration sources with a space-efficient sampling technique inspired by~\cite{GibbKKT15, KingKT15}.
  
 We first provide a high-level overview of the algorithm for constructing the spanner~\cite{ElkinPeleg, ElkinNeiman19, ElkinShaked19}.
 
 Let $G = (V,E)$ be an unweighted, undirected graph on $n$ vertices 
 and let $\epsilon > 0,~\kappa =  2,3,\ldots$ and $1/\kappa \le \rho < 1/2$ be parameters. 
 The algorithm constructs a sparse $(1 + \epsilon, \beta)$ spanner $H = (V, E_{H})$, 
 where $ \beta =  \left(\frac{\log \kappa\rho + 1/\rho}{\epsilon}  \right)^{\log \kappa\rho + 1/\rho}$ and $|E_{H}| =  O_{\epsilon, \kappa} \left(n^{1+ 1/\kappa}  \right)$.

 The algorithm begins by initializing $E_{H}$ as an empty set and proceeds in phases. 
 It starts by partitioning the vertex set $V$ into singleton clusters  $P_{0} =  \{\{v\}~|~v \in V\}$.
 Each phase $i$ for $i = 0,\ldots, \ell$, receives as input a collection of clusters $P_i$,  
 the distance threshold parameter $\delta_i$ and the degree parameter $deg_i$. 
 The maximum phase index $\ell$ is set as  $\ell = \floor{\log \kappa \rho}  + \ceil{\frac{\kappa + 1}{\kappa\rho}} -1$. 
 The values of $\delta_i$ and $\deg_i$ for $i = 0,1,\ldots,\ell$, will be specified later in the sequel.
 
 In each phase, the algorithm samples a set of clusters from $P_i$ 
 and these sampled clusters join the nearby unsampled clusters 
 to create bigger clusters called \emph{superclusters}. 
 Every cluster created by our algorithm has a designated \emph{center} vertex. 
 We denote by $r_C$ the center of cluster $C$ and say that $C$ is \emph{centered around} $r_C$. 
 In particular, each singleton cluster $C = \{v\}$ is centered around $v$. 
 For a cluster $C$, we define $Rad(C) = max\{d_{H}(r_C, v)~|~v \in C\}$. 
 For a set of clusters $P_i$, $Rad(P_i) = \underset{C \in P_i}{max} \{Rad(C)\}$. 
 For a collection $P_i$, we denote by $CP_i$ the set of centers of clusters in $P_i$, i.e., $CP_i = \{r_C~|~C \in P_i\}$.  
 A cluster $C \in P_i$ centered around $r_C$ is considered \emph{close} to another cluster $C' \in P_i$ centered around $r_{C'}$, if $d_G(r_{C}, r_{C'})  \le \delta_i$. 
 
 Each phase $i$, except for the last one, consists of two steps, the \emph{superclustering} step and the \emph{interconnection} step. 
 For a given set of clusters, interconnecting every pair of clusters within a specific distance from each other by adding shortest paths between their respective centers to the spanner guarantees a pretty good stretch for all the vertices in these clusters. 
 However, if a center is close to a lot of other centers, i.e., it is \emph{popular}, interconnecting it to all the nearby centers can add a lot of edges to the spanner. 
 In order to avoid adding too many edges to the spanner while maintaining a good stretch, the process of interconnecting nearby clusters is preceded by the process of superclustering.
 
 The \emph{superclustering} step of phase $i$ randomly samples a set of clusters in $P_i$ and builds larger clusters around them. 
 The sampling probabilities will be specified in the sequel. For each new cluster $C$, a BFS tree of $C$ is added to the spanner $H$. 
 The collection of the new larger clusters is passed on as input to phase $i + 1$.
 
 In the \emph{interconnection} step of phase $i$, the clusters that were not superclustered in this phase are connected to their nearby clusters. 
 For each cluster center $r_{C}$ that was not superclustered, paths to all the nearby centers in $CP_i$ (whether superclustered or not) are added to the spanner $H$. 
 Since $r_C$ was not superclustered, it does not have any sampled cluster centers nearby, as otherwise such a center would have superclustered it. 
 This ensures that, with high probability, we do not add too many edges to the spanner during the \emph{interconnection} step.
 
 In the last phase $\ell$ the superclustering step in skipped and we go directly to the interconnection step. As is shown in~\cite{ElkinNeiman19}, 
 the input set of clusters to the last phase $P_{\ell}$ is sufficiently small to allow us to interconnect all the centers in $P_{\ell}$ to one  another using few edges. 
 
 Next we describe the input parameters, the degree parameter $\deg_i$ and the distance threshold parameter $\delta_i$ of the phase $i$, for each $i =0,1,\ldots, \ell$.
 The distance threshold parameter $\delta_i$ is defined as $\delta_i = (1/\epsilon)^i + 4R_i$, where $R_i$ is determined by the following recurrence relation: $R_0 = 0$, $R_{i+1} = R_i + \delta_i$.
As is shown in~\cite{ElkinNeiman19}, $R_i$ is an upper bound on the radius of the clusters in $P_i$. The distance threshold parameter $\delta_i$ determines the radii of superclusters, 
and it also affects the definition of nearby clusters for the interconnection step. 
The degree threshold parameter $deg_i$ of phase $i$ is used to define the sampling probability with which the centers of clusters in $P_i$ are selected to grow superclusters  around them. 
Specifically, in phase $i$, $i = 0,1,\ldots \ell -1$,  each cluster center $r_C \in CP_i$ is sampled independently at random with probability $1/deg_i$. 
The sampling probability affects the number of superclusters created in each phase and hence the number of phases of the algorithm. 
It also affects the number of edges added to the spanner during the interconnection step.
 We partition the first $\ell -1$ phases into two stages based on how the degree parameter grows in each stage. 
The two stages of the algorithm are the \emph{exponential growth stage} and the \emph{fixed grown stage}. 
In the \emph{exponential growth stage}, which consists of phases $0,1, \ldots, i_0 =\log \floor{\kappa\rho}$, 
we set  $deg_i = n^{\frac{2^i}{\kappa}}$. In the \emph{fixed growth stage}, which consists of phases $i_0 + 1, i_0 +2,\ldots, i_1 = i_0 + \ceil{\frac{\kappa + 1}{\kappa \rho}}$, we set $deg_i = n^{\rho}$.
Observe that for every index $i$, we have $deg_i \le n^{\rho}$.

\input{Superclustering}

\input{Interconnection}

\input{Final}

\input{Applications}

\input{Hopset}

\bibliography{DynamicStream} 
\bibliographystyle{plain}

\input{Appendix}

\end{document}

%% file: IntroductionNew.tex
\section{Introduction}\label{sec:intro}
\subsection{Graph Streaming Algorithms}\label{sec:graphStreaming}
Processing massive graphs is an important algorithmic challenge.
This challenge is being met by intensive research effort.
One of the most common theoretical models for addressing 
this challenge is the \emph{semi-streaming} 
model of computation~\cite{FeigenbaumKannan, AhnGuha12a,McGregorSurvey}.
In this model, edges of an input $n$-vertex graph $G = (V, E)$ arrive one after 
another, while the storage capacity of the algorithm is limited.
Typically it should be close to linear in the number of \emph{vertices}, $n$
(as opposed to being linear in the number of edges $m = |E|)$.
In particular, one usually allows space of $\tilde{O}(n)$, though it is often 
relaxed to $n^{1 + o(1)}$, sometimes to $O(n^{1 + \rho})$, for an arbitrarily small constant 
parameter $\rho > 0$, or even to $O(n^{1 + \eta_0})$, for some fixed constant $\eta_0$, 
$0 < \eta_0 < 1$.
Generally, the model allows several passes over the stream, and the objective
is to keep both the number of passes and the space complexity of the algorithm in check.

The model comes in two main variations. 
In the first one, called \emph{static} or \emph{insertion-only} model~\cite{FeigenbaumKannan}, 
the edges can only arrive, and never get deleted.
If the algorithm employs multiple passes, then the streams of edges observed on 
these passes may be permutations of one another, but are otherwise identical.
In the more general \emph{dynamic} (also known as \emph{turnstile}) streaming setting~\cite{AhnGuha12a},
edges may either arrive or get deleted. On each of the passes, each element of the stream is of
the form $(e_i, \sigma_i)$, where $e_i \in E$ is an edge of the input graph and $\sigma_i \in \{+1, -1\}$
is a sign indicating whether the edge is being inserted or removed. Ultimately, at the end of each pass,
for every edge $e \in E$, it holds that $\sum_{e_i = e | (e_i, \sigma_i)\text{~is in the stream}} \sigma_i = 1$,
while for every non-edge $e'$, the corresponding sum is equal to $0$.


\subsection{Distances in the Streaming Model}\label{sec:streamingDistances}
An important thread of the literature on dynamic streaming algorithms for
graph problems is concerned with computing \emph{distances} and constructing \emph{spanners} and \emph{hopsets}.
This is also the topic of the current paper. 
For a pair of parameters $\alpha \ge 1$, $\beta \ge 0$, 
given an undirected graph $G = (V,E)$, 
a subgraph $G' = (V, H)$ of $G$ is said to be an \emph{($\alpha, \beta$)-spanner} of $G$, 
if for every pair $u,v \in V$ of vertices, it holds that 
$d_{G'}(u,v) \le \alpha \cdot d_{G}(u.v) + \beta$, 
where $d_G$ and $d_{G'}$ are the distance functions of $G$ and $G'$, respectively.
A spanner with $\beta = 0$ is called a \emph{multiplicative} spanner and one with $\alpha = 1$ is called an \emph{additive} spanner. 
There is another important variety of spanners called \emph{near-additive} spanners for which $\beta \ge 0$ and
$\alpha = 1 + \epsilon$, for an arbitrarily small $\epsilon > 0$.
The near-additive spanners are mostly applicable to \emph{unweighted} graphs, even though there are 
some recent results about weighted near-additive spanners~\cite{ElkinNeimanGitlitz20}.

Spanners are very well-studied from both combinatorial and algorithmic viewpoints.
It is well-known that for any parameter $\kappa = 1,2, \ldots,$ and for any $n$-vertex graph
$G = (V,E)$, there exists a $(2\kappa -1)$-spanner with $O(n^{1 + 1/\kappa})$ edges, 
and this bound is nearly-tight unconditionally, and completely tight under 
Erdos-Simonovits girth conjecture~\cite{PelegSchaffer, althofer1989}.
The parameter $2\kappa -1$ is called the \emph{stretch} parameter of the spanner.
Also, for any pair of parameters, $\epsilon >0$ and $\kappa = 1,2, \dots,$ there exists 
$\beta = \beta_{EP} = \beta(\kappa, \epsilon)$, so that for every $n$-vertex undirected graph
$G = (V,E)$, there exists a $(1+\epsilon, \beta)$-spanner 
with $O_{\kappa, \epsilon}(n^{1 + 1/\kappa})$ edges~\cite{ElkinPeleg}.
The additive term $\beta = \beta_{EP}$ in~\cite{ElkinPeleg} behaves as 
$\beta(\kappa, \epsilon) \approx \left(\frac{\log \kappa}{\epsilon} \right)^{\log \kappa}$,
and this bound is the state-of-the-art. 
A lower bound of $\Omega(\frac{1}{\epsilon \cdot \log \kappa})^{\log \kappa}$ 
for it was shown by Abboud et al.~\cite{AbboudBodwin16}.

 Given an $n$-vertex weighted undirected graph $G = (V,E,\omega)$ and
 two parameters $\epsilon > 0$ and $\beta = 1,2,\ldots$, a graph $G' = (V, H, \omega')$ is
called a $(1+\epsilon, \beta)$-\emph{hopset} of $G$, if for every pair of vertices $u, v \in V$, 
we have
\begin{align}
d_G(u,v) \le d^{(\beta)}_{G \cup G'}(u,v) \le (1 + \epsilon) \cdot d_G(u,v)
\end{align}
Here  $d^{(\beta)}_{G \cup G'}(u,v)$ stands for $\beta$-bounded distance (See Definition~\ref{def:tLimitedDist}) 
between $u$ and $v$ in $G \cup G'$. 
(Note that for a weighted graph $G = (V,E, \omega)$, the weight of a non-edge $(u,v) \notin E$ is defined as $\omega((u,v)) = \infty$, and
the weight of an edge $(x,y)$ in the edge set of $G \cup G'$ is given by $\min\{\omega(x,y),~\omega'(x,y) \}$.)
The parameter $\beta$ is called the \emph{hopbound} of the hopset $G'$.
We often refer to the edge set $H$ of $G'$ as the \emph{hopset}.
Just like spanners, hopsets are a fundamental graph-algorithmic construct.
They are extremely useful for computing approximate shortest distances and paths
in various computational settings, in which computing shortest paths with a limited number of
hops is significantly easier than computing them with no limitation on the number of hops.
A partial list of these settings includes streaming, distributed, parallel and centralized dynamic models.
Recently, hopsets were also shown to be useful for computing approximate
shortest paths in the standard centralized model of computation as well~\cite{ElkinNeimanShortestPaths}.

Cohen~\cite{Cohen1994Polylogtime} showed that for any undirected weighted $n$-vertex  graph $G$, and parameters 
$\epsilon >0$, $\rho >0$, and $\kappa = 1,2,\dots$, there exists a $(1+ \epsilon, \beta_{C})$-hopset with
$\tilde{O}(n^{1 + 1/\kappa})$ edges, where $\beta_c = \left( \frac{\log n}{\epsilon} \right)^{O\left(\frac{\log \kappa}{\rho}\right)}$.
Elkin and Neiman~\cite{ElkinNeimanHopsets} improved Cohen's result,
and constructed hopsets with \emph{constant} hopbound.
Specifically, they showed that for any $\epsilon > 0$, $\kappa = 1,2,\ldots$, and any $n$-vertex 
weighted undirected graph, there exists a $(1+\epsilon, \beta_{EN})$-hopset with $\tilde{O}(n^{1 +1/\kappa})$
edges, and $\beta_{EN} = \beta_{EP} \approx (\frac{\log \kappa}{\epsilon})^{\log \kappa}$.
The lower bound of Abbound et al.~\cite{AbboudBodwin16}, 
$\beta = \Omega(\frac{1}{\epsilon \cdot \log \kappa})^{\log \kappa}$ is applicable to hopsets as well.
Generally, hopsets (see~\cite{Cohen1994Polylogtime, HenzingerKNanon16, ElkinNeimanHopsets})
are closely related to near-additive spanners. See a recent survey~\cite{ElkinN20}
for an extensive discussion on this relationship.

Most of the algorithms for computing (approximate) distances and 
shortest paths in the streaming setting compute a sparse spanner, 
and then employ it for computing exact shortest paths and distances
in it offline, i.e., in the post-processing, after the stream is 
over~\cite{KanaanMcGregorDisatncesStreams, ElkinStreaming2, 
 BASWANAStreaming08, ElkinZhangStreaming, ElkinNeiman19,
 AhnGuha12, KaprolovWoodruff, FernandezWYasuda020,
 Kaprolov1Pass}. Feigenbaum et al.~\cite{KanaanMcGregorDisatncesStreams} devised the first efficient
\emph{static} streaming algorithm for building multiplicative spanners.
Their algorithm produces a $(2\kappa +1)$-spanner with $O(n^{1 + 1/\kappa} \kappa^2 \log^2 n)$
edges (and this is also the space complexity of the algorithm) in a single pass,
and its processing time per edge is $\tilde{O}(n^{1/\kappa})$, for a parameter $\kappa = 1,2,\ldots$.
More efficient static streaming algorithms for this problem, that also provide spanners with
a better stretch-size tradeoff, were devised in~\cite{ElkinStreaming2, BASWANAStreaming08}.
Specifically, these static streaming algorithms construct $(2\kappa -1)$-spanners of size
$\tilde{O}(n^{1 + 1/\kappa})$ (and using this space), and as a result produce $(2\kappa -1)$-approximate
all pairs shortest paths (henceforth, $(2\kappa -1)$-APASP) using space $\tilde{O}(n^{1 + 1/\kappa})$ 
in a single pass over the stream. 

The algorithms of~\cite{KanaanMcGregorDisatncesStreams, ElkinStreaming2, BASWANAStreaming08}
apply to unweighted graphs, but they can be extended to weighted graphs by running many copies of them
in parallel, one for each weight scale. Let $\Lambda = \Lambda(G)$ denote the \emph{aspect ratio}
of the graph, i.e., the ratio between the maximum distance between some pair of vertices in $G$
and the minimum distance between a pair of distinct vertices in $G$. 
Also, let $\epsilon \ge 0$ be a slack parameter. Then by running $O(\frac{\log \Lambda}{\epsilon})$
copies of the algorithm for unweighted graphs and taking the union of their outputs as the ultimate spanner,
one obtains a one-pass static streaming algorithm for $2(1 + \epsilon)\kappa$-spanner with 
$\tilde{O}(n^{1 + \frac{1}{\kappa}} \cdot (\log \Lambda)/\epsilon)$ edges.
See, for example,~\cite{ElkinSol13} for more
details. 

Elkin and Zhang~\cite{ElkinZhangStreaming} devised a static streaming algorithm for
building $(1+\epsilon, \beta_{EZ} )$-spanners with $\tilde{O}(n^{1+1/\kappa})$ edges using $\beta_{EZ}$ passes
over the stream and space $\tilde{O}(n^{1+\rho})$, where $\beta_{EZ} = \beta_{EZ}(\epsilon, \rho, \kappa) = 
\left(\frac{\log \kappa}{\epsilon \cdot \rho} \right)^{O(\frac{\log \kappa}{\rho})}$, for any parameters 
$\epsilon, \rho > 0$ and $\kappa =1,2,\ldots$. This result was improved in~\cite{ElkinNeiman19},
where a static streaming algorithm with similar properties, but with 
$\beta = \beta_{EN} = \left(\frac{\log \kappa\rho + 1/\rho}{\epsilon}  \right)^{\log \kappa\rho + 1/\rho}$
was devised. The algorithms of~\cite{ElkinZhangStreaming, ElkinNeiman19} directly
give rise to $\beta$-pass static streaming algorithms with space $\tilde{O}(n^{1 + \rho})$
for $(1 + \epsilon, \beta)$-APASP in unweighted graphs where $\beta(\rho) \approx (1/\rho)^{(1/\rho)(1 + o(1))}$.
They can also be used for producing purely multiplicative $(1+\epsilon)$-approximate 
shortest paths and distances in $O(\beta/\epsilon)$ passes and $\tilde{O}(n^{1+\rho})$ space from up to
$n^{\rho(1-o(1) )} $ designated sources to all other vertices.

There are also a number of additional \emph{not} spanner-based static streaming algorithms for 
computing approximate shortest paths. Henzinger, Krinninger and Nanongkai~\cite{HenzingerSingleDeterministic}
and Elkin and Neiman~\cite{ElkinNeimanHopsets} devised $(1 + \epsilon)$-approximate \emph{single}-source
shortest paths (henceforth, SSSP) algorithms for weighted graphs, that are based on \emph{hopsets}.
The $(1 + \epsilon)$-SSSP algorithm of~\cite{HenzingerKNanon16} employs $2^{O(\sqrt{\log n \log \log n})} = n^{o(1)}$
passes and space $n \cdot 2^{O(\sqrt{ \log n \cdot \log \log n})} \cdot O(\frac{\log \Lambda}{\epsilon}) = n^{1 + o(1)} \cdot O(\frac{\log \Lambda}{\epsilon})$.
Elkin and Neiman~\cite{ElkinNeimanHopsets} generalized and improved this result.
For any parameters $\epsilon, \rho >0$, their static streaming algorithm computes $(1+ \epsilon)$-approximate
SSSP using $\tilde{O}(n^{1+ \rho})$ space and $\left(\frac{\log n}{\epsilon \cdot \rho} \right)^{\frac{1}{\rho}(1 + o(1))}$
passes. Moreover, in fact the same bound for number of passes and space applies in the algorithm of~\cite{ElkinNeimanHopsets}
for computing $S \times V$ $(1+ \epsilon)$-approximately shortest paths, for any subset $S \subseteq V$ of up to $n^{\rho}$
designated sources. Yet more efficient static streaming algorithm for $(1 + \epsilon)$-approximate 
SSSP was devised by Becker et al.~\cite{BeckerKKL17} using techniques from the field of continuous optimization.
Their static streaming algorithm uses polylogarithmically many passes over the stream and space $O(n \cdot polylog (n) )$.
Finally, an \emph{exact} static streaming SSSP algorithm was devised in~\cite{ElkinExact17}.
For any parameter $1 \le p \le n$, it requires $O(n/p)$ passes and $O(n \cdot p)$ space, and applies to weighted undirected graphs.
The algorithm of~\cite{ElkinExact17} also applies to the problem of computing $S \times V$ approximately shortest paths for $|S| \le p$,
and requires the same pass and space complexities as in the single-source case.

Recently Chang et al.~\cite{ChangFHT20} devised a \emph{dynamic streaming} algorithm 
for this problem in \emph{unweighted} graphs. Their algorithm uses $\tilde{O}(n/p)$
passes (for parameter $1 \le p \le n$ as above) and space $\tilde{O}(n + p^2)$ for the
SSSP problem, and space $\tilde{O}(|S|n + p^2)$ for the $S \times V$ approximate 
shortest path computation.
Ahn, Guha and McGregor~\cite{AhnGuha12} devised the first \emph{dynamic streaming} 
algorithm for computing approximate distances. Their algorithm computes a $(2\kappa -1)$-spanner
(for any $\kappa = 1,2, \ldots$) with $\tilde{O}(n^{1 + 1/\kappa})$ edges (and the same space
complexity) in $\kappa$ passes over the stream. This bound was recently improved by 
Fernandez, Woodruff and Yasuda~\cite{FernandezWYasuda020}. Their algorithm computes
a spanner with the same properties using $\floor{\kappa/2} + 1$ passes.
Ahn et al.~\cite{AhnGuha12} also devised an $O(\log \kappa)$-pass algorithm
for building $O( \kappa^{ \log_{2} 5 } )$-spanner with size and space complexity
$\tilde{O}(n^{1 + 1/\kappa})$. This bound was recently improved by 
Filtser, Kapralov and Nouri~\cite{Kaprolov1Pass}, whose algorithm produces
$O(\kappa^{\log_2 3})$-spanner with the same pass and space complexities, and the same size.
Another dynamic streaming algorithm was devised by Kapralov and Woodruff~\cite{KaprolovWoodruff}.
It produces a $(2^{\kappa} -1)$-spanner with $\tilde{O}(n^{1 + 1/\kappa})$ edges (and space usage)
in two passes. Filtser et al.~\cite{Kaprolov1Pass} improved the stretch parameter of the spanner
to $2^{\frac{\kappa + 3}{2}} -3$, with all other parameters the same as in the results of~\cite{KaprolovWoodruff}.
Filtser et al.~\cite{Kaprolov1Pass} also devised a general tradeoff in which the number of passes
can be between $2$ and $\kappa$, and the stretch of the spanner decreases gradually from exponential in $\kappa$ 
(where the number of passes is $2$) to $2\kappa -1$ (when the number of passes is $\kappa$).
They have also devised a single pass algorithm with stretch $\tilde{O}(n^{\frac{2}{3}(1 - 1/\kappa)})$.
As was mentioned above, all these spanner-based algorithms provide a solution 
for $(2\kappa -1)$-approximate all pairs almost shortest paths
(henceforth, $(2\kappa -1)$-APASP) for unweighted graphs with space $\tilde{O}(n^{1 + 1/\kappa})$
and the number of passes equal to that of the spanner-construction algorithm.
Like their static streaming counterparts~\cite{KanaanMcGregorDisatncesStreams, ElkinStreaming2, BASWANAStreaming08},
they can be extended to weighted graphs, at the price of increasing their stretch by a factor of $1+ \epsilon$ (for an 
arbitrarily small parameter $\epsilon >0$), and their space usage by a factor of $O\left(\frac{\log \Lambda}{\epsilon} \right)$.

To summarize, all known dynamic streaming algorithms for computing approximately 
shortest paths (with space $\tilde{O}(n^{1 + 1/\kappa})$, for a parameter $\kappa =1,2,\ldots$), 
can be divided into two categories.
The algorithms in the first category build a sparse multiplicative $(2\kappa -1)$-spanner, 
and they provide a \emph{multiplicative} stretch of at least 
$2\kappa -1$~\cite{AhnGuha12, KaprolovWoodruff, FernandezWYasuda020, Kaprolov1Pass}.
Moreover, due to existential lower bounds for spanners, this approach is doomed to provide stretch of at least $\frac{4}{3} \kappa$~\cite{LazebnikU92}.
The algorithms in the second category compute \emph{exact} 
\emph{single source} shortest paths in \emph{unweighted} graphs, but they employ $n^{1/2 - O(1/\kappa)}$
passes~\cite{ChangFHT20, ElkinExact17}.

%
\subsection{Our Results}
In the current paper, we present the first dynamic streaming algorithm for SSSP with stretch $1+\epsilon$, 
space $\tilde{O}(n^{1+1/\kappa})$, and \emph{constant} (as long as $\epsilon$ and $\kappa$ are constant)
number of passes for unweighted graphs.
For weighted graphs, our number of passes is \emph{polylogarithmic} in $n$.
Specifically, the number of passes of our SSSP algorithm is $(\frac{\kappa}{\epsilon})^{\kappa(1+ o(1))}$ for 
unweighted graphs, and $\left( \frac{(\log n)\kappa}{\epsilon} \right)^{\kappa(1+ o(1))}$ for  weighted ones.
Moreover, within the same complexity bounds, 
our algorithm can compute $(1 + \epsilon)$-approximate $S \times V$ shortest paths
from $|S| = n^{1/\kappa}$ designated sources.
Moreover, in \emph{unweighted} graphs, \emph{all} pairs almost shortest paths with 
stretch $(1 + \epsilon, \left(  \frac{\kappa}{\epsilon}\right)^{\kappa})$
can also be computed within the same space and number of passes.
(That is, paths and distances with multiplicative stretch $1 + \epsilon$ and additive 
stretch $\left ( \frac{\kappa}{\epsilon} \right)^{\kappa }$.)
Note that our multiplicative stretch $(1 + \epsilon)$ is dramatically better than 
$(2\kappa -1)$, exhibited by algorithms based on multiplicative spanners
~\cite{AhnGuha12, KaprolovWoodruff, FernandezWYasuda020, Kaprolov1Pass},
but this comes at a price of at least exponential increase in the number of passes.
Nevertheless, our number of passes is \emph{independent}  \emph{of $n$},
for unweighted graphs, and depends only polylogarithmically on $n$ for weighted ones.


\subsection{Technical Overview}
We devise two algorithms which build structures that help us compute approximate shortest paths.
One of them builds a near-additive spanner and the other builds a near-exact hopset.
The following two theorems summarize the results of our spanner and hopset constructions.
\vspace{-1em}
\begin{theorem}(Theorem~\ref{thm:main} in Appedix~\ref{alg:algorithmMain})
For any unweighted graph $G(V,E)$ on $n$ vertices, parameters $0 < \epsilon <  1$, $\kappa \ge 2$, and $\rho >0$, 
our dynamic streaming algorithm computes a $(1 + \epsilon, \beta)$-spanner with  
$O_{\epsilon,\kappa, \rho}(n^{1 + 1/\kappa})$ edges, in $O(\beta)$ passes using $O(n^{1 + \rho}\log^4 n)$ space with high probability, where $\beta$ is given by:
\begin{align*}
   \beta =  \left(\frac{\log \kappa\rho + 1/\rho}{\epsilon}  \right)^{\log \kappa\rho + 1/\rho}.
\end{align*}
\end{theorem}

\begin{theorem}(Theorem~\ref{thm:AspRatioReductionS} in Section~\ref{sec:hopSummary})
For any $n$-vertex graph $G(V, E, \omega)$ with aspect ratio $\Lambda$, $2 \le \kappa \le (\log n)/4$,
$1/\kappa \le \rho \le 1/2$ and $0 < \epsilon' < 1$, our dynamic streaming algorithm computes whp,
a $(1+\epsilon', \beta')$ hopset $H$ with expected size $O(n^{1+1/\kappa} \cdot \log n)$  and the hopbound $\beta'$ given by
\[
\beta' =  O\left (\frac{(\log \kappa \rho + 1/\rho )\log n}{\epsilon'}  \right)^{\log \kappa \rho + 1/\rho}
\]
It does so by making $O(\beta' \cdot (\log \kappa \rho + 1/\rho))$ passes through the stream and
using $O(n \cdot \log^3 n \cdot \log \Lambda)$ bits of space in the first pass and 
$O(\frac{\beta'}{\epsilon'} \cdot \log^2 1/\epsilon'  \cdot  n^{ 1 +\rho} \cdot \log^5 n)$ 
bits of space (respectively $O(\frac{\beta'^2}{\epsilon'} \cdot \log^2 1/\epsilon'  \cdot  n^{ 1 +\rho} \cdot \log^5 n)$ bits of space for path-reporting hopset)
  in each of the subsequent passes.
\end{theorem}

The hopset is then used ( in Section~\ref{sec:hopApplications}) to compute approximate shortest paths in weighted graphs
 and the spanner is used (in Section~\ref{sec:Applications})  to compute approximate shortest paths in unweighted graphs.

Our algorithms for spanner and hopset construction extend the results of~\cite{ElkinNeiman19, ElkinNeimanHopsets} from the static streaming setting to dynamic streaming one.
The algorithms of~\cite{ElkinNeiman19, ElkinNeimanHopsets}, like their predecessor, the algorithm of~\cite{ElkinPeleg}, 
are  based on the superclustering-and-interconnection (henceforth, SAI) approach.
Our algorithms in the current paper also fall into this framework.
Algorithms that follow the SAI approach proceed in phases, and in 
each phase they maintain a partial partition of the vertex set $V$ of the graph.
Some of the clusters of $G$ are selected to create superclusters around them.
This is the superclustering step. 
Clusters that are not superclustered into these
superclusters are then \emph{interconnected} with their nearby clusters.
The main challenge in implementing this scheme in the dynamic streaming setting 
is in the interconnection step. 
Indeed, the superclustering step requires a single 
and rather shallow BFS exploration, and implementing depth-$d$ BFS in unweighted graphs
 in $d$ passes
over the dynamic stream can be done in near-linear space
 (See, e.g.,~\cite{AhnGuha12,ChangFHT20}).
 For the weighted graphs, we devise a routine for performing 
 an approximate Bellman-Ford exploration up to a given hop-depth $d$,
 using $d$ passes and $\tilde{O}(n)$ space.
 
 On the other hand, the interconnection step requires implementing 
simultaneous BFS explorations originated at multiple sources.
A crucial property that enabled~\cite{ElkinNeiman19, ElkinNeimanHopsets} to implement it
in the static streaming setting is that one can argue that with high probability,
not too many BFS explorations traverse any particular vertex.
Let us denote by $N$, an upper bound on the number of 
explorations (traversing any particular vertex).
In the dynamic streaming setting, however, at any point of the stream,
there may well be much more than $N$ explorations that traverse a specific vertex $v \in V$, 
based on the stream of updates observed so far. Storing data about all these 
explorations would make the space requirement of the algorithm prohibitively large.

To resolve this issue (and a number of related similar issues), 
we incorporate a \emph{sparse recovery} routine into our algorithms.
Sparse recovery is a fundamental and well-studied primitive in the 
dynamic streaming setting~\cite{Ganguly2007CountingDI, CormodeFirmani,PriceIndykW11, BaIPW10}.
It is defined for an input which is a stream of (positive and negative) updates to an $n$-dimensional vector 
$\arrow{a} = (a_1,a_2, \ldots, a_n)$.
In the \emph{strict} turnstile setting, which is sufficient for our application, ultimately
each coordinate $a_i$ (i.e., at the end of the stream) is non-negative, 
even though negative updates are allowed and intermediate values of 
coordinates may be negative. In the \emph{general} turnstile model
coordinates of the vector $\arrow{a}$ may be negative at the end of the stream as well.
The \emph{support} of $\arrow{a}$, denoted $supp(\arrow{a})$, is defined as the set of
its non-zero coordinates. For a parameter $s$, 
an \emph{s-sparse recovery} routine returns the vector $\arrow{a}$, if $|supp(\arrow{a})| \le s$, 
and returns failure otherwise. (It is typically also allowed to return failure with some small probability $\delta >0$, 
given to the routine as a parameter, even if  $|supp(\arrow{a})| \le s$.)

Most of sparse recovery routines are based on $1$-sparse recovery, i.e., the case $s =1$.
The first $1$-sparse recovery algorithm was devised by Ganguly~\cite{Ganguly2007CountingDI},
and it applies to the strict turnstile setting.
The space requirement of the algorithm of~\cite{Ganguly2007CountingDI} is $O(\log n)$.
The result was later extended to the general turnstile setting by Cormode and Fermini~\cite{CormodeFirmani}
(See also,~\cite{MonemizadehW10}).
We devise an alternative streaming algorithm for this basic task in the strict turnstile setting.
The space complexity of our algorithm is $O(\log n)$, like that of~\cite{Ganguly2007CountingDI}.
The processing time-per-item of Ganguly's algorithm is however $O(1)$, 
instead of $polylog(n)$ of our algorithm.~\footnote{If the algorithm knows in advance the dimension $n$ of the vector $\arrow{a}$ and is allowed to compute during preprocessing, 
before seeing the stream, a table of size $n$, then our algorithm can also have $O(1)$ processing time per update.
This scenario occurs in dynamic streaming graph algorithms, including those discussed in the current paper.}
Nevertheless, we believe that our new algorithm for this task is of independent interest.
Appendices~\ref{sec:firstAppedix} and~\ref{apx:l0Sampling} are devoted to our new sparse recovery procedure, and its applications to $L_0$-sampling.
In Appendix~\ref{sec:firstAppedix}, we describe this procedure, and in Appendix~\ref{apx:l0Sampling}, we show how it can be used to build $\ell_0$-samplers,
(See Appendix~\ref{apx:l0Sampling} for their definitions) with complexity that matches the state-of-the-art
bounds for $\ell_0$-samplers due to Jowhari, Sa\u{g}lam and Tardos~\cite{Jowhari11},
but are arguably somewhat simpler.

For the greater part of the paper we analyze our algorithm in terms of the aspect ratio $\Lambda$ of the input graph, given by
$\Lambda = \frac{max_{u,v \in V}d_G(u,v)}{min_{u,v \in V d_G(u,v)} }$.
(All dependencies are polylogarithmic in $\Lambda$.) In Section~\ref{sec:aspectRatio}, however, we show that 
Klein-Subramania's weight reduction~\cite{KLEINSub97} (see also~\cite{ElkinNeimanHopsets}) can be implemented
in the dynamic streaming model. As a result, we replace all appearances of $\log \Lambda$ in the 
hopset's size, hopbound and number of passes of our construction by $O(\log n)$.
However, the space complexity of our algorithm still mildly depends on $\log \Lambda$.
Specifically, it is $\tilde{O} (n^{1+ \rho}) + \tilde{O} (n) \cdot \log \Lambda$.
In all existing dynamic streaming algorithms for computing multiplicative spanners or computing approximate 
shortest paths in weighted graphs~\cite{AhnGuha12, KaprolovWoodruff, FernandezWYasuda020, Kaprolov1Pass},
both the spanner's size and the space requirements are $\tilde{O} (n^{1 + 1/\kappa} \cdot \log \Lambda)$. 
Completely eliminating the dependence on $\log \Lambda$ from these results is left as an open problem.

\subsection{Outline} 
The rest of the paper is organized as follows.
Section~\ref{sec:prelims} provides necessary definitions and concepts. 
Sections~\ref{sec:BFSForest} and~\ref{sec:WeightedForest} provide the subroutines required for our main algorithms presented in Sections~\ref{alg:algorithmMain}-\ref{sec:hopApplications}.
Section~\ref{sec:BFSForest} describes an algorithm for building a BFS forest of a given depth rooted at a subset of vertices of an unweighted input graph.
Section~\ref{sec:WeightedForest} describes an algorithm for performing an approximate Bellman-Ford exploration rooted at a subset of vertices of a weighted input graph.
Section~\ref{alg:algorithmMain} presents an algorithm for constructing near-additive spanners, and
Section~\ref{sec:Applications} shows how we use the algorithm of Section~\ref{alg:algorithmMain} to compute $(1+\epsilon)$-approximate shortest paths in unweighted graphs.
Section~\ref{sec:Hopsets} presents an algorithm for constructing hopsets with constant hopbound,
 and Section~\ref{sec:hopApplications} shows how we use the algorithm of Section~\ref{sec:Hopsets} to compute $(1+\epsilon)$-approximate shortest paths in weighted graphs.

%% file: Preliminaries.tex
\section{Preliminaries}\label{sec:prelims}
\subsection{Streaming Model}\label{sec:Model}
In the streaming model of computation, the set of vertices $V$ of the input graph is known in advance and the edge set $E$ is revealed one at a time. 
In an \emph{insertion-only stream} the edges can only be inserted, and once inserted an edge remains in the graph forever. 
In a \emph{dynamic stream}, on the other hand, the edges can be added as well as removed. 
We will consider unweighted graphs for our spanner construction algorithm and weighted graphs for our hopset construction algorithm.
For an unweighted input graph, the stream $S$ arrives as a sequence of edge updates $S = \langle s_1, s_2,\cdots \rangle$, 
where  $s_t = (e_t, eSign_t)$, where $e_t$ is the edge being updated.
For a weighted input graph, the stream $S$ arrives as a sequence of edge updates $S = \langle s_1, s_2,\cdots \rangle$, 
where  $s_t = (e_t, eSign_t, eWeight_t)$, where $e_t$ is the edge being updated and $eWeight_t$ is its weight.
In unweighted as well weighted case, the $eSign_t \in \{+1, -1\}$ value of an update indicates whether the edge $e_t$ is to be added or removed. 
A value of $+1$ indicates addition and a value of $-1$ indicates removal. 
There is no restriction on the order in which the $eSign$ value of a specific edge $e$ changes. 
The multiplicity of an edge $e$ is defined as $f_e = \sum_{t, e_t = e} eSign_t$. 
We assume that for every edge $e$, $f_e \in \{0, 1\}$ at that at the end of the stream.
The order in which updates arrive may change from one pass of the stream to the other, 
while the final adjacency matrix of the graph at the end of every pass remains the same. 
We assume that the length of the stream or the number of updates we receive is $poly(n)$.
For more details on the streaming model of computation for graphs, 
we refer the reader to the survey~\cite{McGregorSurvey} on graph streaming algorithms.

\theoremstyle{definition}
\begin{definition}\label{def:vSetDegree}
For a vertex $v \in V$ and a vertex set $\mathcal{U} \subseteq V$, 
the \textbf{\emph{degree of $v$ with respect to~$\mathcal{U}$}} is the number of edges connecting $v$ to the vertices in $\mathcal{U}$. 
\end{definition}

For a weighted undirected graph $G = (V,E,\omega)$,
we assume that the edge weights are scaled so that the minimum edge weight is $1$.
Let $maxW$ denote the maximum edge weight $\omega(e)$, $e\in E$. 
For a non-edge $(u,v) \notin E$, we define $\omega((u,v)) = \infty$.

Denote also by $\Lambda$ the \emph{aspect ratio} of the graph, i.e., the maximum \emph{finite} distance between some pair $u,v$ of vertices 
(assuming that the minimum edge weight is $1$).
\theoremstyle{definition}
\begin{definition}\label{def:tLimitedPath}
Given a weighted graph $G(V,E, \omega)$, a positive integer parameter $t$, and a pair $u, v \in V$ of distinct vertices,
a  \textbf{$t$-\emph{bounded} $u$-$v$ path} in $G$ is a path between $u$ and $v$ that contains no more than $t$ edges (also known as \emph{hops}).
\end{definition}

\theoremstyle{definition}
\begin{definition}\label{def:tLimitedDist}
Given a weighted graph $G(V,E, \omega)$, a positive integer parameter $t$, and a pair $u, v \in V$ of distinct vertices, 
 \textbf{\emph{$t$-bounded distance}} between $u$ and $v$ in $G$ denoted  $d^{(t)}_{G}(u,v)$
is the length of the shortest $t$-bounded $u$-$v$ path in $G$.
\end{definition}

Note that all logarithms are to the base $2$ unless explicitly stated otherwise.
We use $\tilde{O}(f(n))$ as a shorthand for $O(f(n) \cdot polylog n)$.

%


\subsection{Samplers}
The main technical tool in our algorithms is a space-efficient sampling technique 
which enables us to sample a single vertex or a single edge 
from an appropriate subset of the vertex set or the edge set of the input graph, respectively. 
Most graph streaming algorithms use standard $\ell_0$-sampler due to Jowhari et al.~\cite{Jowhari11} 
as a blackbox to sample edges or vertices from a graph.
 An $\ell_0$-sampler enables one to sample almost uniformly from the support of a vector. 
We present an explicit construction of a sampling technique inspired by ideas from~\cite{KingKT15, GibbKKT15, ElkinCSmith}.
Our construction is arguably simpler than the standard $\ell_0$-sampler due to Jowhari et al.~\cite{Jowhari11} 
and its space cost is at par with their sampler. 
In contrast to~\cite{Jowhari11} which can handle positive as well as negative updates and final multiplicities 
(also referred to as \emph{general turnstile stream}), 
our sampling technique works on streams with positive as well as negative updates 
provided the final multiplicity of each element is non-negative (also referred to as \emph{ strict turnstile stream}). 
This is a reasonable assumption for graph streaming algorithms, 
and it applies to simple graphs as well as to multigraphs.

For our spanner construction algorithm, we devise two samplers: \emph{FindParent} and \emph{FindNewVisitor} for unweighted graphs.
For our hopset construction algorithm we devise two more samplers: \emph{GuessDistance} and \emph{FindNewCandidate},
which are essentially weighted graph counterparts of \emph{FindParent} and \emph{FindNewVisitror}, respectively.
We will describe each of these samplers in detail in the sequel. 
The procedure \emph{FindParent} works on unweighted graphs and
 enables us to find the parent of a given input vertex in a Breadth First Search (henceforth, BFS) forest 
rooted at a subset of the vertex set $V$ of the input graph.
The procedure \emph{GuessDistance} works on weighted graphs and enables us to find 
the parent of a given vertex in a forest spanned by an \emph{approximate} 
Bellman-Ford exploration.
 It also returns the approximate distance 
of the input vertex to the set of roots of the exploration.
The procedure \emph{FindNewVisitor} helps us to implement multiple simultaneous BFS traversals, 
 each rooted at a different vertex in a subset $S$ of the vertices of an unweighted input graph. 
The procedure \emph{FindNewVisitor} enables us to sample, for a given $v \in V$, 
the root of one of the BFS explorations that $v$ belongs to.
The procedure \emph{FindNewCandidate} is a counterpart of procedure \emph{FindNewVisitor}
%
Although our samplers \emph{FindParent} and \emph{FindNewVisitor} (and their counterparts for weighted graphs)
 are used in a specific context in our algorithm, 
they can be adapted to work in general to sample elements of any type from a dynamic stream with non-negative multiplicities. 
A variant of \emph{FindParent} was described in ~\cite{GibbKKT15, KingKT15} 
in the context of dynamic and low-communication distributed graph algorithms. 
In the context of dynamic graph streams, we have adapted it to work as a sampler for sampling elements (in our case edges of a graph) 
whose multiplicity at the end of the stream is either $0$ or $1$. 
On the other hand, our second sampler, \emph{FindNewVisitor} is more general and 
to the best of our knowledge, new. 
It can sample elements with non-negative multiplicities. 
As an example, \emph{FindNewVisitor} can be adapted to sample edges from a multigraph 
in distributed, dynamic and dynamic streaming models. 

The sampler \emph{FindNewVisitor} (and also its weighted counterpart \emph{FindNewCandidate}) is based on Jarnik's construction of convexly independent sets~\cite{JarnkberDG}, 
and is related to constructions of lower bounds for distance preservers due to Coppersmith and Elkin~\cite{ElkinCSmith}.

\subsection{Hash Functions}\label{sec:hash}
Algorithms for sampling from a dynamic stream are inherently randomized and often use hash functions as a source of randomness.
Appendix~\ref{sec:hashH} is devoted to hash functions.

\subsection{Vertex Encodings}\label{sec:Encodings}
We assume that the vertices have unique IDs from the set $\{1,\dots,n\}$.
The maximum possible ID (which is $n$) of a vertex in the graph is denoted by $maxVID$.
The binary representation of the  ID of a vertex $v$ can be obtained by performing a name operation $\emph{name}(v)$.

We also need the following standard definitions of convex combination, convex hull and a convexly independent set.

\theoremstyle{definition}
\begin{definition}\label{def:conComb}
	Given a finite number of vectors $x_1, x_2,  \cdots, x_k$ in $\mathbb{R}^d$, 
	a \textbf{\emph{convex combination}} of these vectors is a vector of the form $\alpha_1  x_1 + \alpha_2 x_2+ \cdots + \alpha_k x_k$, where the real
	numbers $\alpha_i$ satisfy $\alpha_i \ge 0$ and $\alpha_1 + \alpha_2+ \cdots \alpha_k = 1$. 
\end{definition}

\theoremstyle{definition}
\begin{definition}\label{def:conHull}
	 The \textbf{\emph{convex hull}} of a set $\mathcal{X}$ of vectors in $\mathbb{R}^d$ is the set of all convex combinations of elements of  $\mathcal{X}$.
\end{definition}

\theoremstyle{definition}
\begin{definition}\label{def:CIS}
	A set of vectors ${x_1,x_2,\cdots,  x_k} \in \mathbb{R} ^d $ is called a \textbf{\emph{convexly independent set}} (\emph{CIS} henceforth), 
	if for every index $i \in [n]$, the vector $x_i$ cannot be expressed as a convex combination of the vectors $x_1, . . . , x_{i -1}, x_{i+1}, . . . , x_k$.
\end{definition}

We will use the following $\emph{CIS}$-based encoding for the vertices of the graph:\\
\textbf{CIS Encoding Scheme $\nu$}: We assign a unique code in $\mathbb{Z}^2$ to every vertex $v \in V$. 
The encoding scheme works by generating a set of $n$ convexly independent (See Definition~\ref{def:CIS}) integer vectors in $\mathbb{Z}^2$.
Specifically, our encoding scheme uses as its range, the extremal points of the convex hull (See Definition~\ref{def:conHull}) of $Ball_2(R) \cap \mathbb{Z}^2$, where $Ball_2(R)$ is a two-dimensional disc of radius $R$ centered at origin.
 A classical result by Jarn\'{i}k~\cite{JarnkberDG}, later refined by Balog and B\'{a}r\'{a}ny~\cite{AntalImre91}, states that the number of extremal points of the convex hull of a set of integer points of a disc of radius $R$ is $\Theta(R^{2/3})$. 
 We set $R = \Theta(n^{3/2})$ to allow for all the possible $n = \Theta(R^{2/3})$ vertices to be encoded in $O(\log n)$ bits. The encoding of any vertex $v$ can be obtained by performing an encoding operation denoted by $\nu(v)$.

The following lemma will be useful later in Section~\ref{interconnect} and Section~\ref{interconnectWeighted} to detect if the sampling procedure succeeded in sampling exactly one vertex from a desired subset of the set $V$.
\begin{lemma}\label{lem:convex}
Let $c_1, c_2, \cdots, c_n$ be non-negative integer coefficients of a linear combination of a set $P = \{p_1, p_2, \cdots p_n \}$ of $n$ convexly independent points in $\mathbb{Z}^2$ such that 
$\frac {\sum_{j=1}^{n} c_j \cdot p_j}{\sum_{j=1}^{n} c_j} = p_i$, for some $p_i \in P$. Then $c_j = 0$ for every $j \neq i$.
\end{lemma}
\begin{proof}
The expression $\frac{\sum_{j=1}^{n} c_j \cdot p_j}{\sum_{j=1}^{n} c_j}$ is a convex combination of points $p_1, p_2, \ldots, p_n$~, since for every $j$, we have,  $0 \le \frac{c_j}{\sum_{j=1}^{n} c_j} \le 1$ and 
$\mathlarger{\sum}_{j=1}^{n} \frac{c_j}{\sum_{j=1}^{n} c_j} = 1$. Since $P$ is a CIS, by Definition~\ref{def:CIS}, no point $p_i \in P$ can be represented as a convex combination of other points in $P$. Therefore, $c_j = 0$ for every $j \neq i$. 
\end{proof}

%% file: BFSDynamic.tex
\section{BFS Forest}\label{sec:BFSForest}
In this section, we describe an algorithm that generates a BFS forest rooted at a given set of source vertices of an input unweighted graph in dynamic streaming model.
\subsection{General Outline}\label{sec:BFSOutline}
 Given a  graph $G(V,E)$, a set of source vertices $S \subseteq V$ and a depth parameter $\eta$, 
the algorithm outputs a set of edges $E^{\eta}_S \subseteq E$ of
  non-overlapping BFS  explorations up to depth $\eta$, each rooted at a specific member of $S$. 
Initially,  $E^{\eta}_S$ is set to $\emptyset$. 
The algorithm proceeds in phases $1$ to $\eta$, 
where for each $p \in [\eta]$, we discover the edges belonging to the layer $p$ of the BFS forest in phase $p$.
The layer $p$ of the BFS forest is the set of vertices of $G$ that are at distance $p$ from $S$.

In each phase, we make one pass through the stream. 
 Let $V_p \subseteq V$ denote the set of vertices belonging to the $p^{th}$ layer of the forest. 
 The set  $V^{unc}_{p} = V \setminus \bigcup_{k \in [0,p]} V_k$ is the set of vertices that do not belong to any of the first $p$ layers. 
 The set $V_0$ is initialized to the set $S$ and the set $V^{unc}_0$ is set to $V \setminus V_0 = V \setminus S$.\\ 
 Phase $p$ starts by receiving as input, 
the sets $V_{p-1}$ and $V^{unc}_{p-1}$ computed in the previous phase. 
We invoke for each vertex $x \in V_{p-1}^{unc}$, 
a randomized procedure called \emph{FindParent} 
to sample an edge (if exists) between $x$ and some vertex $y \in V_{p-1}$. 
%

 The pseudocode for procedure $\emph{FindParent}$ is given in Algorithm~\ref{algps:findParent}.
 Its verbal description is provided right after that.

   \begin{algorithm}[h!]
 \caption {Pseudocode for Procedure $FindParent$}\label{algps:findParent}
 \begin{algorithmic}[1]
 \State $\textbf{Procedure FindParent} (x,h)$
  
    \Comment{Initialization}
     \State $\emph{slots} \leftarrow \emptyset$ 
     \Comment{An array with $\lambda$ elements indexed from $1$ to $\lambda$, where $\lambda = \ceil{\log n}.$}
     \LeftComment{Each element of slots is a tuple $(xCount, xName)$.  For a given index $1 \le k \le \lambda$, $xCount$ and $xName$
     of $\emph{slots}[k]$ can be accessed as $\emph{slots}[k].xCount$ and $\emph{slots}[k].xName$, respectively.}\\
     \LeftComment{ $\emph{slots}[k].xCount$ is number of sampled edges $(x,y)$ with $ h(y) \in[2^k]$.  It is initialized as $0$.}\\
      \LeftComment{$\emph{slots}[k].xName$ is encoding of the (binary) names of the endpoints $y$ of the sampled edges $(x,y)$ with $h(y) \in [2^k]$. It is initialized as $\phi$.}
         
   \Comment{Update Stage}
%
   \While{$(\text{there is some update~} (e_t,~eSign_t) \text{~in the stream})$}
         \If{$(e_t \text{~is incident on~} x \text{~and some~} y \in V_{p-1} )$ } 
              		 \State$ k \gets \ceil{\log h(y)}$
               	           \Repeat 
	                         \State $\emph{slots}[k].xCount \gets \emph{slots}[k].xCount + eSign_t $
               			\State $\emph{slots}[k].xName \gets \emph{slots}[k].xName \bigoplus name(y) $ 
               			 \State $k = k+1$
%
%
			 
               		   \Until{$k > \lambda$}
          \EndIf
     \EndWhile
      
 \Comment{Recovery Stage}
  	\If{$(\emph{slots} \text{~vector is empty})$}  
    		\State return $\phi$ 

  		 \ElsIf{$( \exists \text{~index~} k~|~\emph{slots}[k].xCount = 1)$} \label{alglin:singleedge}
        		\State  return $\emph{slots}[k].xName$
    		 \Else
           		\State return $\perp$ 
    \EndIf
 \end{algorithmic}
  \end{algorithm}
  
The procedure \emph{FindParent} takes as input the ID of a vertex and 
  a hash function $h$ chosen at random from a family of pairwise independent hash functions.
  A successful invocation of \emph{FindParent}  for an input vertex $x$ in phase $p$ 
  returns an edge that connects $x$ to some vertex in $V_{p-1}$, 
  if there is at least one such edge in $E$, and $\phi$ otherwise. 
  Note that \emph{FindParent} is a randomized procedure 
  and may fail to sample an edge (with a constant probability)  between $x$ and $V_{p-1}$, 
  even when such an edge exists. It returns an error $\perp$ in that case.
    
  Before we start making calls to procedure $\emph{FindParent}$, 
  we sample uniformly at random a set of functions $H_{p}$ from a family of pairwise independent hash functions 
  $h : \{1,2, \dots, maxVID \} \rightarrow \{1,\dots,2^{\lambda}\}$, where $\lambda = \ceil{\log maxVID} = \ceil{\log n}$. 
  Recall that $maxVID$ is the maximum possible vertex identity.
  The size of the set $H_p$ will be specified later in the sequel. 
  For every vertex $x \in V_{p-1}^{unc}$, we make $|H_{p}|$  parallel calls to procedure \emph{FindParent}, one for each $h \in H_p$. 
  As shown in the sequel, a single call to procedure \emph{FindParent} succeeds only with a constant probability. 
  Hence multiple parallel calls are required to boost the probability of successfully finding a parent for a given vertex. 
  The set $V_{p-1}$ computed in  phase $p-1$ is made available in the global storage 
  for all the calls to procedure $\emph{FindParent}$ in the phase $p$ to access. 
  
  In the following section, we describe in detail the concepts used to implement the procedure \emph{FindParent}.
  \subsection{Procedure FindParent}\label{sec:findParent}
 For a given vertex $x \in V^{unc}_{p-1}$, let $d^{(p-1)}_x$ be the degree of $x$ with respect to set  $V_{p-1}$. 
 In what follows, we will refer to an edge between $x$ and some $y \in V_{p-1}$ as a \emph{candidate edge}.
 A simple randomized technique to find a parent for $x$ is 
 by sampling its incident edges that connect it to the set $V_{p-1}$ 
 with probability $\frac{1}{d^{(p-1)}_x}$ (by flipping a biased coin) 
 and keeping track of all the updates to the sampled edges. 
 A given edge can appear or disappear multiple times in the stream 
 and one needs to remember the random bit  for every candidate edge
  (the result of coin flip for the edge when it appeared for the first time). 
  Remembering random bits is required in order to treat every update 
  to a given candidate edge consistently as the stream progresses. 
  This requires remembering $O(n)$ bits per vertex. 
  Instead, we use a pairwise independent hash function
  to assign hash values to the candidate edges in the range $\{1,2,\ldots,2^{\lambda} \}$, where 
 $\lambda = \ceil{\log maxVID}$. 
  If we knew the exact value of $d^{(p-1)}_x$, 
  we could sample every new candidate edge witnessed by $x$  with probability $1/d^{(p-1)}_x$ to extract exactly one of them in expectation. 
  However, all we know about $d^{(p-1)}_x$ is that it is at most $n$.  
  We therefore sample every new candidate edge on a range of probabilities.
We use an array $\emph{slots}$ of $\lambda$ elements (the structure of each element will be described later in the sequel) 
indexed by \emph{slot-levels} from $1$ to $\lambda = \ceil{\log n}$ to implement sampling on a range of probabilities. 
 We want a given candidate edge $(x,y)$ to be sampled into slot-level $k$ with probability $1/2^{\lambda - k}$. 
 When $d^{(p-1)}_x \approx 2^{\lambda - k}$, with a constant probability there is exactly one candidate edge that gets mapped to $\emph{slots}[k]$.
 Every new candidate edge $e = (x, y)$ witnessed by $x$ with $ y \in V_{p-1}$ is assigned a hash value $h(y)$ by $h$. 
  A given edge $e = (x,y)$ gets mapped into $slots[k]$, if $h(y) \in [2^k]$. 
  Note that a given candidate edge may be assigned to multiple slot-levels. 
  
  In every element of $\emph{slots}$, we maintain a tuple $(xCount, xName)$, and  $xCount$ and $xName$
     of $\emph{slots}[k]$ can be accessed as $\emph{slots}[k].xCount$ and $\emph{slots}[k].xName$, respectively.\\
  The field $xCount \in \mathbb{Z}$ at slot-level $k$ maintains the number of candidate edges with hash values in $[2^k]$. 
  It is initialized to $0$ at the start of the stream. 
  Every time an update to a candidate edge $e = (x,y)$ with $h(y) \in [2^k]$ appears on the stream, 
  $\emph{slots}[k].xCount$ is updated by adding the $eSign$ value of $e$ to its current value.
   The final value of the $xCount$ field is thus given by the following expression:
  
 $$\emph{slots}[k].xCount = \sum_{(e_t,~eSign_t) \mid e_t = (x,y) \text{~for some~} y \in V_{p-1} \text{~and~}  h(y) \in [2^k] } eSign_t$$
  
 The field $xName$ at slot-level $k$ is a bit string which maintains the bitwise XOR of the \emph{binary} names of all the candidate edges sampled at slot-level $k$. 
 It is initalized as an empty string at the start of the stream. 
 Every time an update to a candidate edge $e = (x,y)$ with $h(y) \in [2^k]$, $y \in V_{p-1}$, appears on the stream, 
 $\emph{slots}[k].xName$ is updated by performing a bitwise XOR of its current value with $\emph{name}(y)$.  
 The final value of the $xName$ field is thus given by the following expression: 
$$\emph{slots}[k].xName= \bigoplus_{(e_t,~eSign_t) \mid e_t = (x,y) \text{~for some~} y \in V_{p-1} \text{~and~}  h(y) \in [2^k]} name(y)$$
 
 At the end of the stream, if the \emph{slots} array is empty, then there are no edges incident on $x$ that  connect it to the set $V_{p-1}$ 
 and the \emph{FindParent} procedure returns $\phi$. 
 (Note that $\emph{slots}[\lambda]$ is an encoding of all the candidate edges incident on $x$.)
  If there is a slot-level $k$ such that $\emph{slots}[k].xCount =1$, 
  then only one candidate edge is mapped to slot-level $k$ 
  and $\emph{slots}[k].xName$ gives us the name of the other endpoint of this edge.
  The procedure \emph{FindParent}  returns $\emph{slots}[k].xName$ as a parent of $x$. 
  If the $\emph{slots}$ array is not empty but there is no slot level with its $xCount = 1$, 
  then the procedure \emph{FindParent} has failed to find a parent for $x$ and returns an error $\perp$.
  
If the input vertex $x$ has a non-zero degree with respect to the set $V_{p-1}$, 
we need to make sure that for some $1 \le k \le \lambda$, only one candidate edge will get mapped to $slots[k]$. 
By Corollary~\ref{cor:pairwise}, only one of the $d^{(p-1)}_x$ candidate edge 
gets mapped to the set $[2^k]$, for $k = \lambda - \ceil{\log d^{(p-1)}_x } -1$, with at least a constant probability. 
Therefore, a single invocation of \emph{FindParent} succeeds with at least a constant probability. 
Since we are running $|H_{p}|$ parallel invocations of $FindParent$, 
we pick the output of a successful invocation of procedure $\emph{FindParent}$ 
as the parent. (See Section~\ref{sec:BFSOutline}; $H_p$ is a set of  randomly sampled hash functions.)
If multiple invocations are successful, we use the output of one of them arbitrarily. 
In the case that all the invocations of $\emph{FindParent}$ return an error, the algorithm terminates with an error. 
In the sequel we show that when the set $H_p$ is appropriately sized, 
the event of all the invocations of procedure $\emph{FindParent}$ for a given vertex failing has very low probability.

At the end of phase $p$,  if the algorithm has not terminated with an error, 
every vertex $x \in V_{p-1}^{unc}$ for which we have sampled an edge to the set $V_{p-1}$, is added to the set $V_p$. 
Every sampled edge is added to the set $E^{\eta}_S$. 
The set $V^{unc}_{p}$ is updated as $V^{unc}_{p} = V^{unc}_{p-1}\setminus V_p$.

\begin{lemma} \label{lem:findPProb}
For $|H_p| = c_1\log_{8/7} n$ for some $c_1 \ge1$,  
at least one of the $|H_p|$ invocations of procedure \emph{FindParent} 
for a given vertex in phase $p$ succeeds with probability at least $ 1- \frac{1}{ n^{c_1} }$.
\end{lemma}

\begin{proof}
The procedure \emph{FindParent} relies on the ability of the random pairwise hash function to hash exactly one edge in the target range of $[ 2^ {\lambda - \ceil*{\log d^{(p-1)}_x} - 1} ]$. 
By Corollary~\ref{cor:pairwise}, this happens with at least a constant probability of $1/8$. 
If we invoke procedure \emph{FindParent} $c_1\log_{8/7} n$ times in parallel using independently chosen at random hash functions, 
then all of them fail with a probability  at most $(7/8)^{c_1\log_{8/7} n} = \frac{1}{n^{c_1}}$. 
Therefore, at least one of the $|H_p|$ invocations succeeds with probability at least $ 1- \frac{1}{ n^{c_1} }$. 
\end{proof}

Next, we analyze the space requirements of procedure \emph{FindParent}.

\begin{lemma}\label{lem:findPSpace}
The procedure \emph{FindParent} uses $O(\log^2n)$ bits of memory.
\end{lemma}
\begin{proof}
The input to this procedure is the ID of a vertex $x$ and a pairwise independent hash function $h$. 
This consumes $O(\log n)$ bits. The procedure also needs access to the set of vertices $V_{p-1}$ of the previous layer. 
We will not charge this procedure for the space required for storing $V_{p-1}$, 
since it is output by the phase $p -1$ and is passed on to phase $p$ as an input. 
We instead charge  phase $p-1$ globally for its storage. 
Similarly, we do not charge each invocation of \emph{FindParent} in phase $p$ for the storage of the hash function $h$. 
Rather it is charged to phase $p$ globally. 
Inside the procedure, the \emph{slots} vector is an array of length $\lambda$  and $\lambda =O(\log n)$. 
Every element of $\emph{slots}$ stores two variables $xCount$ and  $xName$ each of which consumes $O(\log n)$ bits. 
Thus the overall space required by this procedure  is $O(\log ^2 n)$ bits.
\end{proof}

We now proceed to analyzing the space requirements of the entire algorithm.
\begin{lemma}\label{lem: space}
In each of the $\eta$ phases, our BFS forest construction algorithm uses $O(n \log^3 n)$ memory.
\end{lemma}

\begin{proof}

In any phase $ p \ge 1 $, we try to find a parent for every vertex in the set $V^{unc}_{p-1}$. 
This requires making multiple simultaneous calls to procedure $\emph{FindParent}$. 
By Lemma~\ref{lem:findPProb}, we need to make  $O(\log n)$ parallel calls to procedure $\emph{FindParent}$ per vertex. 
For this we  sample $O(\log n)$ pairwise independent hash functions.
 Every single pairwise independent hash function requires $O(\log n)$ bits of storage (Lemma~\ref{lem:2wiseSpace}) 
and thus the set $H_p$ requires $O(\log^2 n)$ bits of storage. 
By Lemma~\ref{lem:findPSpace}, a single call to procedure $\emph{FindParent}$ uses $O(\log^2 n)$ bits. 
Thus making $O(\log n)$ parallel calls (by Lemma~\ref{lem:findPProb}) needs $O(\log ^3 n)$ bits per vertex.
The set $V^{unc}_{p-1}$ has size $O(n)$.  
Thus the overall cost of all the calls to procedure $\emph{FindParent}$ is $O(n \log^3 n)$. 
As an output, phase $p$ generates the set $V_p$ and  the set of edges belonging to the layer $p$ of the BFS 
which is then added to the final output set  $E_{S}^{\eta}$. 
Both these sets are of size $O(n)$ and each element of these sets requires $O(\log n)$ bits. 
Thus the cost of maintaining the output of  phase $p$ is bounded by $O(n \log n)$ bits.
Hence the overall storage cost of phase $p$ is dominated by the calls to procedure \emph{FindParent}. 
The overall storage cost of any phase is therefore $O(n\log^3 n)$ bits. 
\end{proof}

In the following lemma, we provide an inductive proof of the correctness of our algorithm.
Recall that $|H_p| = c_1 \log_{8/7} n$, where, $c_1 > 0$ is a positive constant.

\begin{lemma}\label{lem:BFSCorrect}
After $p$ phases of the algorithm described in Section~\ref{sec:BFSOutline}, 
the algorithm has constructed a BFS forest to depth $p$ rooted at $S \subseteq V$ 
with probability at least $1 - p/n^{c_1-1}$.
\end{lemma}
\begin{proof}
The proof follows by induction on the number of  phases, $p$, of the algorithm. 
The base case for $p=0$ holds trivially.
For the inductive step, we assume that after $k$ phases of our algorithm, the set of output edges $E^{\eta}_S$ 
forms a BFS forest to depth $k$ with probability at least $1 - k/ n^{c_1-1}$. 
This implies that all the vertices within distance $k$ from $S$ have found a parent in the BFS forest with probability at least $ 1 - k/ n^{c_1-1}$.
 In phase $k + 1$, we make $|H_{k+ 1}|$ parallel calls to procedure \emph{FindParent} for every vertex not yet in the forest. 
 For all the vertices at a distance more than $k + 1$ from the set $S$, all the calls to procedure \emph{FindParent} return $\phi$ in phase $k + 1$. 
 Let $x$ be a vertex at distance $k + 1$ from the set $S$. 
 By Lemma~\ref{lem:findPProb}, at least one of the $|H_{k+ 1}|$ independent calls to procedure $\emph{FindParent}$ made for $x$ in phase $k + 1$ 
 succeeds in finding a parent for $x$ with probability at least $1 - \frac{1}{n^{c_1}}$. 
 Since there can be at most $O(n)$ vertices at distance $k + 1$ from set $S$, 
 by union bound, phase $k + 1$ fails to find a parent for one of these vertices with probability at most $1/n^{c_1-1}$. 
 Taking a union bound over the failure probability of first $k$ phases from induction hypothesis with the failure probability of phase $k + 1$, 
 we get that all the vertices within distance $k+ 1 $ from the set $S$ successfully add their parent edges in the BFS forest to the output set $E^{\eta}_S$ with probability at least $1 - (k+1)/n^{c_1-1}$.
\end{proof}

Lemmas~\ref{lem: space} and~\ref{lem:BFSCorrect} imply the following theorem:

\begin{theorem}\label{thm: SingleBFS}
For a sufficiently large positive constant $c$, 
given a depth parameter $\eta$, an input graph $G(V,E)$, and a subset $S \subseteq V$,  
the algorithm described in Section~\ref{sec:BFSOutline} generates with probability at least $1 -\frac{1}{n^{c}}$, 
a BFS forest of $G$ of depth $\eta$ rooted at vertices in the set $S$ in $\eta$ passes 
through the dynamic stream using  $O_c(n \log^3 n)$ space in every pass.
\end{theorem}

Note also that the space used by the algorithm on different passes can be reused, i.e., the total space used by the algorithm is $O_c(n \log^3 n)$.

%% file: WeightedBFS.tex
\section{Approximate Bellman-Ford Explorations}\label{sec:WeightedForest}
In this section, we describe an algorithm for performing a given number of iterations 
of an approximate Bellman-Ford exploration from a given
 subset $S \subseteq V$ of \emph{source} vertices in a weighted undirected graph $G(V, E, \omega)$ 
 with aspect ratio $\Lambda$.
We assume throughout that the edge weights are positive numbers between $1$ and $maxW$.
Note that $\Lambda \le (n-1)\cdot maxW$.
Recall that for a pair $u, v \in V$ of distinct vertices and an integer $t \ge 0$, 
the $t$-bounded distance between $u$ and $v$ in $G$, denoted  $d^{(t)}_{G}(u,v)$,
is the length of a shortest $t$-bounded $u$-$v$ path in $G$. (See Definitions~\ref{def:tLimitedPath} and~\ref{def:tLimitedDist}.)
For a given vertex $v \in V$ and a set $S \subseteq V$, 
the $t$-bounded distance between $v$ and $S$ in $G$,  denoted $d^{(t)}_{G}(v, S)$,
is the length of a shortest $t$-bounded path between $v$ and some $s \in S$ 
such that $d^{(t)}_{G}(v,s) = \min\{ d^{(t)}_{G}(s',v)~|~s' \in S \}$.
%

\subsection{Algorithm}\label{sec:algo}
Given an $n$-vertex weighted graph $G(V, E, \omega)$, a set $S \subseteq V$ of vertices,  
an integer parameter $\eta > 0$ and an error parameter $\zeta \ge 0$, 
an $(\eta, \zeta)$-Bellman-Ford exploration (henceforth, BFE) of $G$ 
rooted at $S$ outputs for every vertex $v \in V$, a $(1+ \zeta)$-approximation of its $\eta$-bounded distance to
 to the set $S$. 
 Throughout the execution of our algorithm, we maintain two variables for each vertex $v \in V$.
One of them is a current estimate of $v$'s $\eta$-bounded distance to set $S$, 
denoted  $\hat{d}(v)$, and
the other is the ID of $v$'s neighbour through which it gets its current estimate, denoted $\hat{p}(v)$,
and called the \emph{parent} of $v$.

We start by initializing $\hat{d}(s) = 0$, $\hat{p}(s) = \perp$, for each $s \in S$ and 
$\hat{d}(v) = \infty$, $\hat{d}(v) = \perp$ for each $v \in V \setminus S$.
 As the algorithm proceeds, $\hat{d}(v)$ and $\hat{p}(v)$ values of every vertex $v \in V \setminus S$ 
 are updated to reflect the current best estimate of $v$'s $\eta$-bounded distance to the set $S$.
The final value of $\hat{d}(v)$ for each $v \in V$ is such that 
$d^{(t)}_{G}(v, S) \le \hat{d}(v) \le (1 + \zeta) \cdot d^{(t)}_{G}(v, S)$, 
and the final value of $\hat{p}(v)$ for each $v \in V$ 
contains the ID of $v$'s parent on the forest spanned
by $(\eta, \zeta)$-BFE of $G$ rooted at the set $S$.

The algorithm proceeds in phases, indexed by $p$, $1 \le p \le \eta$.
We make one pass through the stream in each phase.

%

\textbf{Phase $p$:}
In every phase, we search for every vertex $v \in V \setminus S$, 
a \emph{better} (smaller than the current value of $\hat{d}(v)$) estimate (if exists) 
of its $\eta$-bounded distance to the set $S$, by keeping track of updates to edges $e = (v, u)$ incident to $v$. 
Specifically, we divide the search space of potential better estimates, $\left[1, 2\cdot\Lambda \right]$, into sub-ranges
$I_j = \osubrange{\zeta'}{j}$, for $j \in \{0,1,\ldots, \gamma \}$, where $\gamma = \ceil{\log_{1 + \zeta'} 2\cdot \Lambda } -1$ 
and  $\zeta'$ is set to $\zeta/2\eta$ for technical reasons 
to be expounded later in the sequel. For $j=0$, we make the sub-range $I_0 = \left[(1+ \zeta')^{0}, (1+ \zeta')^{1} \right]$
closed to include the value $1$.
Recall that we are doing a $(1+\zeta)$-approximate Bellman-Ford exploration (and not an exact one).
 Due to this, some of the better estimates we get in a given phase may be between $\Lambda$ and $(1+ \zeta) \cdot \Lambda \le 2\cdot \Lambda$,
where $\Lambda$ is the aspect ratio of the input graph.
We therefore keep our search space from $1$ to $2\Lambda$ instead of $\Lambda$.
 
In more detail, we make for  for each $v \in V \setminus S$, 
$\gamma$ \emph{guesses}, one for each sub-range. 
In a specific guess for a vertex $v$ corresponding to sub-range \osubrange{\zeta'}{j} for some $j$, 
we make multiple simultaneous calls to a randomized procedure called 
\emph{GuessDistance} which samples an edge (if exists) between $v$ and some vertex $u$ such that
\begin{align*}
 \hat{d}(u) + \omega(v, u) \in I_j.
 \end{align*}
 
The exact number of calls we make to procedure \emph{GuessDistance} in each guess will be specified later in the sequel.
 

 The smallest index $j \in [0, \gamma ]$, for which the corresponding guess denoted $Guess^{(j)}_v$ successfully
samples an edge which gives a distance estimate better than the current estimate of $v$, 
is chosen to update  $\hat{d}(v)$.

 The pseudocode for procedure $\emph{GuessDistance}$ is given in Algorithm~\ref{algps:guessDistance}.
 Its verbal description is provided right after that.
   \begin{algorithm}[h!]
 \caption {Pseudocode for Procedure $GuessDistance$}\label{algps:guessDistance}
 \begin{algorithmic}[1]
 \State $\textbf{Procedure GuessDistance} (x,h, I)$
  
    \Comment{Initialization}
     \State $\emph{slots} \leftarrow \emptyset$ 
     \Comment{An array with $\lambda$ elements indexed from $1$ to $\lambda$, where $\lambda = \ceil{\log n}$. }
     \LeftComment{Each element of slots is a tuple $(xCount, xDist, xName)$.  For a given index $1 \le k \le \lambda$, fields $xCount$, $xDist$ and $xName$
     of $\emph{slots}[k]$ can be accessed as $\emph{slots}[k].xCount$, $\emph{slots}[k].xDist$ and $\emph{slots}[k].xName$, respectively.}\\ 
     \LeftComment{ $\emph{slots}[k].xCount$ is the number of sampled edges $(x,y)$ with $h(y) \in[2^k]$. Initially, it is set to $0$.}
      \LeftComment{$\emph{slots}[k].xDist$ is the distance estimate for $x$ provided by an edge $(x,y)$ with $h(y) \in [2^k]$. initially, it is set to $0$.} 
       \LeftComment{$\emph{slots}[k].xName$ is encoding of the names of the endpoints $y$ of sampled edges $(x,y)$ with $h(y) \in [2^k]$. Initially, it is set to $\phi$. }

   \Comment{Update Stage}

   \While{$(\text{there is some update~} (e_t,~eSign_t,~eWeight_t) \text{~in the stream})$}
         \If{$(e_t \text{~is incident on~} x \text{~and some~} y \text{~such that~} \hat{d}(y) + eWeight_t \in I)$ } \label{alglin:check}
              		 \State$ k \gets \ceil{\log h(y)}$
               	           \Repeat 
	                         \State $\emph{slots}[k].xCount \gets \emph{slots}[k].xCount + eSign_t $
               			\State $\emph{slots}[k].xDist \gets \emph{slots}[k].xDist + (\hat{d}(y) + eWeight_t)\cdot eSign_t $ 
			         \State  $\emph{slots}[k].xName \gets \emph{slots}[k].xName \bigoplus name(y)$

%
%
          
           \State $k = k+1$
               		   \Until{$k > \lambda$}
          \EndIf
     \EndWhile
           
          \Comment{Recovery Stage}
 	         \If{$(\emph{slots} \text{~array is empty})$}  
    		       \State return $(\infty, \infty)$ 
 		        \ElsIf{$( \exists \text{~index~} k~|~\emph{slots}[k].xCount = 1)$} \label{alglin:singleedge}
          		         \State  return $(\emph{slots}[k].xDist, \emph{slots}[k].xName)$
    		         \Else
           		     \State return $(\perp, \perp)$ 
  		\EndIf
   \end{algorithmic}
     \end{algorithm}
  
The procedure \emph{GuessDistance} can be viewed as an adaptation of procedure \emph{FindParent} from Section~\ref{sec:findParent} for
 weighted graphs. It enables us to find an estimate of $\eta$-bounded distance of an input vertex $x$ to the set $S$ in a given range of distances. 
 It takes as input the ID of a vertex, a hash function $h$ chosen at random from a family of pairwise independent hash functions 
 and an input range $I = (low, high]$. (The input range may be closed as well.)
 A successful invocation of procedure \emph{GuessDistance} for an input vertex $x$ and input range $I$, returns
 a tuple $(dist, parent)$, (if there is at least one edge $(x, y)$ in $G$ such that $\hat{d}(y) + \omega(x,y) \in I$, and $\phi$ otherwise),
 where $dist$ is an estimate of $x$'s  $\eta$-bounded distance to the set $S$ in the range $I$,
 and $parent$ is the $parent$ of $x$ in the forest spanned by $(\eta, \zeta)$-BFE of $G$ rooted at the set $S$.
 
 The procedure  \emph{GuessDistance} may fail to return (with a constant probability) 
 a distance estimate in the desired range, 
 even when such an estimate exists. 
 It returns an error, denoted by $(\perp,\perp)$, in that case.

 As we did for procedure \emph{FindParent} in Section~\ref{sec:BFSForest}, 
 before we start making calls to procedure \emph{GuessDistance}, we sample uniformly at random a set of functions $H_{p}$ of size $c_1\log_{8/7} n$ 
 from a family of pairwise independent hash functions $h : \{1, \dots, maxVID\} \rightarrow \{1,\dots,2^{\lambda}\}$, where $\lambda = \ceil{\log n}$
 and $c_1$ is an appropriate constant.
For every guess for a given vertex $x \in V \setminus S$ and a given subrange $I_j$,
we make $|H_{p}|$  parallel calls to procedure \emph{GuessDistance}, one for each $h \in H_p$, 
to get an estimate of $d^{(\eta)}_G(x, S)$ in the given subrange.
The multiple parallel calls are required since a single call to procedure \emph{GuessDistance}  succeeds only with a constant probability,
while we need to succeed with high probability.

Additionally, before we start the phase $p$, we create for each $v \in V \setminus S$, a copy $\hat{d'}(v)$ of its current distance estimate $\hat{d}(v)$.
Any update to the distance estimate of a vertex $v$ during phase $p$ is made to its \emph{shadow} distance estimate $\hat{d'}(v)$.
On the other hand, the variable $\hat{d}(v)$ for vertex $v \in V \setminus S$ remains unchanged during the execution of phase $p$. 
At the end of phase $p$, we update $\hat{d}(v)$ as $\hat{d}(v) =\hat{d'}(v)$.
The purpose of using the shadow variable is to 
avoid any issues arising due to simultaneous reading from and writing 
to the distance estimate variable of a vertex by multiple parallel calls to procedure \emph{GuessDistance}. 
  
 \subsection{Procedure GuessDistance}\label{sec:GuessDistance}
The overall structure and technique of procedure \emph{GuessDistance} is similar to that of procedure \emph{FindParent}.~(See Section~\ref{sec:findParent}.)
For a given vertex $x$, and a given distance range $I$, let $y \in \Gamma_G(x)$ be such that
\begin{align}\label{eq:Candidate}
 \hatter{y} + \omega(x,y) \in  I
  \end{align}
 In what follows, we will refer to a vertex $y \in \Gamma_G(x)$ for which Equation~\ref{eq:Candidate} holds as a \emph{candidate neighbour}
and the corresponding edge $(x,y)$ as a \emph{candidate edge} in the range $I$.
For a given vertex $x$, let  $c^{(p, j)}_x$ be the number of candidate neighbours of $x$ in the sub-range $I_j$.
A call to procedure \emph{GuessDistance} for vertex $x$ with input range $I = I_j$ works by sampling a \emph{candidate} neighbour
with probability $\frac{1}{c^{(p,j)}_x}$. 
As described in Section~\ref{sec:findParent}, one of the ways to
sample with a given probability in a dynamic streaming setting is
to use hash functions. 
We therefore use a pairwise independent hash function as in Section~\ref{sec:findParent} 
to assign hash values to the candidate edges in the range $\{1,\ldots,2^{\lambda} \}$, where 
 $\lambda = \ceil{\log n}$. 
 As in the case of \emph{FindParent}, we only know an upper bound of $n$ and not the exact value of $c^{(p,j)}_x$. 
 Therefore, we try to guess $c^{(p,j)}_x$ on a geometric scale
 of values $2^{\lambda -k}$, $k = 1,2, \ldots, \lambda$,
  and sample every candidate neighbour on a range of probabilities corresponding to our guesses of $c^{(p,j)}_x$.
 To implement sampling on a range of probabilities, we use an array $\emph{slots}$ of $\lambda$ elements 
 indexed by \emph{slot-levels} from $1$ to $\lambda$.
 Every new candidate neighbour $y$ witnessed by $x$ is assigned a hash value $h(y)$ by $h$.


  
 
In every element of $\emph{slots}$, we maintain a tuple $(xCount, xDist, xName)$, and  $xCount$, $xDist$ and $xName$
of $\emph{slots}[k]$ can be accessed as $\emph{slots}[k].xCount$, $\emph{slots}[k].xDist$ and $\emph{slots}[k].xName$, respectively.\\
The variable $xCount \in \mathbb{Z}$ at slot-level $k$ maintains the number of candidate neighbours with hash values in $[2^k]$. 
It is initialized to $0$ at the beginning of the stream. 
Every time an update to a candidate edge $e_t = (x,y)$ with $h(y) \in [2^k]$ appears on the stream, $\emph{slots}[k].xCount$ is updated by adding
 the $eSign_t$ value of $e_t$ to its current value.
 The variable $xDist$ at slot-level $k$ is an estimate of $\eta$-bounded distance of $x$ limited to the input distance range $I$
  provided by edge $(x,y)$ with $h(y) \in [2^k]$. 
Initially, it is set to $0$.
Every time an update to a candidate edge $e_t = (x,y)$ with $h(y) \in [2^k]$ appears on the stream, $\emph{slots}[k].xDist$ 
is updated by adding the value of the expression $(\hat{d}(y) + eWeight_t)\cdot eSign_t $ to its current value. (Recall that it is initialized as $0$.)
The variable $xName$ is encoding of the names of endpoints $y$ of the sampled edges $(x,y)$ with $h(y) \in [2^k]$. It is set to $\phi$ initially.
Every time an update to a candidate edge $e_t = (x,y)$ with $h(y) \in [2^k]$ appears on the stream,
$\emph{slots}[k].xName$ 
 is updated by performing a bitwise XOR of its current value with $\emph{name}(y)$.

 At the end of the stream, if the \emph{slots} array is empty, then there are no candidate neighbours in $\Gamma_G(x)$
and the procedure \emph{GuessDistance} returns $(\phi, \phi)$. 
If there is a slot-level $k$ such that $\emph{slots}[k].xCount =1$, then only one candidate neighbour is mapped to slot-level $k$.
In this case, $\emph{slots}[k].xDist$ gives us an estimate of $x$'s $\eta$-bounded distance to the set $S$ in the input distance range $I$,
and  $\emph{slots}[k].xName$ gives us the name of $x$'s parent on the forest spanned by the $(\eta,\zeta)$-BFE of $G$ rooted at set $S$.
Indeed, if no smaller scale estimate will be discovered, the vertex recorded in $\emph{slots}[k].xName$ will become the parent of $x$
in the forest.
The procedure \emph{GuessDistance}  returns $(\emph{slots}[k].xDist, \emph{slots}[k].xName)$. 
If the $\emph{slots}$ vector is not empty but there is no slot level with $xCount = 1$, 
then the procedure \emph{GuessDistance} has failed to find a distance estimate in the input range $I$ for $x$, and thus it returns an error $(\perp, \perp)$.
  
If the input vertex $x$ has some candidate neighbours in the input distance range, 
we need to make sure that for some $1 \le k \le \lambda$, only one candidate neighbour will get mapped to $slots[k]$. 
By Corollary~\ref{cor:pairwise}, only one of the $c^{(p, j)}_x$ candidate neighbours gets mapped to the set $[2^k]$, for $k = \lambda - \ceil{\log c^{(p,j)}_x } -1$, 
with at least a constant probability.  
Therefore, a single invocation of procedure \emph{GuessDistance} for a given vertex $x$ and a given distance range succeeds with at least a constant probability. 
Since we are running $|H_{p}|$ parallel invocations of procedure \emph{GuessDistance} 
for a given input vertex $x$ and a given distance range $I$, 
we pick the output of a successful invocation of procedure \emph{GuessDistance} as an estimate for $x$ in the input range.
If multiple invocations in a guess are successful, we use the output of the one with the smallest return value.
In the case that all the invocations of \emph{GuessDistance} in a guess return an error, the algorithm terminates with an error. 
In the sequel we show that when the set $H_p$ is appropriately sized, 
the event of all the invocations of procedure \emph{GuessDistance} in a given guess
 failing has a very low probability.
 
 Once all the $\gamma = O(\frac{\log \Lambda}{\zeta'})$ guesses for a given vertex $x$ have completed their execution without failure,
 we pick the smallest index $j$ for which the corresponding guess $guess^{(j)}_x$ has returned a finite (non-failure) value,
 and compare this value with $\hat{d}(x)$. If this value gives a better estimate than the current value of $\hatter{x}$,
 we update the corresponding shadow variable $\hat{d'}(x)$, and the parent variable $\hat{p}(x)$.
 
 At the end of phase $p$,  if the algorithm has not terminated with an error, 
for every vertex $x \in V\setminus S$, 
we update its current distance estimate variable with 
the value in the corresponding shadow variable as $\hat{d}(x) = \hat{d'}(x)$.

In the following lemma, we analyze the success probability of guessing the $\eta$-bounded distance of a specific vertex in a given 
distance range in a specific phase $p$.

\begin{lemma} \label{lem:guessDProb}
For $|H_p| = c_1\log_{8/7} n$ for some $c_1 \ge1$,  
at least one of the $|H_p|$ invocations of procedure \emph{GuessDistance} 
in a given guess for a vertex $x$, and distance sub-range $I_j = \osubrange{\zeta'}{j}$ for some $j$, 
in a specific phase $p$ succeeds with probability at least $ 1- \frac{1}{ n^{c_1} }$.
\end{lemma}

\begin{proof}
The procedure \emph{GuessDistance} relies on the ability of the random pairwise independent
 hash function to hash exactly one edge in the target range of $[ 2^ {\lambda - \ceil*{\log c^{(p,j)}_x} - 1} ]$. 
By Corollary~\ref{cor:pairwise}, this happens with at least a constant probability of $1/8$. 
If we invoke procedure \emph{GuessDistance} $c_1\log_{8/7} n$ times in parallel using independently chosen at random hash functions, 
then all of them fail with a probability  at most $(7/8)^{c_1\log_{8/7} n} = \frac{1}{n^{c_1}}$. 
Therefore, at least one of the $|H_p|$ invocations succeeds with probability at least $ 1- \frac{1}{ n^{c_1} }$. 
\end{proof} 

\pagebreak

Next, we analyze the space requirements of procedure \emph{GuessDistance}.

\begin{lemma}\label{lem:guessDSpace}
The procedure \emph{GuessDistance} uses  $O(\log n (\log n +  \log \Lambda))$  bits of memory.
\end{lemma}
\begin{proof}
The input to this procedure is the ID of a vertex $x$, a pairwise independent hash function $h$ 
and variables $low$ and $high$, that define the input range $I$.
The ID of the vertex and the representation of the hash function $h$ consume $O(\log n)$ bits. 
The variables $low$ and $high$ correspond to
distances in the input graph and are upper bounded by the aspect ratio $\Lambda$ of the graph. 
Therefore both these variables consume $O(\log \Lambda)$ bits each.
 We do not charge each invocation of \emph{GuessDistance} in phase $p$ for the storage of the hash function $h$. 
Rather it is charged to phase $p$ globally. 
Inside the procedure, the \emph{slots} vector is an array of length $\lambda$  and $\lambda =O(\log n)$. 
Every element of $\emph{slots}$ stores three variables $xCount$, $xDist$ and $xName$. 
The variables $xCount$ and $xName$ consume $O(\log n)$ bits.
The variable $xDist$ is a distance estimate and thus consumes $O(\log \Lambda)$ bits.
Thus the overall space required by this procedure  is $O(\log n (\log n +  \log \Lambda))$  bits.
\end{proof}

We now proceed to analyzing the space requirements of the entire algorithm.
\begin{lemma}\label{lem: wSpace}
In each of the $\eta$ phases, our approximate Bellman-Ford exploration algorithm 
uses $O( n \cdot \log^2 n \frac{\log \Lambda}{\zeta'} (\log n + \log \Lambda ) )$ bits of memory.
\end{lemma}

\begin{proof}
In any phase $ p \ge 1$, we search for a possible better estimate (if exists) of 
$d^{\eta}_G(v,S)$ for every vertex $v \in V \setminus S$.
This requires making $\gamma = \ceil{\log_{(1 +\zeta')} 2\cdot \Lambda} -1$ guesses. 
Each guess in turn makes $|H_p| = c_1 \log_{8/7} n$ 
simultaneous calls to procedure \emph{GuessDistance}.
Therefore, in total, we make $O(\log_{1 + \zeta'} \Lambda \cdot \log_{8/7} n)$ parallel calls 
to procedure \emph{GuessDistance} for each $v \in V \setminus S$.
By Lemma~\ref{lem:guessDSpace}, a single call to procedure $\emph{GuessDistance}$ 
uses  $O(\log n (\log n +  \log \Lambda))$  bits. 
Thus making $O(\log_{1 + \zeta'} \Lambda \cdot \log_{8/7} n)$ parallel calls needs
$O(\log^2 n \log_{1+\zeta'} \Lambda (\log n + \log \Lambda ))$ bits per vertex.

We sample $O(\log_{8/7}n)$ pairwise independent hash functions.
 Every single pairwise independent hash function requires $O(\log n)$ bits of storage (Lemma~\ref{lem:2wiseSpace}) 
and thus the set $H_p$ requires $O(\log^2 n)$ bits of storage. 
We also store three variables $\hatter{v}$, $\hat{d'}(v)$ and $\hat{p}(v)$
for every vertex $v \in V \setminus S$.
Each of the distance variables $\hatter{v}$ and $\hat{d'}(v)$ uses $O(\log \Lambda)$ bits,
making the overall cost of their storage $O(n \log \Lambda)$.
Each of the parent variables $\hat{p}(v)$ uses $O(\log n)$ bits, making the overall cost of their storage $O(n \log n)$.
Hence the overall storage cost of phase $p$ is dominated by the calls to procedure \emph{GuessDistance}. 
The overall storage cost of any phase is therefore $O(n \cdot \log^2 n \cdot \log_{1+\zeta'} \Lambda (\log n + \log \Lambda ))$ bits. 
\end{proof}

Observe that the space used in one phase can be reused in the next phase, and this bound is the total space complexity of the algorithm.

In the following lemma, we provide an inductive proof of the correctness of our algorithm.
Recall that $|H_p| = c_1 \log_{8/7} n$, where $c_1 > 0$ is a positive constant, and that
$\zeta' = \zeta/2\eta$.

\begin{lemma}\label{lem:BFECorrect}
After $p$ phases of our approximate Bellman-Ford exploration algorithm, 
the following holds
 for every vertex $v$ within $p$ hops from the set $S$ of source vertices:
 
$$d^{(p)}_G(v,S) \le \hatter{v} \le (1 + \zeta')^p \cdot d^{(p)}_G(v,S),$$
with probability at least $1 - p/n^{c_1-1}$.
(The left-hand inequality holds with probability $1$, and the right-hand inequality holds with probability 
  at least $1 - p/n^{c_1-1}$.)
\end{lemma}
\begin{proof}
The proof follows by induction on the number of  phases, $p$, of the algorithm. 
The base case for $p=0$ holds trivially.
For the inductive step, we assume that after $k$ phases of our algorithm, 
with probability at least $1 - k/ n^{c_1-1}$,
the following holds:
For every vertex $v$ within $k$ hops from the set $S$,
$$d^{(k)}_G(v,S) \le \hatter{v} \le (1 + \zeta')^k  \cdot d^{(k)}_G(v,S).$$
In phase $k+1$, we make $\gamma$ guesses of a new (better) estimate
 for every $v \in V \setminus S$.
We then update the current estimate $\hatter{v}$ of $v$ with
the smallest guessed value which is better (if any) than the current estimate.
Denote by $u \in \Gamma_G(v)$ the neighbour of $v$ 
on a shortest $(k+1)$-bounded path from $v$ to the set $S$.
By inductive hypothesis, 
with probability at least $1- k/n^{c_1-1}$, all $k$-bounded 
estimates provide stretch at most $(1+ \zeta')^k$.
In particular, 
 $d^{(k)}_G(u,S) \le\hatter{u} \le (1+ \zeta')^k \cdot d^{(k)}_G(u,S)$. 
 Denote by $j = j_v$, the index of a sub-range such that
 $$\hatter{u} + \omega(u,v) \in I_{j}.$$
 During the execution of the $j^{th}$ guess for vertex $v$ in phase $k+1$,
we sample a candidate neighbour $u' \in \Gamma_G(v)$ such that
 $\hatter{u'} + \omega(u',v) \in  I_{j}$.
 Note that $u$ is also a candidate neighbour.
By  Lemma~\ref{lem:guessDProb}, the probability that the procedure \emph{GuessDistance}
fails to find a distance estimate for vertex $v$ in this sub-range is at most $1/n^{c_1}$.
By union-bound, the probability that for \emph{for some} vertex $v \in V \setminus S$, we fail to find 
an estimate for $d^{(k+1)}_G(v,S)$ in the appropriate sub-range is at most $1/n^{c_1 -1}$.
(Our overall probability of failing to find an estimate of $d^{(k+1)}_G(v,S)$ for some vertex $v$
in the appropriate sub-range is therefore at most  $1/n^{c_1 -1}$ plus $k/n^{c_1 -1}$ from 
the inductive hypothesis. 
In total, the failure probability is at most $\frac{k+1}{n^{c_1 -1}}$, as required.)
We assume henceforth that the $j^{th}$ guess for vertex $v$ is successful.

By induction hypothesis, $\hatter{u} \le (1+\zeta')^k \cdot d^{(k)}_G(u,S)$.
Therefore,
\begin{equation*}
\begin{aligned}
\hatter{u} + \omega(u,v) &\le (1+\zeta')^k \cdot d^{(k)}_G(u,S) + \omega(u,v)\\
&\le (1+\zeta')^k \cdot (d^{(k)}_G(u,S) +  \omega(u,v))\\
&= (1+\zeta')^k \cdot d^{(k+1)}_G(v,S).
\end{aligned}
\end{equation*}

Moreover, 
$(\hatter{u'} + \omega(u',v))$ and $(\hatter{u} + \omega(u,v))$ belong to the same sub-range $I_j$, and thus,
\begin{equation*}
\begin{aligned}
\hatter{u'} + \omega(u', v) \le (1 + \zeta') \cdot (\hatter{u} + \omega(u,v))
 \le (1+ \zeta')^{k+1} \cdot d^{(k+1)}_G(v,S).
\end{aligned}
\end{equation*} 

For the lower bound, let $i \le j$ be the minimum index such that procedure \emph{GuessDistance}
succeeds in finding a neighbour $u'_i$ of $v$ with $(\hatter{u'_i} + \omega(u'_i ,v)) \in I_i$.
Then, with probability $1$ we have, $\hatter{u'_i} \ge d^{(k)}_{G} (u'_i, S)$, and thus,
\begin{equation*}
\begin{aligned}
\hatter{v} = \hatter{u'_i} + \omega(u'_i, v)  \ge d^{(k)}_G(u'_i, S) + \omega(u'_i, v) \ge d^{(k+1)}_G(v,S).
\end{aligned}
\end{equation*}
\end{proof}

Lemmas~\ref{lem: wSpace} and~\ref{lem:BFECorrect} imply the following Theorem:
\begin{theorem}\label{thm: SingleBFE}
For a sufficiently large positive constant $c$, 
given an integer parameter $\eta$, an error parameter $\zeta$, an input graph $G(V,E, \omega)$, and a subset $S \subseteq V$,  
the algorithm described in Section~\ref{sec:algo}  performs, with probability at least $1 -\frac{1}{n^{c}}$, 
a $(1+\zeta)$-approximate Bellman-Ford exploration of $G$ rooted at the set $S$ to depth $\eta$,
 and outputs for every $v \in V$, an estimate $\hatter{v}$ of its 
distance to set $S$ and $v$'s parent $\hat{p}(v)$  on the forest spanned by this exploration such that 
$$d^{(\eta)}_G(v,S) \le \hatter{v} \le (1 + \zeta)\cdot d^{(\eta)}_G(v,S)$$
 in $\eta$ passes 
through the dynamic stream using 
$$O_{c} (\eta/\zeta \cdot \log^2 n \cdot \log \Lambda (\log n + \log \Lambda ) ) \text{~space in every pass.}$$
\end{theorem}

The stretch and the space bound follow from Lemmas~\ref{lem: wSpace} and~\ref{lem:BFECorrect} 
  by substituting $\zeta' = \frac{\zeta}{2\eta}$.
Note also that the space used by the algorithm on different passes can be reused, i.e., 
the total space used by the algorithm is $O_{c} ( \eta/\zeta \cdot \log^2 n \cdot \log \Lambda (\log n + \log \Lambda ))$ .

%% file: Superclustering.tex
\subsection{Superclustering} In this section, we describe how the superclustering step of each phase $ i \in \{0,1, \ldots, \ell -1\}$ is executed. The input to phase $i$ is a set of clusters $P_i$.
The phase $i$ begins by sampling each cluster $C \in P_i$ independently at random (henceforth, i.a.r.) with probability $1/deg_i$. 
Let $S_i$ denote the set of sampled clusters. 
We now have to conduct a BFS exploration to depth $\delta_i$ in $G$ rooted at the set $CS_i = \underset{C \in S_i}{\bigcup}\{r_C\}$. 
At this point, we need to move to the dynamic stream to extract the edges of our BFS exploration. 
To do so, we invoke the BFS construction algorithm described in Section~\ref{sec:BFSOutline} with $\eta = \delta_i$ and the set $S = CS_i$ as input. As a result a forest $F_i$ rooted at the centers of the clusters in $S_i$ is constructed. 
By Theorem~\ref{thm: SingleBFS}, the construction fo $F_i$ requires $\delta_i$ passes and $O(n \log^3 n)$ space whp.

For an unsampled cluster center $r_{C'}$ of a cluster $C' \in P_i \setminus S_i$ such that $r_{C'}$ is spanned by $F_i$, 
let $r_C$ be the root of the forest tree in $F_i$ to which $r_{C'}$ belongs. 
The cluster $C'$ now gets superclustered into a cluster $\widehat{C}$ centered around $r_C$. 
The center $r_C$ of $C$ becomes the new cluster center of $\widehat{C}$, i.e., $r_{\widehat{C}} = r_C$. 
The vertex set of the new supercluster $\widehat{C}$ is the union of the vertex set of the original cluster $C$, with the vertex sets of all clusters $C'$ which are superclustered into $\widehat{C}$.  
We denote by $V(C)$ the vertex set of a cluster $C$. 
For every cluster center $r_{C'}$ that  is  spanned by the tree in $F_i$ rooted at $r_C$, the path in $F_i$ from $r_C$ to $r_{C'}$ is added to the edge set $E_H$ of our spanner $H$. 
Recall that $E_H$ is initialized as an empty set (See Section~\ref{sec:Over}.).

Let  $\widehat{P}_i$ denote the set of new superclusters $\widehat{C}$, that were created  by the superclustering step of phase $i$.  
We set $P_{i+1} = \widehat{P}_i$.  
By Theorem~\ref{thm: SingleBFS}, the superclustering step of phase $i$ generates whp, 
a forest of the input graph $G(V,E)$, rooted at the set $CS_i \subseteq V$ in $\delta_i$ passes. 
We conclude that:
\begin{lemma}\label{lem:SuperFinal}
	For a given set of sampled cluster centers $CS_i \subseteq V$ and a sufficiently large constant $c$, 
	the superclustering step of phase $i$ builds with probability at least $1-1/n^c$, 
	disjoint superclusters that contain all the clusters with centers within distance $\delta_i$ from the set of centers $CS_i$. 	
	It does so in $\delta_i$ passes through the stream, 
	using $O_c(n \log^3 n)$ space in every pass.
\end{lemma}

%% file: Interconnection.tex
\subsection{Interconnection}\label{interconnect}
Next we describe the interconnection step of each phase $i \in \{0,1,\ldots,\ell \}$. 
Let $U_i$ denote the set of clusters of $P_i$ that were not superclustered into clusters of $\widehat{P}_i$. 
For the phase $\ell$, the superclustering step is skipped and we set $U_{\ell} = P_{\ell}$.

In the interconnection step of phase $i \ge 1$, 
we want to connect every cluster $C \in U_i$ to every other cluster $C' \in P_i$ that is close to it.
To do this, every cluster center $r_C$ of a cluster $C \in U_i$ performs a BFS exploration up to depth $\frac{1}{2}\delta_i$, 
i.e., half the depth of BFS exploration which took place in the superclustering step, as in~\cite{ElkinNeiman19}. 
For each cluster center $r_{C'}$ of some cluster $C' \in P_i$ which is discovered by the exploration initiated in $r_C$, 
the shortest path between $r_C$ and $r_{C'}$ is inserted into the edge set $E_H$ of our spanner. 
In the first phase $i = 0$, however, we set the exploration depth $\delta_0$ to $1$, 
i.e., to the same value as in the superclustering step. 
Essentially, for every vertex $v \in U_0$, we add edges to all its neighbours to $H$.

Having identified the members of $U_i$, we turn to the stream to find the edges belonging to the BFS explorations 
performed by the centers of clusters in $U_i$. 
The problem here is that we need to perform many BFS explorations in parallel. 
More precisely, there are up to $|P_i |$ explorations in phase $i$.
By Lemma 3.5 of~\cite{ElkinNeiman19}, $|P_i| =  n^{1 - \frac{2^i -1}{\kappa}}$ in expectation for $i \in \{0, 1, \ldots, i_0 \}$ 
and $|P_i| \le n^{1 + 1/\kappa - (i - i_0)\rho}$ in expectation for $i \in \{i_0 + 1, i_0 +2, \ldots, \ell\}$. 
Recall that $i_0 = \floor{\log \kappa\rho}$. 
Invoking Theorem~\ref{thm: SingleBFS}  for $\eta = \delta_i/2$, $S = \{r_C \}$, for some cluster center $r_C$ of a cluster in $U_i$, 
a BFS exploration of depth $\delta_i/2$, rooted at $r_C$ requires $O(n \log^3 n)$ space and $\delta_i/2$ passes.
Running $|P_i|$ explorations in $G$ requires either $O( |P_i|\cdot n \log^3n)$ space or $|P_i| \cdot \delta_i/2$ passes. 
Both these resource requirements are prohibitively large.

We state the following Lemma from~\cite{ElkinNeiman19} here for completeness. \\
We refer the reader to~\cite{ElkinNeiman19} for the proof.
\begin{lemma}[\cite{ElkinNeiman19}]\label{lem:ExpCount} 
For any vertex $v \in V$, the expected number of explorations that visit $v$ in the interconnection step of phase $i$ is at most $deg_i$. 
Moreover, for any constant $c'_1$, with probability at least $1 - 1/n^{c'_1 -1}$,
no vertex $v$ is explored by more than $c'_1\cdot \ln n \cdot \deg_i$ explorations in phase $i$. 
\end{lemma}
In \cite{ElkinNeiman19}, Lemma~\ref{lem:ExpCount} is used to argue that the overall space used by their streaming algorithm in phase $i$ is $O(n \cdot deg_i \log n)$ in expectation. 
Furthermore, since $deg_i \le n^{\rho}$ for all $i \in \{0, 1, \ldots, \ell \}$, the space used by their streaming algorithm is $O(n^{1 + \rho} \log n)$ in expectation in every pass. 
Unfortunately, this argument does not help us to bound the space usage of our algorithm in the dynamic setting. 
When edges may appear as well as disappear, 
a given vertex $v$ may appear on a lot more explorations than $deg_i$ as the stream progresses. 
Lemma~\ref{lem:ExpCount} only guarantees that ultimately paths to at most $deg_i$ centers in $U_i$ will survive for $v$ in expectation.
If we record for every $v \in V$, all the explorations passing through $v$ to identify the ones that finally survive, 
we incur a cost of  $O(|P_i| \cdot  n  \log^3 n)$ space for interconnection during phase $i$, which is prohibitively large.

To tackle this problem, we devise a randomized technique for every vertex to efficiently identify all the (surviving) explorations that it gets visited by in phase $i$. 
For every vertex $v \in V$ with a non-empty subset $U^v_i \subseteq U_i$ of explorations that visit $v$, 
we find for every cluster $C  \in U^v_i$, a neighbour of $v$ on a shortest path between $v$ and the center $r_C$ of $C$. (See Figure~\ref{fig:interconnect}.)

\begin{figure}[h!]
	\centering
	\centerline{\includegraphics[scale=0.17]{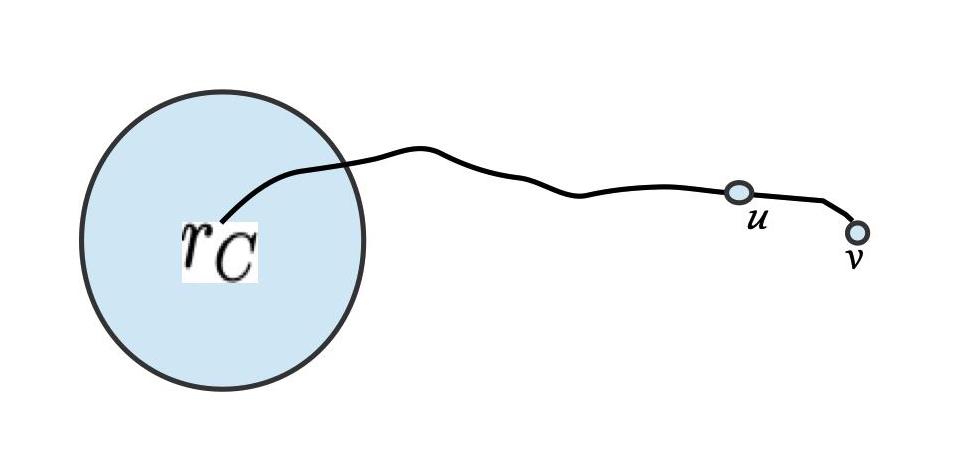}}
	\label{fig:interconnect}
	\caption{A cluster $C \in U^v_i$. The algorithm finds the neighbour $u$ of $v$ on the shortest $r_C - v$ path.}
	\label{fig:interconnect}
\end{figure}

Throughout the interconnection step of phase $i$, we maintain for each vertex $v \in V$, 
a running set $L_v$ of exploration sources that visited $v$. 
Each vertex $s$ in $L_v$ is a center of a cluster $C \in U_i$. 
 We will call the set $L_v$ the \emph{visitor list} of $v$. 
 Initially the visitor lists of all the vertices are empty, except for the centers of clusters in $U_i$. 
 The center $r_C$ of every $C \in U_i$ is initialized with a single element $r_C$ in its visitor list. 
 
  The interconnection step of phase $i$ is carried out in  $\floor{\delta_i/2}$ sub-phases. 
 Each sub-phase of the interconnection step makes two passes through the stream.
 In the following section, we describe the purpose of each of the  $\floor{\delta_i/2}$ sub-phases
 of the interconnection step and the way they are carried out.

\subsubsection{Sub-phase $j$ of interconnection step}\label{sec:interSubPhase}
We discover the edges belonging to the layer $j$ of interconnection in the sub-phase $j$. 
By layer $j$ of interconnection, we mean the set containing every vertex $v$ in $V$, 
whose distance to one or more cluster centers in $U_i$ is exactly $j$.
Note that a given vertex $v$ may belong to more than one layer of interconnection 
since it may be at different distances from different exploration sources, 
and we need to identify all the exploration sources in $U_i$ 
that are within distance $\floor{\delta_i/2}$ from $v$.

 The information regarding the $j^{th}$ layer of interconnection is stored in a set called $S_j$.
Formally, the set $S_j$ consists of tuples of the form $(v, s, k)$, 
where $s$ is an exploration source at distance $j$ from $v$, 
and $k$ is the number of neighbours of $v$ at a distance $j-1$ from $s$. 
 While the visitor list $L_v$ of a specific vertex $v \in V$
 maintains a list of all the exploration sources that visit $v$ 
in  all the sub-phases of the interconnection step, 
 the set $S_j$ is a global list that stores for each vertex $v \in V$, 
 the information about the exploration sources that visited $v$ during sub-phase $j$.
 
Before we start the sub-phase $j$, we create for each $v \in V$, a copy $L'_v$ of its running visitor list $L_v$.
Any new explorations discovered during the sub-phase $j$ are added to the shadow visitor list $L'_v$.
Specifically, $L_v$ is the list of those cluster centers from $U_i$ whose explorations visited $v$ \emph{before}
 sub-phase $j$ started, and $L'_v$ is the list of those centers that visited $v$ on one of the first $j$ sub-phases.

In each of the $\floor{\delta_i/2}$ sub-phases, 
 we make two passes through the stream. 
 In the first pass of sub-phase $j$, 
 we construct the set $S_j$. In more detail, 
 for each vertex $v \in V$, we use a sampler repeatedly in parallel 
 (the exact number of parallel repetitions will be specified later in the sequel)
 to extract whp all the  exploration sources (if there are any) at a distance $j$ from $v$.
 A tuple $(v, s, k_v)$, for some $k_v \ge1$, is added to the set $S_j$ for every source $s$ extracted by the sampler.
The visitor list $L_v$ of $v$ is also updated with the new exploration sources that were observed in this sub-phase.
Specifically, all newly observed exploration sources are added
to $L'_v$. At the end of the sub-phase we set $L_v \leftarrow L'_v$.

The second pass of sub-phase $j$ uses the sets $S_j$ and $S_{j-1}$ to find for every $v \in S_j$, 
 its parent on every exploration whose source is at distance $j$ from $v$. 
 Note that a parent of $v$ on an exploration rooted at the source $s$ is a vertex 
 at distance $j-1$ from $s$. Therefore, we need the set $S_{j-1}$ to extract an 
 edge between $v$ and some vertex $u$ such that a tuple $(u, s, k_u)$, for some $k_u \ge 1$, belongs 
 to the set $S_{j-1}$.
 
 The set $S_{j-1}$, which is constructed during the first pass of phase $j-1$, is used as an input for the second
 pass of phases $j-1$ and $j$. It is therefore kept in global storage until the end of phase $j$.
 
 We next describe how we construct the set $S_j$ during the first pass of sub-phase $j$.

\textbf{First pass of sub-phase $j$  of phase $i$:~} 
Let $c'_1$ be a sufficiently large positive constant (See Lemma~\ref{lem:ExpCount}.), and let $\mathcal{N}_i = c'_1\cdot deg_i \cdot \ln n$.
 For each $v \in V$, we make $\mathcal{\mu}_i = 16\cdot c_4 \cdot \mathcal{N}_i \cdot \ln n$ \emph{attempts}
 in parallel, for some sufficiently large constant $c_4 \ge 1$. 
In each attempt, we invoke a randomized procedure \emph{FindNewVisitor} 
to find an exploration source in $U_i$ at a distance $j$ from $v$. 
The pseudocode for procedure $\emph{FindNewVisitor}$ is given in Algorithm~\ref{alg:findNewVisitor}. 
The procedure \emph{FindNewVisitor} takes as input the ID of a vertex $v$ 
and a hash function $h$, chosen at random from a family of pairwise independent hash functions. 
It returns a tuple $(s, d_s)$,  where $s$ is the ID of an exploration source at distance $j$ from $v$, 
and $d_s$ is the number of neighbours of $v$ that are at distance $j-1$ to $s$. 
This source $s$ is then added to the shadow visitor list $L'_v$ of the vertex $v$.
If there are no exploration sources at distance $j$ from $v$, 
procedure $\emph{FindNewVisitor}$ returns a tuple $(\phi, \phi)$. 
If there are some exploration sources at distance $j$ from $v$ 
but procedure \emph{FindNewVisitor} fails to isolate an ID of such a source, 
it returns $(\perp, \perp)$.

 \begin{algorithm}[h!]
   \caption{Pseudocode for procedure $FindNewVisitor$}
  	\label{alg:findNewVisitor}
 	\begin{algorithmic}[1]
 	  \State $\textbf{Procedure FindNewVisitor} (v,h)$
 	  \Comment{Initialization}
     	\State $\emph{slots} \leftarrow \emptyset$  
	\LeftComment{An array with $\lambda = \ceil{\log n}$ elements indexed  from  $1$ to $\lambda$. }\\
	\LeftComment{Each element of slots is a tuple $(sCount, sNames)$.
	                   For a given index $1 \le k \le \lambda$, fields $sCount$ and 
	                   $sNames$ of $\emph{slots}[k]$ can be accessed as 
	                   $\emph{slots}[k].sCount$ and $\emph{slots}[k].sNames$, respectively.} \\
     	\LeftComment{$\emph{slots}[k].sCount$ counts the new exploration sources seen by $v$ with 
	                     hash values in $[2^k]$.}
          \LeftComment{$\emph{slots}[k].sNames$ is an encoding of the names of new exploration 
                              sources seen by $v$ with hash values in $[2^k].$}


   \Comment{Update Stage}
   \While{$(\text{there is some update~} (e_p,~eSign_p) \text{~in the stream})$} \label{alglin:checkVisitors1}
         \If{$(e_p = (v,u) \text{~satisfies~} L_u \setminus L_v \neq \emptyset)$ } \label{alglin:checkVisitors2}
             	 \ForEach{$s \in L_u \setminus L_v$}
              		 \State$ k \gets \ceil{\log h(s)}$
               	           \Repeat  \Comment{Update $\emph{slots}[k]$ for all $\ceil{\log h(s)} \le k \le \lambda$}
	                         \State $\emph{slots}[k].sCount \gets \emph{slots}[k].sCount + eSign_p $\label{alglin:counteraddtn}
               			\State $\emph{slots}[k].sNames \gets \emph{slots}[k].sNames + \nu(s)\cdot eSign_p$ \label{alglin:vectoraddtn} \\  
               		         \Comment{The function $\nu$ is described in Section~\ref{sec:Encodings}}.\\
		                   \Comment{The addition in line~\ref{alglin:vectoraddtn} is a vector addition.}
%
%
%

               			 \State $k = k+1$
               		   \Until{$k > \lambda$}
             
             	     \EndFor
          \EndIf
     \EndWhile

  \Comment{Recovery Stage}
 	 \If{$(\emph{slots} \text{~vector is empty})$}  
    		 \State return $(\phi, \phi)$ 

  	 \ElsIf{$(\exists \text{~index~} k~\text{s.t.}~\frac{\emph{slots}[k].sName}{\emph{slots}[k].sCount} = \nu(s) \text{~for some $s$ in $V$} )$} \label{alglin:singlesource}
        		\State  return $(s,\emph{slots}[k].sCount )$
     	\Else
           	\State return $(\perp, \perp)$ 
   	\EndIf
 \end{algorithmic}
  \end{algorithm}

 Before we start making our attempts in parallel, 
 we sample uniformly at random a set $H_{j}$ of $\mathcal{\mu}_i$ functions from a family of pairwise independent hash functions 
  $h : \{1,2, \ldots, maxVID \} \rightarrow \{1,\ldots,2^{\lambda}\}$, where $\lambda = \ceil{\log maxVID} = \ceil{\log n}$.
  Having sampled the set $H_j$ of hash functions,
  for every vertex $v \in V$, we make $\mathcal{\mu}_i = |H_{j}|$  parallel calls to procedure $\emph{FindNewVisitor}(v,h)$, one call for each function $h \in H_j$. 
 

 Note that the visitor lists of all the vertices in $V$ are visible to all the calls to procedure $\emph{FindNewVisitor}$, which are made in parallel.

 \textbf{Procedure FindNewVisitor:}
A call to procedure $\emph{FindNewVisitor}$ for a vertex $v$ tracks the edges between $v$ and every vertex $u$ with some explorations in its visitor list $L_u$ that $v$ has not seen so far. 
Let $d^{(j)}_v$ be the number of  exploration sources at distance $j$ from $v$. 
 For every pair of vertices $\{v,u\}$, ultimately either the edge $e = (v,u)$ belongs to $G$ and then $f_e = 1$, or it does not, i.e., $f_e = 0$.
(Recall that $f_e = \sum_{t, e_t = e} eSign_t$ is the multiplicity of edge $e$ in the stream.)
 If we knew the exact value of $d^{(j)}_v$, 
 we could sample every new exploration source witnessed by $v$  with probability $1/d^{(j)}_v$ to extract exactly one of them in expectation.
 However, all we know about $d^{(j)}_v$ is that it is at most $deg_i$ in expectation (Lemma~\ref{lem:ExpCount}) and at most $O(deg_i \cdot \ln n)$ whp. 
We therefore sample every new exploration source seen by $v$ on a range of probabilities, as we did for procedure \emph{FindParent} in Section~\ref{sec:findParent}.
We use an array $\emph{slots}$ of $\lambda$ elements (the structure of each element will be described later in the sequel),
indexed by \emph{slot-levels} from $1$ to $\lambda = \ceil{\log n}$, to implement sampling on a range of probabilities. 
 We want a given source $s$ to be sampled into slot-level $k$ with probability $1/2^{\lambda - k}$. 
 When $d^{(j)}_v \approx 2^{\lambda - k}$, with a constant probability there is exactly one exploration source that gets mapped to $\emph{slots}[k]$.
  
 One way to sample every exploration seen by $v$ with a given probability is to flip a biased coin. 
 As was  discussed in Section~\ref{sec:findParent} in the description of procedure $FindParent$, naively, 
 this requires remembering the random bits for every new exploration source seen by $v$. 
 To avoid storing that much information while still treating all the updates (additions/deletions) to a given exploration source consistently, 
 we use pairwise independent hash functions for sampling explorations. 
 Given a hash function $h : \{1,2,\ldots,maxVID\} \rightarrow \{1,\ldots, 2^{\lambda}\}$, 
 every new exploration source $s$ witnessed by $v$ is assigned a hash value $h(s)$ by $h$. 
 A given source $s$ gets mapped into $slots[k]$ if $h(s) \in [2^k]$, i.e., this happens with probability $1/2^{\lambda - k}$. 
 The description of procedure \emph{FindNewVisitor} is similar to procedure \emph{FindParent} from Section~\ref{sec:findParent} up to this point.
 The major difference between procedure \emph{FindParent} and procedure \emph{FindNewVisitor} is in the information that we store about every sample in a given slot. 
 We cannot afford storing the IDs of all the sampled exploration sources as $v$ may appear on many more explorations than it ends up on. 
  Every new exploration source $s$ assigned to $\emph{slots}[k]$ is first encoded using the CIS encoding scheme $\nu$ described in Section~\ref{sec:Encodings}.
  In every element of $\emph{slots}$, we maintain a tuple $(sCount, sNames)$, 
  where $sCount \in \mathbb{Z}$ at slot-level $k$ maintains the number of new exploration sources seen by $v$ with hash values in $[2^k]$, 
  and $sNames \in \mathbb{Z}^2$ maintains the vector sum of encodings of the IDs of new exploration sources seen by $v$ with hash values in $[2^k]$. 
  This will be discussed in detail in the sequel. The fileds $sCount$ and $sName$ of $\emph{slots}[k]$ can be accessed as $\emph{slots}[k].sCount$ and $\emph{slots}[k].sName$, respectively.

As the stream progresses, every time we encounter an exploration source $s$ with $h(s) \in [2^k]$, 
we update the $sCount$ value of $\emph{slots}[k]$ with the $eSign$ value of the edge from which $s$ was extracted. (See line~\ref{alglin:counteraddtn} of Algorithm~\ref{alg:findNewVisitor}.)
Also, we update the  $sNames$ of $\emph{slots}[k]$ by adding $\nu(s)\cdot eSign_p$ to it (see line~\ref{alglin:vectoraddtn} of Algorithm~\ref{alg:findNewVisitor}), 
where $\nu(s)$ is the encoding of the source $s$ and $eSign_p$ is the $eSign$ value of the edge from which $s$ was extracted. (This addition sums up vectors in $\mathbb{Z}^2$.)
In line~\ref{alglin:singlesource} of Algorithm~\ref{alg:findNewVisitor}, we use Lemma~\ref{lem:convex}  to determine if there is a slot-level $k$ such that only one exploration source was sampled at that level.
Note that the CIS encoding scheme that we use here is more general and can also be used in the implementation of procedure \emph{FindParent}. 
The bitwise XOR-based technique that we use in procedure \emph{FindParent} is an existing technique based on~\cite{GibbKKT15} and~\cite{KingKT15} 
that works for sampling a non-zero element from a Boolean vector. 
The CIS-based technique, on the other hand, allows one to sample a non-zero element from a vector with non-negative entries.

If there is a slot-level $k$ for which $\frac{\emph{slots}[k].sName}{\emph{slots}[k].sCount} = \nu(s)$ for some $s \in V$, 
then by Lemma~\ref{lem:convex}, $s$ is the only exploration source sampled at slot-level $k$. 
The value of $sCount$ at slot-level $k$ will then be the number of neighbours of $v$ at distance $j-1$ from $s$.

 We need to make sure that for some $1 \le k \le \lambda$, exactly one exploration source will get mapped to $slots[k]$. 
 By Corollary~\ref{cor:pairwise}, exactly one exploration source gets mapped to $slots[k]$ for 
  $k = \lambda - \ceil{\log d^{(j)}_v } -1$, with at least a constant probability.
   (Here $\mathcal{S}$ is the set of exploration sources 
  at distance $j$ from $v$ and $s = |\mathcal{S}| = d^{(j)}_v$.)
  Therefore, a single call to procedure \emph{FindNewVisitor} succeeds with at least a constant probability. 

\textbf{Analysis of first pass}: We now analyze the success probability and space requirements of the first pass of sub-phase $j$ of interconnection step.

Recall that, for every vertex $v \in V$, we make $\mathcal{\mu}_i = 16 \cdot c_4 \cdot \mathcal{N}_i \cdot \ln n$ parallel attempts to isolate the exploration sources that visit $v$
during sub-phase $j$ of the interconnection step of phase $i$.

\begin{lemma} \label{lem:findNProb}
	On  any single attempt for a vertex $v \in V$, a given exploration source $s$ at distance $j$ from $v$ is discovered with probability at least $\frac{1}{16 \mathcal{N}_i }$.
	\end{lemma}
\begin{proof}
	Recall that by Lemma~\ref{lem:ExpCount}, with probability at least $1 - \frac{1}{n^{c'_1 -1}}$,
	the number $d^{(j)}_v$ of the exploration sources that visit $v$ is at most $\mathcal{N}_i$. 
	For a specific exploration source $s$ that visits $v$ during sub-phase $j$, let $DISC^{(s)}$ 
	denote the event that it is discovered in a specific attempt.
	Then:
	
	\begin{equation*}
		\begin{aligned}
			Pr \left[DISC^{(s)} \right] &\ge  Pr\left[DISC^{(s)}~|~d^{(j)}_v \le \mathcal{N}_i \right] \cdot Pr\left[d^{(j)}_v \le \mathcal{N}_i \right] \\
			&\ge 	Pr\left[DISC^{(s)}~|~d^{(j)}_v \le \mathcal{N}_i \right] \cdot \left(1 - \frac{1}{n^{c'_{1} -1}} \right) \\
			&\ge \frac{1}{8 \mathcal{N}_i } \left( 1 - \frac{1}{n^{c'_1 -1}}\right) \\
			&\ge \frac{1}{16 \mathcal{N}_i}
		\end{aligned}	
	\end{equation*}
Note that the third inequality follows by applying Lemma~\ref{lem:Pairwise} to the event  $\{DISC^{(s)}~|~d^{(j)}_v \le \mathcal{N}_i \}$.	
\end{proof}

In the next lemma we argue that procedure \emph{FindNewVisitor} does not require too much space.
\begin{lemma}\label{lem:findNewVisitorSpace}
	The procedure FindNewVisitor uses $O(\log^2 n)$ bits of memory.
\end{lemma}
\begin{proof}
 	Procedure \emph{FindNewVisitor} receives as input two variables: the ID of a vertex $v$ and a pairwise independent hash function $h$.
	The ID of any vertex requires $O(\log n)$ bits of space and by Lemma~\ref{lem:2wiseSpace}, 
	a pairwise independent hash function can be encoded in $O(\log n)$ bits too.
 	The visitors lists of all the vertices are available in global storage. 
 	The internal variable $slots$ is an array of size $\ceil{\log n}$. 
 	Each element of the array $slots$ stores an integer counter $sCounter$ of size $O(\log n)$ bits 
 	and an integer vector $sNames$ in $\mathbb{Z}^2$, 
 	which also requires $O(\log n)$ bits of space (See Section~\ref{sec:Encodings}). 
 	The space usage of $slots$ array is therefore $O(\log^2 n)$ bits.
 	It follows thus that procedure \emph{FindNewVisitor} uses $O(\log^2 n)$ bits of memory.
\end{proof}

For a vertex $v \in V$, if there are no exploration sources at a distance $j$ from $v$, 
all the calls to procedure \emph{FindNewVisitor} in all the attempts return $(\phi, \phi)$. 
For all those vertices, we do not need to update their visitor lists. 
For every other vertex $v \in V$, each attempt yields the name of an exploration source at a distance $j$ from $v$ with at least a constant probability. 
We extract the names of all the \emph{distinct} exploration sources from the results of successful attempts and add tuples $(v, s, sCount)$ to the set $S_j$. 
Recall that the set $S_j$ contains tuples $(v,s,k_v)$, where $s$ is an exploration source at distance $j$ from $v$ and $k_v$ is the number of neighbours of $v$ that are at distance $j-1$ from $s$. 
In addition, the source $s$ is added to the visitor list $L'_v$ of vertex $v$.

We next show that making $\mathcal{\mu}_i = 16\cdot c_4 \cdot \mathcal{N}_{i} \cdot \ln n $ attempts in parallel for every vertex $v \in V$ ensures that 
all the relevant exploration sources for every vertex are extracted whp.

\begin{lemma}\label{lem:findNAttempts}
	Let $c_3$ be a sufficiently large constant. For a given vertex $v \in V$, 
	 with probability at least $1 - 1/n^{c_3}$, all the  exploration sources at a distance $j$ from $v$ will be successfully extracted in
	 $\mathcal{\mu}_i = 16 \cdot c_4 \cdot \ln n \cdot \mathcal{N}_i $ attempts made in parallel for $v$ in the first pass of sub-phase $j$.
\end{lemma}
\begin{proof}
	For a given vertex $v$, let $d^{(j)}_v$ be the number of explorations that are at a distance $j$ from $v$. 
	By Lemma~\ref{lem:findNProb}, on each single attempt (out of $\mathcal{\mu}_i$ attempts) for a vertex $v$, a specific exploration source
	that visits $v$ is isolated with probability at least $1/16\mathcal{N}_i$, independently of other attempts.
	Thus, for a given exploration source $s$, the probability that no attempt will isolate it is at most $\left (1 - \frac{1}{16 \mathcal{N}_i} \right)^{16 \cdot c_4 \cdot \ln n \cdot \mathcal{N}_i} \le 1/n^{c_4}$.
	Hence, by union-bound over all the exploration sources at distance $j$ from $v$, 
	all the exploration sources will be isolated during $16 \cdot c_4 \cdot  \ln n \cdot \mathcal{N}_i$ attempts, 
	with probability at least $1 - \frac{1}{n^{c_4 -1}}$.
	Thus, for $c_3 = c_4 -1$, with probability at least $1 - 1/n^{c_3}$, all the  exploration sources at a distance $j$ from $v$ will be successfully extracted.
\end{proof}

We next provide an upper bound on the space usage of the first pass of the interconnection step.
\begin{lemma}\label{lem:firstPassInterSpace}
	The overall space usage of the first pass of every sub-phase of interconnection is $O(n^{1 + \rho} \log^4 n)$ bits.
\end{lemma}
\begin{proof}
	The first pass of every sub-phase makes $\mathcal{\mu} _i =  O(deg_i \cdot \log^2 n)$ attempts in parallel for every $v \in V$. 
	Recall that for all $i$, $deg_i \le n^{\rho}$ (See Section~\ref{sec:Over}). 
	Combining this fact with Lemma~\ref{lem:findNewVisitorSpace}, 
	we get that the space usage of all the invocations of procedure \emph{FindNewVisitor} for all the $n$ vertices during the first pass is $O(n^{1 + \rho} \log^4 n)$. 
	In addition, we use a set of $O(deg_i \cdot \log^2 n) = O(n^{\rho} \cdot \log ^{2} n)$  randomly sampled hash functions, one hash function per attempt.
	Each hash function can be encoded using $O(\log n)$ bits.
	The overall space used by the storage of hash functions during the first phase is thus $O(n^{\rho} \log^3 n)$. 
	As an output, we produce  the set $S_j$, which consists of tuples $(v,s,k)$ of $O(\log n)$ bits each. 
	By Lemma~\ref{lem:ExpCount}, a vertex $v$ is visited by at most $O(deg_i \log n)  \le O(n^{\rho} \log n)$ explorations whp in the phase $i$.  
	In any case, we record just $O(n^{\rho} \cdot \log n)$ of them, even if $v$ is visited by more explorations.
	Hence, the storage of $S_j$ requires $O(n^{1 + \rho} \log^2 n)$ bits. 
	Finally, we need to store the visitor lists of all $v \in V$. 
	By Lemma~\ref{lem:ExpCount}, no vertex is visited by more than $O(deg_i \log n) =  O(n^{\rho} \log n)$ explorations whp. 
	As above, we record just $O(n^{\rho} \cdot \log n)$ of the visitors for $v$.
	We need to store $O(\log n)$ bits of information for every exploration source that visited a given vertex. 
	The overall storage cost of all the visitor lists of all the vertices is therefore $O(n^{1+ \rho} \log ^2 n)$ bits.
	Thus, the storage cost of first pass of every sub-phase is dominated by the cost of parallel invocations of procedure  \emph{FindNewVisitor}. 
	This makes the overall cost of first pass of every sub-phase $O(n^{1 + \rho} \log^4 n)$.
\end{proof}

\textbf{Second pass of sub-phase $j$ of Phase $i$:~} The second pass of sub-phase $j$ starts with the sets $S_{j-1}$ and $S_j$ as input. 
Recall that the set $S_j$ consists of tuples for all the vertices in $V$ that are at distance $j$ from one or more exploration sources in $U_i$.
The algorithm also maintains an additional intermediate edge set $\hat{H}_i$, which will contain all the BFS trees rooted at cluster centres $r_C$, $C \in U_i$, constructed to depth $\delta_i/2$.
Inductively, we assume that before sub-phase $j$ starts, the edge set $\hat{H}_i$ contains the first $j-1$ levels of these trees.
Note that since by Lemma~\ref{lem:ExpCount}, whp, every vertex $v$ is visited by $O(deg_i \cdot \log n)$ explorations rooted at 
$\{r_C\}_{C \in U_i}$, it follows that,
whp, $|\hat{H}_i| = \tilde{O}(n \cdot deg_i) = O(n^{1 + \rho} \cdot \log n)$.
Thus our algorithm can store the set $\hat{H}_i$.
We find for every tuple $(v,s,k)$ in $S_j$, $v$'s parent $p_s$ on the exploration rooted at $s$ 
by invoking procedure  \emph{FindParent} (described in Section~\ref{sec:findParent}) $O(\log n)$ times. 
As a result, an edge $(v, p_s)$ between $v$ and $p_s$ is added to the edge set $\hat{H}_i$.

We sample uniformly at random a set of pairwise independent hash functions $H'_{j}$, $|H'_{j}| = c_1 \cdot \log_{8/7}n$,
from the family of functions $h : \{1,2,\dots,\emph{maxVID} \} \rightarrow \{1,2,\ldots,2^{\lambda}\}$, $\lambda = \ceil{\log n}$.
These functions will be used by invocations of procedure \emph{FindParent}. 

 We need to change slightly the original procedure \emph{FindParent} (Section~\ref{sec:findParent}) to work here. 
 Specifically, we change the part where we decide whether to sample an incoming edge update or not (Line~\ref{alglin:singleedge} of Algorithm~\ref{algps:findParent}). 
 It is updated to check if the edge $e_p$ is incident between the input vertex $v$ and some vertex $u$ such that for some $k$, the tuple $(u, s, k)$  belongs to the set $S_{j-1}$.
 Recall that for a tuple $(v,s,k) \in S_j$, $k$ is the number of neighbours of $v$ that are at a distance $j-1$ from $s$. 
 This information can be used to optimize the space usage of procedure \emph{FindParent} by a factor of  $O(\log n)$. 
 Since we know the probability ($\approx 1/k$) with which to sample every candidate edge for $v$, 
 we can get rid of the array $\emph{slots}$ and maintain only two running variables $xCount$ and $xName$ corresponding to slot-level $\lambda - \ceil{\log k} -1$. 
 
 Finally, after all the $\delta_i/2$ sub-phases are over, we extract from $\hat{H}_i$ edges that need to be added to the spanner $H$ offline,
 during post-processing. Specifically, for every cluster center $r_C$, $C \in U_i$, we consider the BFS tree $T(r_C)$ rooted at $r_C$ of depth $\delta_i/2$,
 which is stored in $\hat{H}_i$. For any leaf $z$ of $T(r_C)$ which is not a center of a cluster $C' \in P_i$, the leaf $z$ and the 
 the edge connecting $z$ to its parent $p_z$ in $T(r_C)$ are removed from $T(r_C)$ (and thus from $\hat{H}_i$).
 This process is then iterated, until all leaves of $T(r_C)$ are cluster centers. This is done for all cluster centers $r_C$, $C \in U_i$, one after another.
 The resulting edge set $H'_i$ (a subset of $\hat{H}_i$) is then added to the spanner $H$.

Observe that this edge set $H'_i$ is precisely the union of all shortest paths $r_C-r_C'$, for $C \in U_i$ and $C' \in P_i$,
such that $d_G(r_c, r_C') \le \delta_i/2$.
It follows that, (see~\cite{ElkinNeiman19}), its size is at most $\delta_i/2 \cdot |U_i| \cdot deg_i = \tilde{O}(\delta_i \cdot n^{1+ 1/\kappa})$.
This bound can be further refined by optimizing the degree sequence $(deg_i)_{i=1}^{\ell}$. (See~\cite{ElkinNeiman19} for details.)
 
 \textbf{Analysis of Second Pass:~} We now analyze the space requirements of the second pass of sub-phase $j$ of interconnection step.
\begin{lemma}\label{lem:secondPassInterSpace}
	The overall space usage of the second pass of every sub-phase of interconnection is $O(n^{1 + \rho} \log^4 n)$.
\end{lemma}
\begin{proof}
	The second  pass of every sub-phase invokes procedure \emph{FindParent}  $O(\log n)$ times in parallel for every tuple in the set $S_j$. 
	By Lemma~\ref{lem:findPSpace}, each invocation of procedure \emph{FindParent}  uses $O(\log ^2 n)$ bits of space.
	The number of elements in $S_j$ is at most $O( n^{1 + \rho} \log n)$.
	(Recall that by Lemma~\ref{lem:ExpCount}, whp there are at most $\tilde{O}(n^{\rho})$ explorations per vertex.
	But even if there are more explorations,
	our algorithm records just $\tilde{O}(n^{\rho})$ explorations per vertex.)
 	Therefore the overall cost of all the invocations of procedure \emph{FindParent} is $O(n^{1 + \rho} \log^4 n )$. 
 	In addition, we need to store a set of $O(\log n)$  hash functions of size $O(\log n)$ each in global storage. 
 	This requires $O(\log^2 n)$ bits of space.
	Therefore, the overall storage cost  of the second pass of any sub-phase is dominated by the space required for invocations of \emph{FindParent}.
 	Hence the overall space requirement of second pass of  interconnection is $O(n^{1 + \rho} \log^4 n)$.
\end{proof}

In the following lemma we prove the correctness of the interconnection step.
\begin{lemma}\label{lem:InterConnectCorrect}
	For a sufficiently large constant $c'$, after $j$ sub-phases of phase $i$ of the interconnection step, with probability at least $1 - j/n^{c'}$,
	for every cluster $C \in U_i$ and for every vertex $v$ within distance $j$ from the center $r_C$ of $C$, 
	a shortest path between $r_C$ and $v$ is added to the edge set $\hat{H}_i$.
\end{lemma}
\begin{proof}
	The proof follows by induction on the number of sub-phases, $j$, of the interconnection step of phase $i$. 
	The base case for $j = 0$ holds trivially.
	For the inductive step, we assume that after $j=t$ sub-phases of interconnection step (Section~\ref{sec:interSubPhase}), 
	for every cluster $C \in U_i$ and for every vertex $v$  within distance $t$ from the center $r_C$ of $C$, 
	a shortest path between $r_C$ and $v$ has been added to $\hat{H}_i$ with probability at least $1 - t/ n^{c'}$. 
	Given this assumption, we only need to prove that in the sub-phase $t +1$, 
	we find for every cluster $C \in U_i$ and  for every vertex $v$ at distance $t + 1$ from the center $r_c$ of $C$, 
	a parent for $v$ on the BFS exploration rooted at $r_C$ with probability at least $1 -1/n^{c'}$.
	In the first pass of sub-phase $t +1$, for every vertex $v \in V$, 
	we make $\mathcal{\mu}_i =16 \cdot c_4 \cdot \ln n \cdot \mathcal{N}_i$ attempts to extract all the cluster centers at distance $t + 1$ from $v$. 
	By Lemma~\ref{lem:findNAttempts}, each such center gets extracted with probability at least $1- 1/n^{c_3}$. 
	There are no more than $n$ clusters in $U_i$. 
	Applying union bound over all the clusters in $U_i$ and over all the vertices at distance $t + 1$ from one or more centers in $U_i$, 
	we successfully extract all the  exploration sources at distance $t + 1$ from every vertex in the sub-phase $t+1$ with probability at least $1- 1/n^{(c_3-2)}$.
	 In the second pass of sub-phase $t+1$, we try to find a parent for $v$ on every exploration at distance $t+ 1$ by making multiple parallel calls to procedure \emph{FindParent}. 
 	By Lemma~\ref{lem:findPProb}, we succeed in finding a parent for $v$ on a single BFS exploration with probability at least $1- 1/n^{c_1}$. 
 	By union bound over all the clusters in $U_i$  and all the vertices at  distance $t + 1$ from one or more centers,  the second pass of sub-phase $t+ 1$ 
 	succeeds with probability at least $1 - 1/n^{c_1 -2}$.
 	Taking a union bound on both the passes of sub-phase $t+1$, 
 	we get that for an appropriate constant $c'$, in the sub-phase $t+ 1$, 
 	for every cluster $C \in U_i$ and  for every vertex $v$ at distance $t+1$ from the center $r_c$ of $C$, 
 	we find a parent for $v$ on the BFS exploration rooted at $r_C$ with probability at least $1 - 1/n^{c'}$.
\end{proof}

Lemmas~\ref{lem:firstPassInterSpace},~\ref{lem:secondPassInterSpace} and~\ref{lem:InterConnectCorrect} together imply the following corollary about the interconnection step of phase $i$:
\begin{corollary}\label{lem:interFinal}
For a sufficiently large constant $c''$, after $\floor{\delta_i/2}$ sub-phases of phase $i$ of the interconnection step, the following holds with probability at least $1-1/n^{c''}$:
\begin{enumerate}
\item The interconnection step of phase $i$ makes $\delta_i$ passes through the stream, and the total required space is $O(n^{1+\rho} \log^4 n)$ bits.
\item For every cluster $C \in U_i$ and every other cluster $C' \in P_i$ such that the centers $r_C'$ of $C'$ is within distance $\floor{\delta_i/2}$ from center $r_C$ of $C$, 
a shortest $r_C- r_{C'}$ path between them is added to the spanner.
\end{enumerate}
\end{corollary}

%% file: final.tex
\subsection{Putting Everything Together}\label{sec:finalAnalysis}
Lemma~\ref{lem:SuperFinal} and Corollary~\ref{lem:interFinal} imply that, whp, our algorithm simulates phase $i$ of~\cite{ElkinNeiman19}.  
The following lemma follows by induction on the number of phases of our algorithm.
\begin{lemma}\label{lem:Final}
After $\ell$ phases, whp, our spanner construction algorithm simulates the algorithm of~\cite{ElkinNeiman19} in the dynamic streaming setting.
\end{lemma}
Next, we provide a bound on the number of passes of our algorithm.
\begin{lemma}\label{lem:totalPass}
Our spanner construction algorithm makes $O(\beta)$ passes in total.
\end{lemma}
\begin{proof}
In a given phase $i$ of our construction algorithm, the superclustering step makes $\delta_i$ passes and the interconnection step makes $2\floor{\delta_i /2}$ passes. 
The number of passes of phase $i$ is therefore bounded by $O(\delta_i)$.
 Note that $\sum_{i =1}^{\ell} \delta_i = O(\beta)$, where $\beta$ is  the additive term in the stretch of our construction (See~\cite{ElkinNeiman19}). 
 The number of passes made altogether is thus bounded by $O(\beta)$.
\end{proof}
The stretch and sparsity analysis of our dynamic streaming algorithm remains the same as that of the centralized algorithm of~\cite{ElkinNeiman19}. 
Hence we obtain the following analogue of Corollary $3.2$ of~\cite{ElkinNeiman19} for the dynamic streaming setting.

\begin{theorem}\label{thm:main}
For any unweighted graph $G(V,E)$ on $n$ vertices, parameters $0 < \epsilon <  1$, $\kappa \ge 2$, and $\rho >0$, 
our dynamic streaming algorithm computes a $(1 + \epsilon, \beta)$-spanner with  
$O_{\epsilon,\kappa, \rho}(n^{1 + 1/\kappa})$ edges, in $O(\beta)$ passes using $O(n^{1 + \rho}\log^4 n)$ space with high probability, where $\beta$ is given by:
\begin{align*}
   \beta =  \left(\frac{\log \kappa\rho + 1/\rho}{\epsilon}  \right)^{\log \kappa\rho + 1/\rho}.
\end{align*}
\end{theorem}

%% file: Applications.tex
In the following section, we show some applications of our construction of near-additive spanners.

\section{$(1+\epsilon)$-Approximate Shortest Paths in Unweighted Graphs}\label{sec:Applications}
An immediate application of our dynamic streaming algorithm for constructing $(1 + \epsilon, \beta)$-spanners is 
a dynamic streaming algorithm for computing \emph{all pairs almost shortest paths} (APASP) with multiplicative 
stretch $1 + \epsilon$ and additive stretch $\beta$ (henceforth, $(1+\epsilon, \beta)$-APASP) in unweighted undirected graphs.
The algorithm uses $O(\beta)$ passes over dynamic stream and $\tilde{O}(n^{1 + \rho})$ space.
Our $(1+\epsilon, \beta)$-APASP algorithm computes a $(1+\epsilon, \beta)$-spanner with $O_{\epsilon, \kappa, \rho}(n ^{1 + 1/\kappa})$ using Theorem~\ref{thm:main},
and then computes offline all pairs exact shortest paths in the spanner.

We note also that within almost the same complexity bounds, the algorithm can also compute $(1+\epsilon)$-approximate
shortest paths $S \times V$ (henceforth, $(1+ \epsilon)$-ASP), for a subset $S$ of size $n^{\rho}$ of designated sources.
Specifically, the algorithm computes the $(1 + \epsilon, \beta)$-APASP in the way described above.
It then uses $O(\beta/\epsilon)$ more passes to compute BFS trees rooted in each of the sources $s \in S$ to depth 
$\beta/\epsilon$ in the original graph $G$. The space usage of this step is $\tilde{O}(|S| \cdot n) = \tilde{O} (n^{1 + \rho})$. (see Theorem~\ref{thm: SingleBFS})

As a result, for every pair $(s, v) \in S \times V$ such that $d_G(s,v) \le \beta/\epsilon$, our algorithm returns an exact distance.
For each pair $(s, v) \in S \times V$ with $d_G(s,v) > \beta/\epsilon$, the estimate computed using $(1+\epsilon, \beta)$-APASP
algorithm provides a purely multiplicative stretch of $1 + O(\epsilon)$. The algorithm returns the minimum of these two estimates.

By setting $\kappa = 1/\rho$ we obtain:

\begin{theorem}
	For any undirected $n$-vertex graph $G = (V, E)$, and any $\epsilon > 0$, $\rho > 0$, our dynamic streaming algorithm computes
	$(1+\epsilon, \beta)$-APASP and $(1 + \epsilon)$-ASP for a set $S$ of $|S| = n^{\rho}$ sources 
	using $\beta = O(\frac{1}{\rho \epsilon})^{\frac{1}{\rho} (1 + o(1))}$ passes
	and $\tilde{O}(n^{1 + \rho})$ memory.
\end{theorem}

One noteable point on the tradeoff curve is $\rho = \sqrt{\frac{\log\log n}{\log n}}$. Then we get $2^{O(\sqrt{\log n \cdot \log\log n})}$ passes and
$n \cdot 2^{O(\sqrt{\log n \cdot \log\log n})}$ space. Also using $\rho = \frac{(\log \log n)^c}{\log n}$ for sufficiently large constant $c$, we get
$n^{o(1)}$ passes and $\tilde{O}(n)$ space.

%% file: Hopset.tex
\section{Hopsets with Constant Hopbound in Dynamic Streaming Model}\label{sec:Hopsets}
%
%
%

Our hopset construction algorithm is based on superclustering and interconnection approach 
that was originally devised for the construction 
of near-additive spanners~\cite{ElkinPeleg}. (See Section~\ref{alg:algorithmMain} for more details.)
Elkin and Neiman~\cite{ElkinNeimanHopsets} used the superclustering and interconnection approach
for the construction of hopsets with constant hopbound in 
various models of computation including the insertion-only streaming model. 
We adapt here the insertion-only streaming algorithm of~\cite{ElkinNeimanHopsets} to 
work in the dynamic streaming setting.

The main ingredient of both the superclustering and interconnection steps is 
a set of Bellman-Ford explorations up to a given distance in the input graph from a set of chosen vertices.  
The insertion-only streaming algorithm of~\cite{ElkinNeimanHopsets} identifies
 all the edges spanned by $\Theta(\beta)$ iterations of 
certain Bellman-Ford explorations
up to a distance $\delta$ from a set of chosen vertices,
 by making $\Theta(\beta)$ passes through the stream. 
Other parts of the hopset construction, such as identifying the vertices of the graph
 from which to perform Bellman-Ford explorations 
and subsequently adding edges corresponding to certain paths 
traversed by these explorations to the hopset, are performed offline. 
 
We devise a technique to perform a given number of iterations of 
a Bellman-Ford exploration from a set of chosen vertices
and up to a given distance in the graph in the dynamic streaming setting,
 and as in~\cite{ElkinNeimanHopsets}, perform the rest of the work offline.
The difference however is that in the dynamic streaming setting, 
we do not perform an exact and deterministic Bellman-Ford exploration (as in~\cite{ElkinNeimanHopsets}).
A randomized algorithm for performing an approximate Bellman-Ford exploration 
originated at a subset of source vertices in a weighted graph, 
that succeeds whp,  is described in Section~\ref{sec:WeightedForest}. 
 We use this algorithm as a subroutine in the superclustering step of our main algorithm. 

 The interconnection step is more challenging and involves performing multiple simultaneous 
 Bellman-Ford explorations in a weighted graph, each from a separate source vertex.
Here, each vertex in the graph needs to identify all the Bellman-Ford explorations it is a part of, 
and to find its (approximate) distance to the source of each such exploration.
Due to the dynamic nature of the stream, 
a given vertex may find itself on a lot more explorations than it finally ends up belonging to. 
As shown in Section~\ref{interconnect} in the context of near-additive spanner construction, 
this can be dealt with by combining a delicate encoding/decoding scheme 
for the IDs of exploration sources with a space-efficient sampling technique.
We adapt here the technique used in Section~\ref{interconnect} to work in weighted graphs.

In the following section, we provide an overview of our hopset construction algorithm. 
\subsection{Overview}\label{sec:HopsetsAlgoOverview}
Our hopset construction algorithm takes as input an $n$-vertex weighted undirected graph $G = (V, E, \omega)$,
and parameters $0 < \epsilon' < 1/10$, $\kappa = 1,2,\ldots$ and $1/\kappa < \rho < 1/2$, and produces as output 
a $(1+\epsilon', \beta')$-hopset of $G$. 
The hopbound parameter $\beta'$ is a function of $\epsilon'$, $\Lambda$, $\kappa$,
$\rho$ and is given by 

\begin{align}\label{eq:hopboundclaim}
\beta' = O\left(\frac{\log \Lambda}{\epsilon'} \cdot \left(\log \kappa \rho + 1/\rho \right)  \right)^{\log \kappa \rho + 1/\rho}
\end{align}

Let $k = 0,1,\ldots, \ceil{\log \Lambda} -1$. Given two parameters $\epsilon > 0$ and $\beta = 1,2,\ldots$,
 a set of weighted edges $H_k$ on the vertex set $V$ 
of the input graph is said to be a $(1 + \epsilon, \beta)$-\emph{hopset for the scale $k$} 
or a \emph{single-scale hopset}, 
if for every pair of vertices $u, v\in V$ with $d_G(u,v) \in (2^k, 2^{k+1}]$
we have that:

\begin{align*}
d_G(u,v) \le d^{(\beta)}_{G_k}(u,v) \le (1 + \epsilon) \cdot d_G(u,v),
\end{align*}
where $G_k = (V, E \cup H_k, \omega_k)$ and $\omega_k(u,v) = min\{\omega(u,v), \omega_{H_k}(u,v)\}$, for every edge $(u,v) \in E \cup H_k$.

Let $\epsilon >0$ be a parameter that will be determined later in the sequel. 
Set also $\ell = \floor{\log \kappa \rho} + \ceil{\frac{\kappa + 1}{\kappa \rho}} -1$.
Let $\beta = (1/\epsilon)^{\ell}$.

The algorithm constructs a separate $(1+\epsilon, \beta)$-hopset $H_k$ for every scale \\
$(2^0, 2^1], (2^1, 2^2], \ldots,  (2^{\ceil{\log \Lambda} -1}, 2^{\ceil{\log \Lambda}} ]$ one after another.
For $k \le \floor{\log \beta} -1$, we set $H_k = \phi$. 
We can do so because for such a $k$, it holds that
 $2^{k+1} \le \beta$, 
and for every pair of vertices $u,v$ with $d_G(u,v) \le 2^{k+1}$,
the original graph $G$ itself contains a shortest path between
 $u$ and $v$ that contains at most $\beta$ edges. 
 (We remark that after rescaling, we will have $\beta' = \beta$. 
 See Section~\ref{sec:HopsetFinal}.)
In other words, $d_G(u,v) = d_G^{(\beta)} (u,v)$. 
Denote $k_0 = \floor{\log \beta}$
and $k_{\lambda} = \ceil{\log \Lambda} -1$. 
We construct a hopset $H_k$ for every $k \in [k_0, k_{\lambda}]$.

During the construction of the hopset $H_k$ for some $k \ge k_0$, 
we need to perform explorations from certain vertices in $V$ 
up to distance $\delta \le 2^{k +1}$ in $G$.
An exploration up to a given distance from a certain vertex in $G$
 may involve some paths with up to $n-1$ hops. 
This can take up to $O(n)$ passes through the stream.
We overcome this problem by using the hopset edges
$H^{(k-1)} = \bigcup_{k_0\le~j~\le k-1} H_j$ 
for constructing hopset $H_k$.
The hopset $H_{k}$ has to take care of  all pairs of vertices 
$u,v$ with $d_G(u,v) \in (2^{k}, 2^{k +1}]$,
whereas the edges in $E \cup H^{(k -1)}$ 
provide a $(1 + \epsilon_{k-1})$-approximate shortest path with up to $\beta$ hops,
for every pair $u,v$ with $d_G(u,v) \le 2^k$. 
The value of $\epsilon_{k-1}$ will be specified later in the sequel.
Denote by $G^{(k-1)}$ the graph obtained by adding the edge set $H^{(k -1)}$ to the input graph $G$.
Instead of conducting explorations from a subset $S \subseteq V$
up to distance $\delta \le 2^{k+1}$ in the input graph $G$, 
we perform $2\beta + 1$ iterations of Bellman-Ford algorithm on the graph $G^{(k-1)}$
up to distance $(1 + \epsilon_{k-1}) \cdot \delta$.
The following lemma from~\cite{ElkinNeimanHopsets} 
shows that $2\beta + 1$ iterations of Bellman-Ford algorithm
on $G^{(k-1)}$ up to distance $(1 + \epsilon_{k-1}) \cdot \delta$ suffice to reach all the vertices within distance $\delta$ from 
set $S$ in the original graph $G$. We refer the reader to Lemma 3.9 (and its preamble) of~\cite{ElkinNeimanHopsets} for the proof.

\begin{lemma}~\cite{ElkinNeimanHopsets}\label{lem:2beta +1}
For $u,v \in V$ with $d_G(u,v) \le 2^{k+1}$, the following holds:

\begin{align}\label{eq:2beta +1}
d^{(2\beta + 1)}_{G^{(k-1)}} (u,v)  \le (1 + \epsilon_{k-1}) \cdot d_G(u,v)
\end{align}
\end{lemma}

\subsection{Constructing $H_k$}\label{sec:constructH_k}
We now proceed to the construction of the hopset $H_k$
for the scale $(2^k, 2^{k+1}]$, for some $k \in [k_0, k_{\lambda}]$.
The algorithm is based on the superclustering and interconnection approach.
The overall structure and technique of the construction of a single scale hopset
 is similar to that of the construction of a near-additive sparse spanner. (See Section~\ref{alg:algorithmMain}.) 
The spanner construction algorithm of Section~\ref{alg:algorithmMain} works on an 
unweighted input graph and selects a subset of edges of the input graph as output.
On the other hand, the hopset construction algorithm presented here works on a weighted input graph 
and produces as output a set of new weighted edges that need to be added to the input graph.

The algorithm starts by initializing the hopset $H_k$ as an empty set.
 As in the construction of near-additive spanners (See Section~\ref{alg:algorithmMain}),
the algorithm proceeds in phases $0,1,\ldots, \ell$. 
 The maximum phase index $\ell$ is set as  $\ell = \floor{\log \kappa \rho}  + \ceil{\frac{\kappa + 1}{\kappa\rho}} -1$. 
 Throughout the algorithm, we build clusters of nearby vertices. The input to phase $i \in [0, \ell]$ 
is a set of clusters $P_i$, a distance threshold parameter $\delta_i$ and a degree parameter $\deg_i$.
 For phase $0$, the input  $P_0$ is a partition of the vertex set $V$ into singleton clusters.
The definitions of the center $r_C$ of a cluster $C$, its radius $Rad(C)$ and the radius of a partition $Rad(P_i)$ 
remain the same as in the case of spanner construction. (See Section~\ref{alg:algorithmMain} for more details.)
Note, however that in the current context, the distances are in a weighted graph, $G^{(k-1)}$, rather than in the 
the unweighted input graph $G$, as it was the case in the construction of spanners.

The degree parameter $deg_i$ follows the same sequence 
as in the construction of near-additive spanners.
The set of phases $[0, \ell]$ is partitioned into two stages 
based on how the degree parameter changes 
from one phase to the next. (See Section~\ref{sec:Over} for more details.) 
The distance threshold parameter grows at the same steady rate
(increases by a factor of $1/\epsilon$) in every phase.

For clarity of presentation, we first define the sequence of the distance
 threshold parameters for hopset $H_k$ as if all the explorations during 
 the construction of $H_k$ are exact 
 and are performed on the input graph $G$ (as in the centralised setting) itself. 
Then we modify this sequence to account for the fact that 
the explorations during the construction of $H_k$
are actually conducted on the graph $G^{(k-1)}$ 
and not on the input graph $G$.
The sequence of the distance threshold parameters for the centralized construction 
as defined in~\cite{ElkinNeimanHopsets} 
is given by $\alpha = \alpha^{(k)} = \epsilon^{\ell} \cdot   2^{k+1}$, 
$\delta_i = \alpha(1/\epsilon)^i + 4R_i$, where $R_0 = 0$ and
$R_{i+1} = R_i + \delta_i = \alpha(1/\epsilon)^i + 5R_i$ for $i \ge 0$. 
Here $\alpha$ can be perceived as a unit of distance.
To adjust for the fact that explorations are performed on the graph $G^{(k-1)}$, 
we multiply all the distance thresholds $\delta_i$
by a factor of $1 + \epsilon_{k-1}$, the stretch guarantee of the graph $G^{(k-1)}$.
We further modify this sequence to account for the fact 
that our Bellman-Ford explorations (during superclustering as well as interconnection)
 in the dynamic stream are not exact and incur a 
multiplicative error.
Throughout the construction of $H_k$, 
we set the multiplicative error of every approximate Bellman-Ford Exploration 
we perform to $1 + \chi$, for a parameter $\chi >0$ which will be determined later.
Therefore we multiply all the distance thresholds by a factor of $1 + \chi$.
We define $R'_i  = (1 + \chi) \cdot (1 + \epsilon_{k-1}) R_i$ 
and $\delta'_i = (1 + \chi) \cdot (1 + \epsilon_{k-1}) \delta_i$ for every $i \in [0,\ell]$.
In the centralized setting, $R_i$ serves as an upper bound on the radii of the input clusters of phase $i$. 
As a result of rescaling, $R'_i$ becomes the new upper bound on the radii of input clusters of phase $i$.

All phases of our algorithm except for the last one consist of two steps, 
a superclustering step and an interconnection step. 
In the last phase, the superclustering step is skipped  
and we go directly to the interconnection step.
The last phase is called the \emph{concluding} phase.

The \emph{superclustering} step of phase $i$ randomly samples
 a set of clusters in $P_i$ and builds larger clusters around them. 
The sampling probability for phase $i$ is $1/deg_i$. 
In the insertion-only algorithm of~\cite{ElkinNeimanHopsets}, 
for every unsampled cluster center $r_C'$ 
within distance $\delta_i$ (in $G$) from 
the set of sampled centers, an edge $(r_C, r_C')$ 
between $r_C'$ and a nearest sampled center $r_C$
of weight  $\omega_{H_{k}}(r_C, r_C') = d^{(2\beta + 1)}_{G^{(k-1)}} (r_C, r_C')$ 
is added into the hopset $H_k$.
In the dynamic stream, the distance exploration we do in $G^{(k-1)}$ is not exact 
and we have an estimate of $d^{(2\beta + 1)}_{G^{(k-1)}} (r_C, r_C')$ which is stretched
at most by a multiplicative factor of $1 + \chi$.
Hence in our algorithm, 
$\omega_{H_k}(r_C, r_C') \le (1 + \chi) \cdot d^{(2\beta + 1)}_{G^{(k-1)}} (r_C, r_C')$.
The collection of the new larger clusters $\hat{P}_i$ is passed on as input to phase $i + 1$.
In the \emph{interconnection} step of phase $i$, the clusters that were not superclustered
in this phase are connected to their nearby clusters. 
In the insertion-only algorithm of~\cite{ElkinNeimanHopsets}, $2\beta + 1$ iterations 
of a Bellman-Ford exploration from the center $r_C$ of every cluster in $U_i = P_i \setminus P_{i+1}$ 
are used to identify every other cluster in $U_i$ whose center is 
 within distance $\delta_i/2$ (in $G$) from $r_C$.
For every center $r_C'$ within distance $\delta_i/2$ (in $G$) from the center $r_C$ of $C \in U_i$,
 an edge $(r_C, r_C')$ of weight  
$\omega_{H_k} (r_C, r_C') = d^{(2\beta + 1)}_{G^{(k-1)}} (r_C, r_C')$ 
is added into the hopset $H_k$.
In the dynamic stream, we do $2\beta +1$ iterations of a $(1 + \chi)$-approximate Bellman-Ford
 exploration from every center. 
 Therefore as in superclustering step, 
 the weights of hopset edges added during interconnection step are stretched at most by a factor of $1 + \chi$.
 In the concluding step $\ell$, we skip the superclustering step. 
 As was shown in~\cite{ElkinNeimanHopsets},  the input set of clusters to the last phase 
 $P_{\ell}$ is sufficiently small to allow us to interconnect all the centers in $P_{\ell}$ 
 to one  another using few hopset edges. 
 
We are now ready to describe in detail, the execution of superclustering step.
The interconnection step will be described in Section~\ref{interconnectWeighted}.
\subsubsection{Superclustering}\label{sec:HopSuperclustering} 
The phase $i$ begins by sampling each cluster $C \in P_i$ 
independently at random with probability $1/deg_i$. 
Let $S_i$ denote the set of sampled clusters. 
We now have to conduct (approximate) distance exploration up to depth 
$\delta'_i$ in $G^{(k-1)}$ rooted at the set $CS_i = \bigcup_{C \in S_i} \{r_C\}$. 
By Lemma~\ref{lem:2beta +1}, this can be achieved by $2\beta + 1$ 
iterations of Bellman-Ford algorithm on the graph $G^{(k-1)}$.
For this, we invoke the approximate Bellman-Ford exploration algorithm of 
Section~\ref{sec:WeightedForest} on graph $G^{(k-1)}$ with set $CS_i$ as the set $S$ 
of source vertices and parameters $\eta = 2\beta +1$, $\zeta = \chi$.


One issue with invoking the Algorithm of Section~\ref{sec:WeightedForest} as a 
blackbox for graph $G^{(k-1)}$ is that only the edges of the input graph $G$ appear on the stream 
and the edge set $H^{(k-1)}$ of all the lower level hopsets is available offline.
We therefore slightly modify the algorithm of Section~\ref{sec:WeightedForest} 
and then invoke the modified version with $S = CS_i$, $\eta = 2\beta +1$ and $\zeta = \chi$. 
In the modified version, at the end of each pass through the stream,
 for every vertex $v \in V$, we scan through the edges incident to $v$ in the set $H^{(k-1)}$ 
 and update its distance estimate $\hat{d}(v)$ as:
$$ \hat{d}(v) = \min\{\hat{d}(v), \min_{(v,w) \in H^{(k-1)}}\{ \hat{d}(w) + \omega_{H^{(k-1)}}(v,w) \} \}.$$
The parent of $v$, $\hat{p}(v)$ is also updated accordingly.
Note that this modification does not affect the space complexity, 
stretch guarantee or the success probability of the algorithm of Section~\ref{sec:WeightedForest}. 
The upper bound on the stretch guarantee still applies
 since we update the distance estimate of a given vertex $v$ 
 only if the estimate provided by the edges in the set  $H^{(k-1)}$ 
 is better than $v$'s estimate from the stream.
 The success probability and space complexity
  are unaffected since the modification deterministically 
 updates the distance estimates and does not use any new variables.
 This provides us with a $(1+ \chi)$-approximation of $d^{(2\beta + 1)}_{G^{ (k-1) } } (v,CS_i)$, for all $v \in V$. 


Hence, by Theorem~\ref{thm: SingleBFE}, 
an invocation of modified version of approximate Bellman-Ford algorithm of Section~\ref{sec:WeightedForest} during the
 the superclustering step of phase $i$ generates whp, 
an approximate Bellman-Ford exploration of the graph $G^{(k-1)}$, rooted at the set $CS_i \subseteq V$ in $2\beta + 1$ passes. 
 It outputs for every $v \in V$ an estimate $\hatter{v}$ of its 
distance to set $CS_i$ such that:
\begin{align}\label{eq:BFEStretch}
d^{(2\beta + 1)}_{G^{(k-1) } }(v, CS_i) \le \hatter{v} \le (1 + \chi)\cdot d^{(2\beta + 1)}_{G^{(k-1) } }(v, CS_i).
\end{align}

Moreover, the set of parent variables $\hat{p}(v)$ of every $v \in V$ with $\hat{d}(v) < \infty$ span a forest $F$ of $G^{(k-1)}$ rooted
at the set of sampled centers $CS_i$. For every vertex $v$, one can compute its path to the root $r_C$ of the tree in forest $F$, 
to which $v$ belongs, through a chain of parent pointers. 
For every cluster center $r_{C'}$, $C' \in P_i \setminus S_i$, such that $\hat{d}(r_{C'}) \le \delta'_i$, 
the algorithm adds an edge $(r_C, r_{C'})$ of weight $\hat{d}(r_{C'})$ to the hopset $H_k$, 
where $r_C$ is the root of the tree in $F$ to which $r_{C'}$ belongs.
We also create a supercluster rooted at $r_C$ which contains all the vertices of $C'$ as above.
Note that if $d_G(r_C,  r_{C'} ) \le \delta_i$, then by equations~(\ref{eq:2beta +1}) and~(\ref{eq:BFEStretch}), 
$\hat{d} ( r_{C'} ) \le (1 + \chi) \cdot (1 + \epsilon_{k-1} ) d_G(r_C, r_{C'}) = \delta'_i$.
Therefore, the edge $(r_C, r_{C'})$ will be added in to the hopset and the cluster $C'$ will be superclustered into a supercluster
centered at $r_C$.


We conclude that:
\begin{lemma}\label{lem:SuperFinalHop1}
	For a given set of sampled cluster centers $CS_i \subseteq V$ and a sufficiently large constant $c$, 
	the following holds with probability at least least $1-1/n^c$:
	 
	 \begin{enumerate}
	  
	\item The superclustering step of phase $i$ creates disjoint superclusters that contain all the clusters with centers within distance $\delta_i$ (in $G$) from the set of centers $CS_i$. 
	         It does so in $2\beta + 1$ passes through the stream,  
	       using $O_{c} (\beta/\chi \cdot \log^2 n \cdot \log \Lambda (\log n + \log \Lambda ))$ space.	
	\item \label{item:Super2} For every unsampled cluster center $r_C$ within distance $\delta_i$ (in $G$) from the set $CS_i$, an edge to the nearest center $r_C' \in CS_i$
	of weight $\omega_{H_k}(r_C, r_C') \le (1 + \chi) \cdot d^{(2\beta + 1)}_{G^{(k-1)}} (r_C, r_C') \le (1 + \chi) \cdot (1 + \epsilon_{k-1}) d_G(r_C, r_C')$ 
	is added into the hopset $H_k$, \\
	where $\epsilon_{k-1}$ is the stretch guarantee of the graph $G^{(k-1)}$.
	\end{enumerate}
\end{lemma}

\input{InterconnectionWeighted}

%% file: InterconnectionWeighted.tex
\subsubsection{Interconnection}\label{interconnectWeighted}
Next we describe the interconnection step of each phase $i \in \{0,1,\ldots,\ell \}$. 
Recall that $U_i$ is the set of clusters of $P_i$ that were not superclustered in phase $i$. 
Let $CU_i$ be the set of centers of clusters in $U_i$, i.e., $CU_i =\bigcup_{C \in U_i} \{r_C\}$.
For the phase $\ell$, the superclustering step is skipped and we set $U_{\ell} = P_{\ell}$.

In the interconnection step of phase $i \ge 0$, 
we want to connect every cluster $C \in U_i$ to every other cluster $C' \in U_i$ that is close to it.
To do this, we want 
to perform $2\beta + 1$ iterations of a $(1 + \chi)$-approximate Bellman-Ford exploration
from every cluster center $r_C \in CU_i$ \emph{separately} 
 in $G^{(k-1)}$. 
These explorations are, however, conducted to a  bounded depth (in terms of number of hops), and 
to bounded distance. Specifically, the hop-depth of these explorations will be at most $2\beta +1$,
while the distance to which they are conducted is roughly $\delta_i/2$.
For every cluster center $r_{C'}$, $C' \in U_i$ within distance $\delta_i/2$ from $r_C$ in $G$, 
we want to add an edge $ e= (r_C, r_{C'})$ of weight at most 
$(1+ \chi) \cdot d^{(2\beta + 1)}_{G^{(k-1)}} (r_C, r_C')$ to the hopset $H_k$.
To do so, we turn to the stream to find an estimate of 
$d^{(2\beta +1)}_{G^{(k-1)}} (v,r_C)$ for every $v \in V$ and every center $r_C \in U_i$. 
As discussed in the construction of spanners, 
we cannot afford to invoke the algorithm of Section~\ref{sec:WeightedForest} multiple times in parallel 
to conduct a separate exploration from every center $r_C$ in $CU_i$,
due to space constraints. (See Section~\ref{interconnect} for more details.)
As shown in~\cite{ElkinNeimanHopsets} (See Lemmas 3.2 and 3.3 of~\cite{ElkinNeimanHopsets}), 
Lemma~\ref{lem:ExpCount} holds in the 
interconnection step of (a single-scale) hopset construction as well.
Specifically, if one conducts Bellman-Ford explorations to depth at most $\delta'_i/2$ in \hopgraph
to hop-depth at most $2\beta +1$, then, with high probability, every vertex is traversed by at most $O(deg_i \ln n)$
explorations.

Therefore, we adapt the randomized technique of Section~\ref{interconnect} to efficiently identify 
for every $v \in V$, the sources of all the explorations it gets visited by in phase $i$. 
Moreover, for every vertex $v \in V$ with a non-empty subset $U^v_i \subseteq U_i$ of explorations that visit $v$, 
we find for every cluster $C  \in U^v_i$, an estimate of  $d^{(2\beta + 1)}_{G^{(k-1)}} (v, r_C)$.
Note, however, that not all the edges of the graph $G^{(k-1)}$ on which 
we have to perform our Bellman-Ford
explorations are presented on the stream. 
We show in the sequel, how we adjust the distance estimates of every vertex $v \in V$ 
by going through the edges of the lower level hopsets 
$H^{(k-1)}$ offline.  

Throughout the interconnection step of phase $i$, 
we maintain for every vertex $v \in V$, 
a set $LCurrent_v$ (called \emph{estimates list} of $v$) 
of sources of Bellman-Ford explorations that visited $v$ so far.
Each element of $LCurrent_v$ is a tuple 
$(s, \hat{d}(v,s) )$, where $s$ is the center of some cluster in $U_i$,
and $\hat{d}(v, s)$ is the current estimate of $d^{(2\beta + 1)}_{G^{(k-1)}} (v, s)$.
For any center $s' \in CU_i$, for which we do not yet have a tuple in $LCurrent_v$, 
$\hat{d}(v, s')$ is implicitly defined as $\infty$.
Initially, the estimates lists  of all the vertices are empty, except for the centers of clusters in $U_i$. 
The estimates list of every center $r_C \in CU_i$ is initialized with a single element $(r_C, 0)$ in it. 
The interconnection step of phase $i$ is carried out in  $2\beta + 1$ sub-phases.  
In the following section, we describe the purpose of each of the  $2\beta + 1$ sub-phases
 of the interconnection step and the way they are carried out.

\textbf{Sub-phase $p$ of interconnection step:}\label{sec:interSubPhaseWeighted} 
Denote $\zeta' = \frac{\chi}{ 2\cdot(2\beta + 1)}$. 
Our goal is to ensure that by the end of sub-phase $p$, 
for every vertex $v \in V$ and every exploration source $s \in CU_i$ 
with a $p$-bounded path to $v$ in \hopgraph,
there is a tuple $(s, \hat{d}(v,s))$ in the estimates list $LCurrent_v$ such that:
 $$d^{(p)}_{G^{ (k-1) } } (v,s) \le \hat{d} (v,s) \le (1 + \zeta')^p \cdot  d^{(p)}_{G^{ (k-1) } } (v,s).$$  
To accomplish this, in every sub-phase $p$, we search for every vertex $v \in V$, 
a \emph{better} (smaller than the current value of $\hat{d}(v,s)$) estimate (if exists) 
of its $(2\beta+1)$-bounded distance to every source $s \in CU_i$, 
by keeping track of edges $e = (u, v)$ incident to $v$ in $G^{(k-1)}$.
In each of the $2\beta + 1$ sub-phases, 
 we make two passes through the stream. 
 For a given vertex $v \in V$, an exploration source $s \in CU_i$ 
 is called an \emph{update candidate} of $v$ in sub-phase $p$,
 if a better estimate
of \limitedapprox{v}{s} is available in sub-phase $p$ through some edge $e = (u,v)$ on the stream.
(Recall that the current estimate of \limitedapprox{v}{s'} for some source $s' \in CU_i$ 
for which we do not yet have an entry in $LCurrent_v$ is $\infty$.)
Note that a better estimate of $\limitedapprox{v}{s}$, for some vertex $v$ and some source $s$
 in sub-phase $p$,
 may also be available through some edges in $H^{(k-1)}$.
We therefore go through the edge set $H^{(k-1)}$ offline
  at the end of every sub-phase and update  all our estimates lists 
with the best available estimates in $H^{(k-1)}$.
 

 In the first pass of sub-phase $p$, 
 we identify for every $v \in V$, all of $v$'s update candidates in sub-phase $p$.
All of these update candidates are added to a list called the \emph{update list} of $v$, denoted $LUpdate_v$.
 Each element of $LUpdate_v$ is a tuple $(s, range, r)$, where $s$ is the ID of an exploration source 
in $CU_i$ for which a better estimate of \limitedapprox{v}{s} is available, 
$range$ is the distance range $I =(low, high]$
in which the better estimate is available, and 
$r$ is the number of vertices $u \in \Gamma_{G}(v)$, 
such that $\hat{d} (u, s) + \omega(u,v)  \in range$.

The second pass of sub-phase $p$
uses the update list of every vertex $v \in V$
to find a better estimate of $\limitedapprox{v}{s}$, 
for every update candidate $s$ in $LUpdate_v$.
The new better estimate of $\limitedapprox{v}{s}$ for every source $s$ in $LUpdate_v$ 
is then used to update the estimates list $LCurrent_v$ of $v$.

\textbf{First pass of sub-phase $p$ of phase $i$:~}
By Lemma~\ref{lem:ExpCount}, the number of explorations that visit a vertex $v \in V$ during the
interconnection step of phase $i$ is at most $\deg_i$ in expectation 
and at most $c'_1\cdot \ln n \cdot deg_i$ whp, 
where $c'_1$ is a sufficiently large positive constant. 
Hence, the number of  update candidates of $v$ in 
any sub-phase of interconnection step of phase $i$ is at most $c'_1\cdot \ln n \cdot deg_i$ whp.
(Recall that all the explorations are restricted to distance at most $\delta'_i/2$.)

As in  Section~\ref{sec:interSubPhase}, 
we denote $\mathcal{N}_i = c'_1\cdot \ln n \cdot deg_i$ and 
$\mathcal{\mu}_i = 16\cdot c_4 \cdot \mathcal{N}_i \cdot \ln n$, 
where $c_4 \ge1$ is a sufficiently large positive constant.

At a high level, in the first pass of every sub-phase, 
we want to recover, for every vertex $v \in V$, a vector (containing sources of explorations 
that visit $v$ in sub-phase $p$) with at most $\mathcal{N}_i$ elements in its support.
In other words, we want to perform an $s$-sparse recovery for every vertex $v \in V$, where $s = \mathcal{N}_i$.
In the unweighted case in Section~\ref{sec:interSubPhase}, 
we perform $\mathcal{N}_i$-sparse recovery for a given vertex $v$ by multiple simultaneous 
invocations of a sampler \emph{FindNewVisitor} 
that samples (with at least  a constant probability) 
one exploration source out of at most $\mathcal{N}_i$ sources that visit $v$.
In the weighted case, we do something similar
 but with a more involved sampling procedure called \emph{FindNewCandidate}. 
 The pseudocode for procedure \emph{FindNewCandidate}
is given in Algorithm~\ref{alg:findNewCandidate}.
The procedure \emph{FindNewCandidate} enables us to sample an update candidate $s$
of $v$ (if exists), with a
better (than the current) estimate of \limitedapprox{v}{s} in a
a specific distance range. 

\begin{algorithm}[h!]
   \caption{Pseudocode for procedure $FindNewCandidate$}
  	\label{alg:findNewCandidate}
 	\begin{algorithmic}[1]
 	  \State $\textbf{Procedure FindNewCandidate} (v,h, I)$\\
 	  \Comment{Initialization}
     	\State $\emph{slots} \leftarrow \emptyset$ \hspace{0.45in} 
	\Comment{An array with $\lambda = \ceil{\log~n}$ elements indexed \\ \hspace{1.5in} from  $1$ to $\lambda$. }

	\LeftComment{Each element of slots is a tuple $(sCount, sNames)$.
	                   For a given index $1 \le k \le \lambda$, fields $sCount$ and 
	                   $sNames$ of $\emph{slots}[k]$ can be accessed as 
	                   $\emph{slots}[k].sCount$ and $\emph{slots}[k].sNames$, respectively.} \\
     	\LeftComment{$\emph{slots}[k].sCount$ counts the new update candidates seen by $v$ with 
	                     hash values in $[2^k]$. It is set to $0$ initially.}
          \LeftComment{$\emph{slots}[k].sNames$ is an encoding of the names of candidate
                              sources seen by $v$ with hash values in $[2^k].$ It is set to $\phi$ initially.}

  \Comment{Update Stage}
   \While{$(\text{there is some update~} (e_t,~eSign_t,~eWeight_t) \text{~in the stream})$} \label{alglin:updateStart} 
        \If{ $ (  e_t \text{~is incident on~} v \text{~and some~} u \in V ) $  } \label{alglin:checkCandidates1}
             	 \ForEach{$( s, \hat{d}(u,s) ) \in LCurrent_u$} \label{alglin:checkCandidates2}
	                  \If{ $ (( \hat{d}(u,s) + eWeight_t ) \in I $ \text{~and~} \\ \hspace{1.0in}
	                            $ \hat{d}(u,s) + eWeight_t < \hat{d}(v,s) ) $ }  \label{alglin:checkCandidates3}
              		 	\State$ k \gets \ceil{\log h(s)}$
               	           	   \Repeat  \Comment{Update $\emph{slots}[k]$ for all $\ceil{\log h(s)} \le k \le \lambda$}
	                         		\State $\emph{slots}[k].sCount \gets \emph{slots}[k].sCount + eSign_t $ \label{alglin:counteraddtnC}
               				\State $\emph{slots}[k].sNames \gets \emph{slots}[k].sNames + \nu(s)\cdot eSign_t$ \label{alglin:vectoraddtnC} \\  
               		        	         \Comment{The function $\nu$ is described in Section~\ref{sec:Encodings}}.\\
		                           \Comment{The addition in line~\ref{alglin:vectoraddtnC} is a vector addition.}
		                   
%
%

               			 \State $k = k+1$
               		   \Until{$k > \lambda$}
		   
		        \EndIf
                     	   \EndFor
	                \EndIf
     \EndWhile \label{alglin:updateEnd} 
     
      \Comment{Recovery Stage}
 	 \If{$(\emph{slots} \text{~vector is empty})$}  
    		 \State return $(\phi, \phi)$ 

  	 \ElsIf{$(\exists \text{~index~} k~\text{s.t.}~\frac{\emph{slots}[k].sNames}{\emph{slots}[k].sCount} = \nu(s) \text{~for some $s$ in $V$} )$} \label{alglin:singleCheck}
        		\State  return $(s,\emph{slots}[k].sCount )$
     	\Else
           	\State return $(\perp, \perp)$ 
   	\EndIf
 \end{algorithmic}
  \end{algorithm}

For every vertex $v \in V$, we divide the possible range of better estimates of $v$'s $(2\beta +1)$-bounded 
distances to its update candidates,
into sub-ranges on a geometric scale. 
 We then invoke the procedure \emph{FindNewCandidate} repeatedly in parallel to 
perform an $\mathcal{N}_i$-sparse recovery for $v$ on every sub-range.
Specifically, we divide the search space of potential better estimates, $\left[1,\delta'_i/2 \right]$, into sub-ranges
$I_j = \osubrange{\zeta'}{j}$, for $j \in \{0,1,\ldots, \gamma \}$, where $\gamma = \ceil{\log_{1 + \zeta'} \delta'_i/2} -1 $. 
For $j=0$, we make the sub-range $I_0 = \left[(1+ \zeta')^{0}, (1+ \zeta')^{1} \right]$
closed to include the value $1$.
Note that we are only interested in distances at most $\delta'_i/2$.
Therefore we restrict our search for distance estimates to the range $[1,\delta'_i/2 ]$,
as opposed to the search range $[1, \Lambda]$ that we had in Section~\ref{sec:algo}.\\

In more detail, we make for  for each $v \in V$ 
and for each sub-range $I_j$, 
$\mu_i$ \emph{attempts} in parallel. 
In a specific attempt for a given vertex $v$ and a given sub-range $I_j$, 
we make a single call to procedure 
\emph{FindNewCandidate} which samples an update candidate $s$ (if exists) of $v$
with a better estimate of  \limitedapprox{v}{s}
 in the sub-range $I_j$. 
 Henceforth, we will refer to an update candidate $s$ of a vertex $v$
 with a better estimate of \limitedapprox{v}{s} in 
 a given distance range $I$,
  as the \emph{update candidate of $v$ in the range} $I$.
 
 The procedure \emph{FindNewCandidate} can be viewed as an adaptation of procedure \emph{FindNewVisitor}
  from Section~\ref{sec:interSubPhase} for weighted graphs. 
 It takes as input the ID of a vertex, 
 a hash function $h$ chosen at random from a family of pairwise independent hash functions 
 and an input range $I = (low, high]$. (The input range may be closed as well.)
A successful invocation of \emph{FindNewCandidate} for an input vertex $v$ 
 and a distance range $I$ returns a tuple $(s, c_s)$, 
 where $s$ is the ID of an update candidate of $v$ in the range $I$, 
 and $c_s$ is the number of edges $(v,u) \in E$ such that $\hat{d}(u,s) + \omega(v,u) \in I$.
 If there is no update candidate of $v$ in the input range $I$, 
 procedure \emph{FindNewCandidate} returns a tuple $(\phi, \phi)$.
 If there are update candidates of $v$ in the input range, 
 but procedure \emph{FindNewCandidate} fails to isolate an ID of such a candidate, 
 it returns  $(\perp, \perp)$.
 
 Before we start making our attempts in parallel, 
 we sample uniformly at random a set of functions $H_{p}$ ($|H_p| =  \mu_i$)
 from a family of pairwise independent hash functions 
  $h : \{1, \dots, maxVID \} \rightarrow \{1,\dots,2^{\lambda}\}$, where $\lambda = \ceil{\log maxVID} = \ceil{\log n}$.
  Then, for every vertex $v \in V$ 
  and every distance sub-range $I_j$, $j \in \{0,1,\ldots, \gamma \}$,
  we make $\mathcal{\mu}_i$  parallel calls to procedure 
   $\emph{FindNewCandidate}(v,h, I_j )$, one call for each $h \in H_p$.

  \textbf{Procedure FindNewCandidate:} As mentioned above, 
  the procedure \emph{FindNewCandidate} is similar to procedure \emph{FindNewVisitor} 
  (See Algorithm~\ref{alg:findNewVisitor}) of Section~\ref{interconnect}. 
  It uses a function $h$ chosen uniformly at random from a family of pairwise independent hash functions 
  to sample for the input vertex $v$, an update candidate of $v$ in the  
 input range $I$.
 Just like procedure \emph{FindNewVisitor}, it also uses 
 the CIS-based encoding scheme $\nu$ described in Section~\ref{sec:Encodings}
 to encode the names of the exploration sources it samples, 
 and uses Lemma~\ref{lem:convex} to check 
 (See line~\ref{alglin:singleCheck} of Algorithm~\ref{alg:findNewCandidate}), 
 if it has successfully isolated the ID of a single update candidate in the desired distance range.
We will mainly focus here on the details of Algorithm~\ref{alg:findNewCandidate}
 which are different from that of Algorithm~\ref{alg:findNewVisitor}.
We refer the reader to Sections~\ref{sec:interSubPhase} and~\ref{sec:Encodings} 
for a detailed exposition of our sampling technique and the CIS-based encoding scheme.

The procedure \emph{FindNewCandidate} (Algorithm~\ref{alg:findNewCandidate}) 
differs from procedure \emph{FindNewVisitor} (Algorithm~\ref{alg:findNewVisitor})
mainly in its input parameters and 
its handling of the incoming edges during the \emph{Update Stage}. (See lines~\ref{alglin:updateStart} to~\ref{alglin:updateEnd}.)
Specifically, procedure \emph{FindNewCandidate} takes an additional input parameter $I$ corresponding to
a range of distances. 
It looks for an update candidate of input vertex $v$ in the input range  $I$.
The update stage of a call to procedure \emph{FindNewCandidate} for an input vertex $v$
 and an input distance range $I$ proceeds as follows.
For every update $(e_t,~eSign_t,~eWeight_t)$ to an edge $e_t$ incident to $v$ and some vertex $u$, 
we look at every exploration source $s$ in the estimates list $LCurrent_u$ of $u$, 
(see line~\ref{alglin:checkCandidates2} of Algorithm~\ref{alg:findNewCandidate})
and check whether 
 the distance estimate of $v$ to $s$ via edge $e_t = (v,u)$ is better than the current value of $\hat{d}(v,s)$,
 and whether it falls in the input distance range $I$.
(See line~\ref{alglin:checkCandidates3} of Algorithm~\ref{alg:findNewCandidate}.)
If this is the case, then, we sample $s$ just like we sample new exploration sources in \emph{FindNewVisitor}.
This completes the description of procedure~\emph{FindNewCadidate}.

As in procedure \emph{FindNewVisitor}, by Corollary~\ref{cor:pairwise},
a single call to procedure \emph{FindNewCandidate} succeeds with at least a constant probability. 

For a vertex $v \in V$, if there are no update candidates of $v$ in sub-phase $p$, 
all the calls to procedure \emph{FindNewCandidate} in all the attempts return $(\phi, \phi)$. 
For every such vertex, we do not need to add anything to its update list $LUpdate_v$. 
 At the end of the first pass, if no invocation of procedure \emph{FindNewCandidate} returns as error,
we extract for every vertex $v \in V$ and every distance range $I_j$ ($j \in \{0,1,\ldots, \gamma\}$),
all the distinct update candidates of $v$ in the range $I_j$  
sampled by $\mu_i$ attempts made for $v$ and sub-range $I_j$.
For a given update candidate $s$ of $v$, 
let $j = j_{v,s}$ be the smallest index in  $\{0,1,\ldots, \gamma \}$, 
such that a tuple $(s, c_s)$ (for some $c_s > 0$) is returned by a call to procedure
$\emph{FindNewCandidate} (v, h, I_j ) $.
We add a tuple $(s, I_j, c_s)$ to the list of
 update candidates $LUpdate_v$ of $v$.
 Recall that the set $LUpdate_v$ of vertex $v$ contains tuples $(s, range, r_s)$, 
 where $s$ is the ID of an update candidate of $v$,
$range$ is the distance range
in which a better estimate of \limitedapprox{v}{s} lies, and 
$r$ is the number of edges $(u,v) \in \Gamma_{G}(v)$ such that $\hat{d}(u,s) + \omega(u,v) \in range$.

\textbf{Analysis of first pass}: We now analyze the success probability and space requirements of the first pass of sub-phase $p$ of interconnection step.
Recall that, in sub-phase $p$, for every vertex $v \in V$ and every distance sub-range
$I_j = \osubrange{\zeta'}{j}$ ($j \in \{0,1,\ldots, \gamma \}$, where $\gamma = \ceil{\log_{1 + \zeta'} \delta'_i/2} -1$), 
we make $\mathcal{\mu}_i = 16 \cdot c_4 \cdot \mathcal{N}_i \cdot \ln n$ parallel attempts 
or calls to procedure \emph{FindNewCandidate} to isolate all the update candidates of $v$ in the range $I_j$.

We first show that making $\mathcal{\mu}_i = 16\cdot c_4 \cdot \mathcal{N}_{i} \cdot \ln n$ attempts in parallel for a given vertex $v \in V$
and a given distance range $I_j$, $j \in \{0,1,\ldots,  \gamma \}$,
 ensures that 
a specific update candidate of vertex $v$ in a specific distance range $I$ in sub-phase $p$ is extracted whp.

\begin{lemma} \label{lem:findNProbW}
	 For a given vertex $v \in V$ and a specific distance sub-range $I_j$, during sub-phase $p$,
	a given update candidate $s$ of $v$ in the range $I_j$ is discovered with probability at least $1- 1/n^{c_4}$.
	\end{lemma}
\begin{proof}
        Let $d^{(p,j)}_v$ be the number of update candidates of $v$ in the range $I_j$ in sub-phase $p$.
 	By Lemma~\ref{lem:ExpCount}, with probability at least $1 - \frac{1}{n^{c'_1 -1}}$,
	the number of the exploration sources that visit $v$ during interconnection step of phase $i$ is at most $\mathcal{N}_i$. 
	Observe that $\mathcal{N}_i$ is an upper bound on the number of update candidates of $v$ (over the entire distance range $[1, \delta'_i/2]$) during sub-phase $p$.
	 It follows therefore that $d^{(p,j)}_v \le \mathcal{N}_i$.
	  For a specific update candidate $s$ of $v$ in the range $I_j$ in sub-phase $p$, let $DISC^{(s)}$ 
	denote the event that it is discovered in a specific attempt.
	Then:
	
	\begin{equation*}
		\begin{aligned}
			Pr \left[DISC^{(s)} \right] &\ge  Pr\left[DISC^{(s)}~|~d^{(p,j)}_v \le \mathcal{N}_i \right] \cdot Pr\left[d^{(p,j)}_v \le \mathcal{N}_i \right] \\
			&\ge 	Pr\left[DISC^{(s)}~|~d^{(p,j)}_v \le \mathcal{N}_i \right] \cdot \left(1 - \frac{1}{n^{c_{1} -1}} \right) \\
			&\ge \frac{1}{8 \mathcal{N}_i } \left( 1 - \frac{1}{n^{c'_1 -1}}\right) \\
			&\ge \frac{1}{16 \mathcal{N}_i}
		\end{aligned}	
	\end{equation*}
Note that the third inequality follows by applying Lemma~\ref{lem:Pairwise} to the event  $\{DISC^{(s)}~|~d^{(j)}_v \le \mathcal{N}_i \}$.	

Thus, for a given update candidate of $v$ in the sub-range $I_j$, 
the probability that none of the $\mu_i = 16\cdot c_4 \cdot \mathcal{N}_{i} \cdot \ln n$ attempts
 will isolate it is at most $\left (1 - \frac{1}{16 \mathcal{N}_i} \right)^{16 \cdot c_4 \cdot \ln n \cdot \mathcal{N}_i} \le 1/n^{c_4}$.
\end{proof}

Next, we analyze the space requirements of procedure \emph{FindNewCandidate}.
Procedure \emph{FindNewCandidate} is similar to 
procedure \emph{FindNewVisitor} of Section~\ref{sec:interSubPhase} in terms of its sampling technique. 
In addition to all the variables that procedure \emph{FindNewVisitor} uses, 
procedure \emph{FindNewCandidate} also uses distance variables $low$ and $high$, 
 that define the input range $I = (low, high]$, in which it looks for an update candidate of its input vertex. 
Each of these distance variables consume $O(\log \Lambda)$ bits.
Adding the cost of additional variables used in procedure \emph{FindNewCandidate}
to the space usage of procedure \emph{FindNewVisitor} (Lemma~\ref{lem:findNewVisitorSpace}),
we get the following lemma:

\begin{lemma}\label{lem:findNewCandidateSpace}
	The procedure FindNewCandidate uses $O(\log^2 n + \log \Lambda)$ bits of memory. \\
\end{lemma}

We next provide an upper bound on the space usage of the first pass of the interconnection step.
\begin{lemma}\label{lem:firstPassInterWSpace}
	The overall space usage of the first pass of every sub-phase of interconnection is \\
	$$O(n^{1 + \rho}\cdot \frac{ \log \Lambda}{\zeta'} \cdot \log^2 n \cdot (\log^2 n + \log \Lambda)) \text{~bits.}$$
\end{lemma}
\begin{proof}
	The first pass of every sub-phase makes \\
	$\gamma \cdot \mathcal{\mu} _i =  (\ceil{\log_{1 + \zeta'} \delta'_i/2} -1) \cdot  \mathcal{\mu} _i =  O(\log_{1 + \zeta'} \Lambda \cdot deg_i \cdot \log^2 n)$
	attempts in parallel for every $v \in V$. 
	Recall that for all $i$, $deg_i \le n^{\rho}$ (See Section~\ref{sec:Over}). 
	Combining this fact with Lemma~\ref{lem:findNewCandidateSpace}, 
	we get that the space usage of all the invocations of procedure \emph{FindNewCandidate} 
	for all the $n$ vertices during the first pass is $O(n^{1 + \rho}\cdot \log_{1 + \zeta'} \Lambda \cdot \log^2 n \cdot  (\log^2 n + \log \Lambda))$.	
	We use $ |H_p| = \mathcal{\mu}_i$ hash functions during the first pass.
	Each hash function can be encoded using $O(\log n)$ bits.
	The overall space used by the storage of hash functions during the first phase is thus $O(n^{\rho} \cdot \log^3 n)$. 
	As an output, we produce an update list $LUpdate_v$ for every $v \in V$. 
	Each of these update lists  consists of tuples $(s, range, r)$ of $O(\log n + \log \lambda)$ bits each. 
	By Lemma~\ref{lem:ExpCount}, a vertex $v$ is visited by at most $O(deg_i \log n)  \le O(n^{\rho} \log n)$ explorations whp in the phase $i$.  
	In any case, we record just $O(n^{\rho} \cdot \log n)$ of them, even if $v$ is visited by more explorations.
	Hence, the storage of all the update lists during a given sub-phase  requires $O(n^{1 + \rho} \log n (\log n + \log \lambda))$ bits. 
	Finally, we need to store the estimates lists $LCurrent_v$ of all $v \in V$. This requires at most $O(n^{1 + \rho} \log n(\log n + \log \Lambda))$ bits of space.
	Thus, the storage cost of first pass of every sub-phase is dominated by the cost of parallel invocations of procedure  \emph{FindNewCandidate}. 
	This makes the overall cost of first pass of every sub-phase 
	$$O(n^{1 + \rho}\cdot \frac{ \log \Lambda}{\zeta'} \cdot \log^2 n \cdot (\log^2 n + \log \Lambda))\text{~bits.}$$
\end{proof}

\textbf{Second pass of sub-phase $j$ of phase $i$:~} The second pass of sub-phase $p$ starts with the update lists $LUpdate_v$ of every $v \in V$. 
Recall that the update list $LUpdate_v$  of a given vertex $v \in V$ consists of tuples of the form $(s, range, r)$,
where $s$ is an exploration source in $CU_i$ for which a better estimate of \limitedapprox{v}{s} is available in the distance sub-range $range$,
and $r$ is the number of edges in the edge set $E$ of the original graph $G$ through which the better estimate is available.
We find for every tuple $(s, range ,r)$ in $LUpdate_v$, a better estimate of \limitedapprox{v}{s} in the sub-range $range$, 
by invoking procedure  \emph{GuessDistance}  (described in Section~\ref{sec:GuessDistance}) 
$O(\log n)$ times. 

We sample uniformly at random a set of $c_1 \log_{7/8} n$ pairwise independent hash functions $H'_{p}$ 
from the family $h : \{1,\dots,\emph{maxVID} \} \rightarrow \{1,2,\ldots,2^{\lambda}\}$ ($\lambda = \ceil{\log n}$), 
to be used by invocations of procedure \emph{GuessDistance}. 

 We need to change the original procedure \emph{GuessDistance} (Section~\ref{sec:GuessDistance}) slightly to work here. 
Specifically, we need to change the part where we decide whether to sample an incoming edge update or not (Line~\ref{alglin:check} of Algorithm~\ref{algps:guessDistance}). 
It should be updated to check if the edge $e_t$ is incident between the input vertex $v$ and some vertex $u$ 
such that there is a tuple $(s, \hat{d}(u,s))$ in the estimates list of $u$ and that $(\hat{d} (u,s) + eWeight_t) \in range$ and $\hat{d} (u,s) + eWeight_t < \hat{d}(v,s)$.
Note that the current estimate $\hat{d}(v,s)$ of input vertex $v$'s distance to its update candidate $s$ is either available in its estimates list $LCurrent_v$ or 
is implicitly set to $\infty$. The latter happens if $v$ has not yet been visited
by the exploration rooted at source $s$.

 
 At the end of the second pass, we have the results of all the invocations of procedure \emph{GuessDistance}, 
for a given vertex $v$ corresponding to the tuple $(s, range, r) \in LUpdate_v$.
We update the corresponding tuple $(s, \hat{d}(v,s))$ in the estimates list $LCurrent_v$ of $v$ 
with the minimum value returned by any invocation of \emph{GuessDistance} for vertex $v$.
If an entry corresponding to $s$ is not present in the estimates list $LCurrent_v$ at this stage 
(i.e., $\hat{d}(v,s) = \infty$ as above),
then we add a new tuple to the estimates list of $v$.
Finally, the updates lists of all the vertices are cleared to be re-used in the next sub-phase.
So far, we have only looked at the edges of the original graph presented to us in the stream while looking for better estimates of $\limitedapprox{v}{s}$.
Recall that we need to perform $2\beta + 1$ iterations of the Bellman-Ford algorithm in the graph $G^{(k-1)}$ which is a union of the original graph 
$G$ and $H^{(k-1)} = \bigcup_{\floor{\log \beta}\le~j~\le k-1} H_j$ of all the lower level hopsets.
Having updated the estimates lists of all the vertices with the best estimate available from the stream,
at the end of second pass of sub-phase $p$ we go through the edges  of the lower
level hopsets and check for each $v \in V$ whether a better estimate of $\limitedapprox{v}{s}$ for any source $s \in CU_i$ is available
  through one of the hopset edges. If this is the case, then we update the estimates lists accordingly.
  
\textbf{Analysis of Second Pass:~} We now analyze the space requirements of the second pass of sub-phase $j$ of interconnection step.
\begin{lemma}\label{lem:secondPassInterWSpace}
	The overall space usage of the second pass of every sub-phase of the interconnection step is  $O(n^{1 + \rho} \cdot \log^3 n \cdot ( \log n + \log \Lambda) )$. 
\end{lemma}
\begin{proof}
	The second  pass of every sub-phase invokes procedure \emph{GuessDistance}  
	$O(\log n)$ times in parallel for every tuple in the update list $LUpdate_v$ of every $v \in V$. 
	The number of elements in the update list $Lupdate_v$ of a given vertex $v$ is at most $O( n^{\rho} \log n)$.
	(Recall that by Lemma~\ref{lem:ExpCount}, whp there are at most $\tilde{O}(n^{\rho})$ explorations per vertex.
	But even if there are more explorations,
	our algorithm records just $\tilde{O}(n^{\rho})$ explorations per vertex.)
	Therefore, we make a total of $O(n^{1 + \rho} \cdot \log^2 n)$ calls to procedure \emph{GuessDistance} during the second pass of any sub-phase.
	By Lemma~\ref{lem:guessDSpace}, each invocation of procedure \emph{GuessDistance}  uses $O(\log n \cdot (\log n +  \log \Lambda))$ bits of space.
	Therefore the overall cost of all the invocations of procedure \emph{GuessDistance} is $O(n^{1 + \rho} \cdot \log^3 n \cdot ( \log n + \log \Lambda))$. 
 	In addition, we need to store a set of $O(\log n)$  hash functions of size $O(\log n)$ each in global storage. 
 	This requires $O(\log^2 n)$ bits of space.
	Therefore, the overall storage cost  of the second pass of any sub-phase is dominated by the space required for invocations of \emph{GuessDistance}.
 	Hence the overall space requirement of second pass of  interconnection is $O(n^{1 + \rho} \cdot \log^3 n \cdot (\log n + \log \Lambda) )$.
\end{proof}
Recall that $\zeta' = \frac{\chi}{2 \cdot(2\beta+1)}$, 
$|H'_p| = c_1 \log_{8/7} n$
and $\mu_i = c_4 \cdot \ln n \cdot \deg_i$, where $c_1,c_4 > 0$ are positive constants. 

\begin{lemma}\label{InterWCorrect1}
For a sufficiently large constant $c'$,  with probability at least $1 - p/n^{c' -1}$,
after $p$ sub-phases of the interconnection step of phase $i$, 
the following holds
 for a given cluster $C \in U_i$ and for every vertex $v$ within $p$ hops from the center $r_C$ of $C$ in $\hopgraph$ :
 
There is a tuple $(r_C, \hat{d}(v, r_C))$ in the estimates list $LCurrent_v$ of $v$ such that
 
$$d^{(p)}_{G^{(k-1)}} (v, r_C) \le \hat{d}(v, r_C) \le (1 + \zeta')^p \cdot d^{(p)}_{G^{(k-1)}} (v, r_C)$$

(The left-hand inequality holds with probability $1$, and the right-hand inequality holds with probability 
  at least $1 - p/n^{c' -1}$.)
\end{lemma}
\begin{proof}
The proof follows by induction on the number of phases, $p$, of the algorithm.
The base case for $p=0$ holds trivially.
For the inductive step, we assume that the statement of the lemma holds for $p = t$,
 for some $0 \le t < 2\beta +1$, and 
prove it for $p = t+1$.
Let $v$ be a vertex with a $(t+1)$-bounded shortest path to $r_C$ in $\hopgraph$.
Denote by $u \in \Gamma_G(v)$, the neighbour of $v$ 
on a shortest $(t+1)$-bounded path between $v$ and $r_C$.
By inductive hypothesis, 
with probability at least $1- t/n^{c' -1}$, every vertex with a $t$-bounded shortest path to 
$r_C$ has a tuple for $r_C$ in its estimates list 
and the corresponding estimate provides a stretch at most $(1+ \zeta')^t$.
In particular, there is a tuple $(r_C, \hat{d}(u, r_C))$ in the estimates list $LCurrent_u$ of $u$ such that
 $d^{(t)}_G(u, r_C) \le\hatter{u, r_C} \le (1+ \zeta')^t \cdot d^{(t)}_G(u,r_C)$. 
Denote by $j = j_v$ the index of a sub-range such that
 $$\hatter{u, r_C} + \omega(u,v) \in I_{j}.$$
 In the first pass of sub-phase $t+1$, 
 we make $\mu_i$ attempts in parallel to identify all the update candidates of $v$
 in the distance range $I_j$. By Lemma~\ref{lem:findNProbW}, 
 $r_C$ will be sampled in one of the $\mu_i$ attempts,
  with probability at least at least $1 - 1/n^{c_4}$.
 In the second pass of sub-phase $t+1$, we make $O(\log n)$ calls
 to procedure \emph{GuessDistance} to find an  estimate of
 $v$'s $(t+1)$-bounded distance to the center $r_C$ in the sub-range $I_j$.
 By Lemma~\ref{lem:guessDProb}, with probability at least $1- 1/n^{c_1}$,
  at least one of the calls to procedure 
 \emph{GuessDistance} will successfully return an estimate of  
 $d^{(t+1)}_G(v, r_C)$ in the sub-range $I_j$.
 By a union bound over the failure probability of the first two passes for vertex $v$,
we get that for an appropriate constant $c'$, with probability at least,
$1-1/n^{c'}$, vertex $v$ will be able to find an estimate of $d^{(t+1)}_G(v, r_C)$
in the sub-range $I_j$.
 By union bound over all the vertices with a $(t+1)$-bounded shortest path to $r_C$,
we get that with probability at least $1-1/n^{c' -1}$, all the vertices with a $(t+1)$-bounded shortest path to $r_C$ 
will be able to find an estimate of their $(t+1)$-bounded distance to $r_C$ in the appropriate sub-range.
The overall failure probability of phase $t+1$ 
 is therefore at most  $1/n^{c' -1}$ plus $t/n^{c' -1}$ from 
the inductive hypothesis. 
In total, the failure probability is at most $\frac{t+1}{n^{c' -1}}$, as required.
We assume henceforth that every vertex will successfully find an estimate
of its $(t+1)$-bounded distance to $r_C$ in the appropriate sub-range.

For a given vertex $v$, during the second pass of sub-phase $t+1$, we sample a candidate neighbour $u' \in \Gamma_G(v)$ 
such that
 $\hatter{u'} + \omega(u',v) \in  I_{j}$.

By induction hypothesis, vertex $u$ has a tuple $(r_C, \hatter{u, r_C})$ in its estimates list such that,
$\hatter{u, r_C} \le (1+\zeta')^t \cdot d^{(t)}_G(u,r_C)$.
Therefore,
\begin{equation*}
\begin{aligned}
\hatter{u, r_C} + \omega(u,v) &\le (1+\zeta')^t \cdot d^{(t)}_G(u,r_C) + \omega(u,v)\\
&\le (1+\zeta')^t \cdot (d^{(t)}_G(u,r_C) +  \omega(u,v))\\
&= (1+\zeta')^t \cdot d^{(t+1)}_G(v, r_C).
\end{aligned}
\end{equation*}

Moreover, 
$(\hatter{u', r_C} + \omega(u',v))$ and $(\hatter{u, r_C} + \omega(u,v))$ belong to the same sub-range $I_j$, and thus,
\begin{equation*}
\begin{aligned}
\hatter{u', r_C} + \omega(u', v) &\le (1 + \zeta') \cdot (\hatter{u, r_C} + \omega(u,v)) \\
 &\le (1+ \zeta')^{t+1} \cdot d^{(t+1)}_G(v, r_C).
\end{aligned}
\end{equation*} 

Finally, any update made to $\hatter{v, r_C}$ offline at the end of the sub-phase $t+1$ 
does not increase the stretch, since we update $\hatter{v, r_C}$ only if there is a smaller 
estimate available through some edges in $H^{(k-1)}$.

For the lower bound, let $i \le j$ be the minimum index such that we
succeed in finding a neighbour $u'_i$ of $v$ with $(\hatter{u'_i, r_C} + \omega(u'_i ,v)) \in I_i$.
Then, with probability $1$, $\hatter{u'_i, r_C} \ge d^{(t)}_{G} (u'_i, v)$ and thus,
\begin{equation*}
\begin{aligned}
\hatter{v, r_C} = \hatter{u'_i, r_C} + \omega(u'_i, v)  \ge d^{(t)}_G(u'_i, r_C) + \omega(u'_i, v) \ge d^{(t+1)}_G(v,r_C).
\end{aligned}
\end{equation*}
\end{proof}
Observe that Lemma~\ref{InterWCorrect1} implies that
for some $p \ge 1$, a single $(1+ \chi)$-approximate Bellman-Ford exploration to hop-depth $p$,
rooted at a specific center $r_C \in CU_i$ (conducted during the interconnection step of phase $i$)
succeeds with probability at least $1- p/n^{c'-1}$. There are at most $\deg_i \le n^{\rho} < n$ centers in $CU_i$.
Taking a union bound over all the centers in $CU_i$, we get the following lemma:
 
\begin{lemma}\label{lem:InterWCorrect2}
For a sufficiently large constant $c'$,  with probability at least $1 - p/n^{c' -2}$,
after $p$ sub-phases of the interconnection step of phase $i$, 
the following holds
 for any cluster $C \in U_i$ and for every vertex $v$ within $p$ hops from the center $r_C$ of $C$ in $\hopgraph$:
 
There is a tuple $(r_C, \hat{d}(v, r_C))$ in the estimates list $LCurrent_v$ of $v$ such that
 
$$d^{(p)}_{G^{(k-1)}} (v, r_C) \le \hat{d}(v, r_C) \le (1 + \zeta')^p \cdot d^{(p)}_{G^{(k-1)}} (v, r_C)$$
\end{lemma}

Recall that $\zeta' = \frac{\chi}{2\cdot(2\beta + 1)}$. Invoking Lemma~\ref{lem:InterWCorrect2} with $p = 2\beta +1$ and $\zeta' = \frac{\chi}{2\cdot(2\beta + 1)}$,
 implies the following corollary about the interconnection step of phase $i$:
\begin{corollary}\label{cor:interWFinal}
For a sufficiently large constant $c''$,  with probability at least $1 - 1/n^{c''}$,
after $2\beta +1$ sub-phases of the interconnection step of phase $i$, 
the following holds
for any cluster $C \in U_i$ and for every vertex $v$ within $2\beta +1$ hops from the center $r_C$ of $C$ in $\hopgraph$:
 
There is a tuple $(r_C, \hat{d}(v, r_C))$ in the estimates list $LCurrent_v$ of $v$ such that
\begin{equation}\label{eq:interBFEStretch}
d^{(2\beta + 1)}_{G^{(k-1)}} (v, r_C) \le \hat{d}(v, r_C) \le (1 + \chi) \cdot d^{(2\beta +1)}_{G^{(k-1)}} (v, r_C)
\end{equation}
\end{corollary}

Finally, after $2\beta + 1$ sub-phases of the interconnection step of phase $i$, 
we go through the estimates list of every center $r_C \in CU_i$
to check for every center $r_C' \in CU_i$, whether,
there is a tuple $(r_C', \hat{d}(r_C, r_C') ) \in LCurrent_{r_C}$
and $\hat{d}(r_C, r_C') \le \delta'_i/2$.
Then, for every such center $r_C'$ found,
we add an edge $(r_C, r_C')$ of weight $\hat{d}(r_C, r_C')$ into hopset $H_k$.
Note that if $d_G(r_C,  r_{C'} ) \le \delta_i/2$, then by equations (\ref{eq:2beta +1}) and (\ref{eq:interBFEStretch}), 
$\hat{d} (r_C,  r_{C'} ) \le (1 + \chi) \cdot (1 + \epsilon_{k-1} ) d_G(r_C, r_{C'}) = \delta'_i/2$.
Therefore, the edge $(r_C, r_{C'})$ will be added in to the hopset.

Lemmas~\ref{lem:firstPassInterWSpace},~\ref{lem:secondPassInterWSpace} and Corollary~\ref{cor:interWFinal} together imply the following corollary about the interconnection step of phase $i$:
\begin{lemma}\label{lem:SuperFinalHop2}
		
       For a sufficiently large constant $c''$, after $2\beta +1$ sub-phases of the interconnection step of phase $i$ during the construction of hopset $H_k$, $k \in [k_0, k_{\lambda}]$, 
       the following holds with probability at least $1-1/n^{c''}$:
      \begin{enumerate}
             \item The interconnection step of phase $i$ makes $2\beta + 1$ passes through the stream, and the total required space is $O(\frac{\beta}{\chi} \cdot  n^{1 + \rho} \cdot \log \Lambda \cdot \log^2 n \cdot (\log^2 n + \log \Lambda))$ bits.
            \item \label{item:Inter2} For every cluster $C \in U_i$ and every other cluster $C' \in U_i$ such that the center $r_C'$ of $C'$ is within distance $\delta_i/2$ in $G$ from center $r_C$ of $C$, 
             an edge $(r_C, r_{C'})$ of weight at most $(1+\chi) \cdot (1+ \epsilon_{k-1}) \cdot d_G(r_C, r_C')$ is added into hopset $H_k$, \\
             where $\epsilon_{k-1}$ is the stretch guarantee of the graph $G^{(k-1)}$.
             \end{enumerate}
\end{lemma}

Lemmas~\ref{lem:SuperFinalHop1} and~\ref{lem:SuperFinalHop2} imply that our algorithm simulates phase $i$ 
of insertion-only streaming algorithm (of~\cite{ElkinNeimanHopsets}) for the construction of a single scale hopset $H_k$  whp.
Note, however, that the edges added to the hopset $H_k$ by our algorithm during any phase $i$ ($0 \le i \le \ell$), incur an extra stretch of $(1+\chi)$
compared to the insertion-only algorithm. The reason is that in the insertion-only algorithm, every pair of sufficiently close cluster centres are connected via an edge of weight
\emph{exactly equal} to the length of the path between them in $\hopgraph$, while in our algorithm, the weight of the connecting edge is a $(1+\chi)$-approximation of the length of the path between them in
$\hopgraph$.\\

The following lemma follows by induction on the number of phases of our algorithm.
\begin{lemma}\label{lem:FinalH-K}
After $\ell$ phases, our single-scale hopset construction algorithm simulates the
insertion-only streaming algorithm of~\cite{ElkinNeimanHopsets} for constructing a single-scale hopset $H_k$ for scale $(2^k, 2^{k+1}]$, $k_0 \le k \le k_{\lambda}$,
 in the dynamic streaming setting whp such that \\
any edge $e$ added to the hopset $H_k$ by our algorithm is stretched at most by a factor of $(1+\chi)$ compared to the insertion-only algorithm.
 \end{lemma}
We return the edges of the set $H =\bigcup_{k_0 \le j \le k_{\lambda}} H_j$ as our final hopset.

Next, we analyze the properties of our final hopset $H$.

\input{HopFinal}

\input{AspectRatio}

\input{HopApplications}

%% file: HopFinal.tex
\subsection{Putting Everything Together}\label{sec:HopsetFinal}
\textbf{Size:} The size of our hopset $H$ is the same as that of the insertion-only algorithm of~\cite{ElkinNeimanHopsets}, since we follow the same criteria (as in~\cite{ElkinNeimanHopsets}),
when deciding which cluster centres to connect via a hopset edge during our construction.
Thus, the overall size of the hopset produced by our construction is $O(n^{1 + 1/\kappa} \cdot \log \Lambda)$ in expectation. \\
\textbf{Stretch and Hopbound:} Recall that $\epsilon_k$ is the value such that the graph $G^{(k)}$ 
(which is a graph obtained by adding the edges of hopset $H^{(k)}= \bigcup_{k_0 \le j \le k} H_j$ to the 
input graph $G$) provides stretch at most $1 + \epsilon_k$.
Also, recall that $k_0 = \floor{\log \beta}$ and $k_{\lambda} = \ceil{\log \Lambda}$.

Write $c_5 =2$.
We need the following lemma from~\cite{ElkinNeimanHopsets} regarding the stretch of a single scale hopset $H_k$, $k \in [k_0, k_{\lambda}]$
produced by the insertion-only algorithm.
We refer the reader to Lemma 3.10 and preamble of Theorem 3.11 of~\cite{ElkinNeimanHopsets} for the proof.
(Note that Lemma 3.10 and Theorem 3.11 of~\cite{ElkinNeimanHopsets} are proved for the construction of a single scale hopset 
in the congested clique model. These also apply to their insertion-only construction.(See Section 3.5 of~\cite{ElkinNeimanHopsets}.))
\begin{lemma}~\cite{ElkinNeimanHopsets}\label{lem:stretchOne}
Let $x,y \in V$ be such that $2^{k} \le d_G(x,y) \le 2^{k+1}$, then it holds that
\begin{equation}\label{eq:stretchLem1}
\begin{aligned}
d^{(h_{\ell})}_{G \cup H_k}(x,y) \le (1 + \epsilon_{k-1}) (1 + 16 \cdot c_5 \cdot \ell \cdot \epsilon)d_G(x,y),
\end{aligned}
\end{equation}
and $h_\ell = O(\frac{1}{\epsilon})^{\ell}$ is the hopbound.
\end{lemma}
\textbf{Rescaling:}
Define $\epsilon'' = 16 \cdot c_5 \cdot \ell \cdot \epsilon$. Therefore, the stretch of a single scale hopset $H_k$, $k \in [k_0, k_{\lambda}]$,
produced by the insertion-only algorithm of~\cite{ElkinNeimanHopsets} becomes $(1+ \epsilon_{k-1}) (1+ \epsilon'')$.

After rescaling, the hopbound $h_\ell$ becomes $O(\frac{\ell}{\epsilon''})^{\ell}$.
Recall that $\ell = \ell(\kappa, \rho) = \floor{\log(\kappa\rho)} + \ceil{ \frac{\kappa + 1}{\rho \kappa} } -1 \le \log(\kappa \rho) + \ceil{1/\rho} $,
 is the number of phases of our single-scale hopset construction.
It follows that the hopbound of the insertion-only algorithm is 
\begin{equation}
\label{eq:hopInstert}
\begin{aligned}
 \beta_{EN} = O\left(\frac{\log \kappa \rho + 1/\rho} {\epsilon''} \right)^{\log \kappa \rho + 1/\rho}.
\end{aligned}
\end{equation}

Observe that for $k = k_0$, graph $G^{(k-1)}$ is the input graph $G$ itself, since $H_k$ for all $k < k_0$ is $\phi$. (See Section~\ref{sec:HopsetsAlgoOverview} for details.)
Therefore, $1 + \epsilon_{k-1}$ for $k = k_0$ is equal to $1$. It follows therefore that\\
the stretch $1 + \epsilon_k = 1 + \epsilon_{k_{EN} }$, of the insertion-only algorithm follows the following sequence:
 $1 + \epsilon_{{k_0}_{EN}} = (1+\epsilon'')$ and for the higher scales, $1 + \epsilon_{{k+1}_{EN}} = (1 + \epsilon'')\cdot (1+ \epsilon_{k_{EN}})$.
 
 By Lemma~\ref{lem:FinalH-K}, the stretch of our single scale hopset construction (Section~\ref{sec:constructH_k}) for any scale $(2^k, 2^{k+1}]$, $k_0 \le k \le k_{\lambda}$
 is $(1+\chi)$ times the stretch of the corresponding hopset produced by the insertion-only algorithm.
 We set $\chi = \epsilon''$.
 Incorporating the additional stretch incurred by our algorithm into the stretch analysis of~\cite{ElkinNeimanHopsets}, we get the following lemma about the stretch
 of our dynamic streaming algorithm
 \begin{lemma}\label{lem:stretchHop1}
	For $k \in [k_0, k_{\lambda}]$, we have
	\begin{align*}
	1 + \epsilon_{k_0} &= (1 + \epsilon'')^2\\
        1 + \epsilon_k &= (1+\epsilon'')^2 (1 + \epsilon_{k-1})\text{~for~} k > k_0
	\end{align*}
\end{lemma}
Observe that Lemma~\ref{lem:stretchHop1} implies that the overall stretch of our hopset $H$ is at most $(1+ \epsilon'')^{2 \log \Lambda}$.\\
Recall that the desired stretch of our hopset construction is $1 + \epsilon'$ (see Section~\ref{sec:HopsetsAlgoOverview}), where $\epsilon' >0$ is an input parameter of our algorithm.\\
We set $\epsilon'' = \frac{\epsilon' } {4 \cdot \log \Lambda}$, and it follows that our overall stretch is \\
$$\left(1 + \frac{\epsilon'}{4\log \Lambda}\right)^{2 \log \Lambda} \le 1 + \epsilon'$$.

Plugging in $\epsilon'' = \frac{\epsilon' } {4 \cdot \log \Lambda}$ in (\ref{eq:hopInstert}), we get the following expression for the hopbound of our dynamic streaming hopset: 
\begin{equation}
\begin{aligned}\label{eq:hopbound}
\beta' = O\left(\frac{\log \Lambda}{\epsilon'} (\log \kappa\rho + 1/\rho) \right)^{\log \kappa \rho + 1/\rho}.
\end{aligned}
\end{equation}
(See also (\ref{eq:hopboundclaim}).)

Also recall that we had defined $\beta = (\frac{1}{\epsilon})^{\ell}$ for using $2\beta +1$ as the hop-depth of our explorations.
After the two rescaling steps as above, we get that
$\beta = \beta'$.

Next we analyze the pass complexity of our overall construction.
\begin{lemma}\label{lem:HopsetPassComplexity}
Our dynamic streaming algorithm makes $O(\beta' \log \Lambda \cdot(\log \kappa\rho + 1/\rho))$ passes through the stream.
\end{lemma}

\begin{proof}
In our single scale hopset construction (See Section~\ref{sec:Hopsets}), we make $O(\beta')$ passes during 
the superclustering step and $O(\beta')$ passes during the interconnection step of any phase.
(Note that $\beta' = \beta$ and $\beta' = \beta'(\epsilon, \kappa, \rho)$ is given by (\ref{eq:hopbound}).)
There are $\ell \le \log(\kappa \rho) + \ceil{1/\rho}$ phases in total.
Thus, we make $O(\beta' \cdot(\log \kappa \rho + 1/\rho))$ passes through the stream during the construction of a single scale hopset.
We build at most $\log \Lambda$ single scale hopsets one after the other.
Therefore, the overall pass complexity of our hopset construction is $O(\beta' \cdot \log \Lambda \cdot(\log \kappa \rho + 1/\rho))$.
\end{proof}

To summarize, we get the following equivalent of Theorem 3.16 of~\cite{ElkinNeimanHopsets} summarizing our results:

\begin{theorem}\label{thm:HopSummary}
For any $n$-vertex graph $G(V, E, \omega)$ with aspect ratio $\Lambda$, $2 \le \kappa \le (\log n)/4$,
$1/\kappa \le \rho \le 1/2$ and $0 < \epsilon' < 1$, our dynamic streaming algorithm computes 
a $(1+\epsilon', \beta')$ hopset $H$ with expected size $O(n^{1+1/\kappa} \cdot \log \Lambda)$  and the hopbound $\beta'$ given by
(\ref{eq:hopbound}) whp.\\
It does so by making $O(\beta' \cdot \log \Lambda\cdot(\log \kappa \rho + 1/\rho))$ passes through the stream and
using $O(\frac{\beta'}{\epsilon'} \cdot  n^{1 + \rho} \cdot \log \Lambda \cdot \log^2 n \cdot (\log^2 n + \log \Lambda))$ bits of space.
\end{theorem}

\subsection{Path-Reporting Hopsets}\label{sec:PathReporting}
In certain applications of hopsets in the streaming model such as in the computation of approximate shortest paths,
knowing all the hopset edges is not sufficient. For every hopset edge $ e=(u,v)$, one also 
needs to know the actual path $\pi(u,v)$ in $G$ that implements $e$.
We say that a hopset $H$ is \emph{path-reporting}, 
if for every hopset edge $(u,v) \in H$, there exists a path $\pi(u,v)$ between $u$ and $v$ in $G$ such that
$\omega_G (\pi(u,v)) = \omega_H(u,v)$ and $H$ has enough enough information to compute $\pi(u,v)$.
(See~\cite{ElkinNeimanHopsets} for more information on path-reporting hopsets.)
Recall that for a given vertex $v \in V$,
the procedure $\emph{GuessDistance}$ executed as part of the $(2\beta' + 1)$-limited 
 Bellman-Ford explorations conducted 
during our hopset construction gives us an estimate of $v$'s distance to a source $s$.
In addition, it also gives us the name of a neighbour $u$ of $v$ such that $u$ is 
$v$'s parent on the exploration tree rooted at $s$. 
When adding an edge $(r_C, r_C')$ between two nearby cluster centres to some hopset $H_k$, $k \in [k_0, k_\lambda]$,
the path $\pi(r_C, r_C')$ between $r_C$ and  $r_C'$ in $G^{(k-1)}$ can be deduced from the parent pointers returned by calls
to procedure $\emph{GuessDistance}$. For every hopset edge $e$ added to the hopset $H_k$ for scale index $k \in  [k_0, k_\lambda]$,
we can store the corresponding path between its endpoints in $G^{(k-1)}$ to the hopset.
Note that some of the edges on this path may themselves by hopset edges from a lower level hopset.
This is not a problem, since the path corresponding to such an edge $e'$ should have been stored during the construction of the hopset
in which $e'$ was added.
Since all the explorations conducted during the construction of $H_k$ are $(2\beta' + 1)$-limited, it increases the space usage of
the hopset construction by a factor of $O(\beta')$.  

To summarize, we get the following equivalent of Theorem~\ref{thm:HopSummary}
 for path-reporting hopsets:

\begin{theorem}\label{thm:HopSummaryPathReporting}
For any $n$-vertex graph $G(V, E, \omega)$ with aspect ratio $\Lambda$, $2 \le \kappa \le (\log n)/4$,
$1/\kappa \le \rho \le 1/2$ and $0 < \epsilon' < 1$, our dynamic streaming algorithm computes 
a $(1+\epsilon', \beta')$ hopset $H$ with expected size $O(n^{1+1/\kappa} \cdot \log \Lambda)$  and the hopbound $\beta'$ given by
(\ref{eq:hopbound}) whp.\\
It does so by making $O(\beta' \cdot \log \Lambda\cdot(\log \kappa \rho + 1/\rho))$ passes through the stream and
using $O(\frac{\beta'^2}{\epsilon'} \cdot  n^{1 + \rho} \cdot \log \Lambda \cdot \log^2 n \cdot (\log^2 n + \log \Lambda))$ bits of space.
\end{theorem}

%% file: AspectRatio.tex
\section{Eliminating Dependence on the Aspect Ratio}\label{sec:aspectRatio}
In this section, we devise a reduction that eliminates the dependence on the aspect ratio $\Lambda$ of the graph from the hopbound, size and the overall pass complexity of our hopset construction. 
We note, however, that the space complexity of our algorithm is still linear in $\log \Lambda$.
The reduction was previously used by Elkin and Neiman~\cite{ElkinNeimanHopsets} for the same purpose and was implemented by them in Congested Clique, CONGEST, PRAM and  insertion-only streaming model. We adapt their insertion-only implementation to work in dynamic streaming model. 
The result of applying the reduction to our algorithm is summarized in Theorem~\ref{thm:AspRatioReduction}.

\subsection{Overview}\label{subsec:aspRatioOver}
In this sub-section, we give a general overview of the reduction for a weighted input graph $G =(V,E, \omega)$, and then 
describe its implementation in the dynamic streaming model in the subsequent sub-section.
Fix a parameter $0 < \epsilon < 1/2$. Recall that $k_0 = \floor{\log \beta}$ and $k_\lambda = \ceil{\log \Lambda} -1$.
For every scale index $k \in [k_0, k_\lambda]$, we build a graph $\mathcal{G}_k$ that contains edges with weights in the range $((\epsilon/n)\cdot 2^k, (1 + \epsilon/2) \cdot2^{k+1}]$.
This graph can be constructed by deleting the \emph{heavy} edges (of weight more than $2^{k+1}$) 
and contracting the \emph{light} edges (of weight less than $(\epsilon/n)\cdot 2^k$) from $G$. By contraction, we mean grouping vertices with light edges between them into supervertices which we call \emph{nodes}. 
Each node of graph $\mathcal{G}_k$ is a subset of $V$. 
Observe that the nodes of $\mathcal{G}_k$ are connected components of the graph obtained by deleting all the edges of weight more that $(\epsilon/n)\cdot 2^k$ from $G$. 
The weight of an edge $(X,Y) \in E(\mathcal{G}_k)$ is set to be
 
 \begin{equation}\label{eq:aspWeightExternal}
 \mathcal{W}(X,Y) = \omega(x,y) +  (\epsilon/n)\cdot 2^k \cdot (|X| + |Y|),
 \end{equation}
where $x \in X$ and $y \in Y$ and edge $(x,y) \in E$ is the shortest edge between a vertex of $X$ and a vertex of $Y$.

Observe that the minimal weight of an edge in $\mathcal{G}_k$ is at least $(\epsilon/n)\cdot 2^k + 2(\epsilon/n)\cdot 2^k  > (\epsilon/n)\cdot 2^{k+1}$
and the maximal weight is at most $2^{k+1} + \epsilon \cdot  2^k = (1+ \epsilon/2) \cdot 2^{k+1}$.
Therefore, for every scale index $k \in [k_0, k_\lambda]$, the corresponding graph $\mathcal{G}_k$ has aspect ratio $O(n/\epsilon)$. 

Every node $U$ in $\mathcal{G}_k$ is assigned a designated center $u$ and 
we add edges from $u$ to every other vertex in $U$. These edges are called the \emph{star} edges.
Let $S_k$ denote the set of all the star edges of  $\mathcal{G}_k$. Consider a contraction of an edge $(x', y')$,
$x' \in X$, $y' \in Y$, connecting nodes $X$, $Y$ with centers $x^{*} \in X$ and $y^{*} \in Y$. Assuming that $|X| \ge |Y |$, the vertex
$x^{*}$ is declared the center of $U = X \cup Y$, and we add to $S_k$ edges from $x^{*}$ to every vertex of $Y$. The
weight of the edge $(x, z)$ for each $z \in Y$ is set as
\begin{equation}\label{eq:aspWeightInterrnal}
\mathcal{W}(x, z) = (\epsilon/n) \cdot 2^k \cdot |U| . 
\end{equation}
Observe that the weight of the star edge $(x,z)$ dominates the value of $d_G(x,z)$, since there exists a path between $x$ and $z$ in $G$ consisting of at most 
$|U|-1$ edges each of weight at most $\epsilon/n \cdot 2^k$.

For every $k \in [k_0, k_\lambda]$, 
a separate single-scale $(1+ \epsilon, \beta)$-hopset $\mathcal{H}_k$ for the scale $(2^k, 2^{k+1}]$ 
is computed on the graph $\mathcal{G}_k$. This is done in parallel for all $k \in [k_0, k_\lambda]$.
The ultimate hopset $\mathcal{H}$ is computed as follows. 
For every scale index $k \in [k_0, k_\lambda]$, 
and for every edge $(X,Y)$ of weight $d$ in the hopset $\mathcal{H}_k$, 
a corresponding edge $(x^{*}, y^{*})$ of weight $d$ between the centers $x^{*}$ and $y^{*}$ of $X$ and $Y$ respectively is added to the final hopset $\mathcal{H}$. 
To ensure that the number of hops within each node is also small, 
we also add to $\mathcal{H}$ the set $S = \bigcup_{k \in [k_0, k_\lambda]} S_k$ of all the star edges.

\subsubsection{Analysis}\label{subsec:aspRatioAnalysis}
In this section, we analyze the properties of the graphs $\mathcal{G}_k$ for $k \in [k_0, k_\lambda]$ 
and of the corresponding hopset $S_k \cup \mathcal{H}_k$.
Let $\mathcal{V}_k$ be the set of nodes of $\mathcal{G}_k$. 
For a node $U \in \mathcal{V}_k$, let $S(U)$ denote the set of star edges of the node $U$, i.e., $S(U) = \{(x,y) \in S | x,y \in U  \}$.

\begin{lemma}\label{lem:starEdges}
 $|S| \le n\log n$.
\end{lemma}
\begin{proof}
The proof follows by induction on the scale index $k$. For $k = k_0$, a node $U \in \mathcal{V}_{k_0}$ contains $|U| - 1 \le |U| \cdot \log |U|$ edges.

We assume that the claim holds for some scale index $k \in [k_0 +1, k_\lambda]$ and prove that it also holds for $k +1$.
Let $U$ be a node in $\mathcal{V}_{k+1}$ and $X_1, X_2, \ldots, X_t$ be nodes in $\mathcal{V}_k$ such that $U = \bigcup_{j \in [1,t]} X_j$ and 
$|X_1| \ge |X_2| \ge \ldots \ge |X_t|$. Let $s$ denote the size of $U$ and $s_j$ denote the size of the node $X_j$ for $j \in [1,t]$.
By induction hypothesis, $X_j$ (for $j \in [1,t]$) contains at most $s_j \cdot \log s_j$ star edges.
When we merge $X_1, X_2, \ldots X_t$ to form $U$, the center $x_1$ of $X_1$ becomes the center of $U$ and we add edges from $x_1$ to 
all the vertices in $\bigcup_{j \in [2,t]} X_j$. It follows that
\begin{align*}
\begin{aligned}
s &\le  \sum_{j \in [2,t]} s_j   + \sum_{j \in [1,t]} s_j \cdot  \log s_j \\
   &=  s_1 \log s_1 + \sum_{j \in [2,t]} s_j (1 + \log s_j)
   =   s_1 \log s_1 + \sum_{j \in [2,t]} s_j  \log(2s_j)\\
   & \le s_1 \log (s_1 + s_2) +  \sum_{j \in [2,t]} s_j \log (s_1 + s_j) \le s_1 \log (s_1 + s_2) +  \sum_{j \in [2,t]} s_j \log (s_1 + s_2)\\
   & = \log(s_1 + s_2) \cdot  \sum_{j \in [1,t]} s_j  = s \log (s_1 + s_2) \le s \log s.  
   \end{aligned}
\end{align*}

Observe that at a certain point when the scale index is sufficiently large, we have a graph with a single node containing all the vertices of the input graph.
At this point, we have added at most $n\log n$ star edges to the hopset $\mathcal{H}$.
\end{proof}

\textbf{Relevant Scales:} Recall that we build a separate hopset $\mathcal{H}_k$ for every scale index $k \in [k_0, k_\lambda]$.
Specifically, hopset $\mathcal{H}_k$ is used to approximate distances in the range $(2^k, 2^{k+1}]$.
If no edge in $G$ has weight in the range $(2^k/n, 2^{k+1}]$, then there is no pair of vertices in $V$ with distance in the range $(2^k, 2^{k+1}]$.
In this case, the hopset $\mathcal{H}_k$ is \emph{redundant}. We call a scale index $k \in [k_0, k_\lambda]$ \emph{relevant} 
if there exists an edge $(u,v) \in E$ such that $\omega(u,v) \in (2^k/n, 2^{k+1}]$, and \emph{redundant} otherwise.
Let $K$ be the set of relevant scale indices from $ [k_0, k_\lambda]$.
Observe that a given edge can induce at most $\log n$ relevant scales and hence $|K| = O(|E| \cdot \log n)$.
We construct a hopset $\mathcal{H}_k$ only for a graph $\mathcal{G}_k$ with $k \in K$.

\textbf{Active Nodes:}
The nodes of the graphs $\{\mathcal{G}_k\}_{k \in K}$ induce a laminar family $\mathcal{L}$ on $V$ which contains at most $2n-1$ distinct sets. 
We say that a node $U$ in the graph $\mathcal{G}_k$ is \emph{active} if it has degree at least $1$.
Denote by $n_k$ the number of active nodes in $\mathcal{G}_k$.
By arguments similar to those used in~\cite{ElkinNeimanHopsets, Cohen1994Polylogtime}, one can bound  the number of active nodes in 
graphs $\{\mathcal{G}_k\}_{k \in K}$ and get that
\begin{equation}
\label{eq:activeNodes}
\begin{aligned}
\sum_{k \in K} n_k = O(n \log n)
\end{aligned}
\end{equation}

\textbf{Hopset Size:}
Recall that for every hopset edge $(X,Y)$ of a single-scale hopset $\mathcal{H}_k$ of weight $d$, 
we add to the ultimate hopset $\mathcal{H}$, an edge of the same weight between the centers of $X$ and $Y$. 
In addition, we add all the star edges of hopsets $\{\mathcal H_k\}_{k \in K}$ to $\mathcal{H}$.
Also, recall from Section~\ref{sec:HopsetFinal} that the expected size of a single-scale hopset $\mathcal{H}_k$
on a graph on $n$ vertices is $O(n^{1 + 1/\kappa})$. 
For every $k \in [k_0, k_\lambda]$, only active nodes of $\mathcal{G}_k$ 
take part in the construction of  hopset $\mathcal{H}_k$.
 Therefore, Lemma~\ref{lem:starEdges} and equation (\ref{eq:activeNodes}) together imply that
 
 \begin{equation}
\label{eq:aspHopsetSize}
\begin{aligned}
|\mathcal{H}| &= |S| + \sum_{k \in K} |\mathcal{H}_k| \le n\log n + \sum_{k \in K} O( n_k^{1+ 1/\kappa}) \le n \log n + n^{1/\kappa} \sum_{k \in K} O( n_k)
= O(n^{1 + 1/\kappa} \log n)
 \end{aligned}
\end{equation}

\subsection{Implementation in Dynamic Streaming model}\label{subsec:aspRatioImp}
The algorithm proceeds in two phases. 
In the first phase, we make one pass through the stream and compute all the nodes of $\{\mathcal{G}_k\}_{k \in K}$.
In the second phase, we compute in parallel a single-scale hopset $\mathcal{H}_k$ for every $k \in K$.
The details of the two phases are provided below.

Recall that each node $U$ of $\{\mathcal{G}_k\}_{k \in K}$ is a subset of the vertex set $V$ of $G$ and has a 
designated center $u^{*}$. Whenever we contract an edge between nodes $X$ and $Y$ with $|X| \ge |Y|$, the vertices 
of $Y$ get a new center but the size of the node containing them is at least doubled. This implies that each vertex 
changes the center of the node containing it at most $\log n$ times. For every vertex $v \in V$, we store a \emph{list} $L(v)$
of pairs $(i,c)$. Each pair $(i,c^*) \in L(v)$ indicates that at scale $i \in K$, the node containing $v$ was merged into
a larger node centred at $c^*$. Initially, $L(v)$ for every $v \in V$ is empty. The set $\emph{Lists} = \{L(v)| v \in V\}$ encodes
the laminar family $\mathcal{L}$ induced by the nodes of $\{\mathcal{G}_k\}_{k \in K}$.
The description of first phase up to this point is similar to that of the insertion-only streaming implementation of~\cite{ElkinNeimanHopsets}. 
(See Section 4.2 of~\cite{ElkinNeimanHopsets}  for more details.)
In~\cite{ElkinNeimanHopsets}, the $\emph{Lists}$ data structure 
and the set $S$ of star edges are recomputed upon arrival of every edge on the stream. 
 In the dynamic streaming setting, this approach is not applicable, 
 because a removal of an edge $e = (x,y)$ from the stream may 
 cause a split of an existing node $U$ that contains both $x$ and $y$.
 On the other hand, if there is an alternative $x-y$ path in $G(U)$,
 the node will stay intact. To be able to distinguish between these two scenarios, 
 one needs to take a different approach. 
 
 For each $v \in V$, we update a
data structure $\emph{XORSlots}(v)$ (described in the sequel) upon arrival of edge updates,
 and later use these data structures
 to compute $\emph{Lists}$ and star edges offline.
 The data structure  $\emph{XORSlots}(v)$ for vertex $v$ consists of $\lambda + 1 = O(\log \Lambda)$ arrays, 
 $\emph{Slots}(v, k)$ for $k \in [k_0, k_\lambda]$.
Each $\emph{Slots}(v, k)$ array is similar to the $\emph{Slots}$ array used in procedure \emph{GuessDistance}. (See Section~\ref{sec:GuessDistance} for more details.) 
It enables us to sample from the edges incident on $v$, an edge with weight in the range $(\epsilon/n) \cdot (2^{k-1}, 2^{k}]$ for 
$k \in [k_0 +1, k_\lambda]$ and 
$(0, (\epsilon/n) \cdot 2^k]$ for $k =k_0$.
 As in $\emph{Slots}$ array (Section~\ref{sec:GuessDistance}), elements of $\emph{Slots}(v, k)$ array are used to
sample certain edges incident on $v$ on a range of probabilities 
and we use a function chosen uniformly at random from a family of pairwise independent hash functions for sampling. 
Note that we use the same hash function for computing the $\emph{Slots}(v,k)$ arrays for every $v \in V$ and every $k \in K$.
Specifically, the element at index $i$ of $\emph{Slots}(v, k)$ maintains 
the bitwise XOR of the binary names of the edges sampled with probability corresponding to the  index $i$. 
For a given edge $e = (u,v)$, its binary name is a concatenation of the binary representation of the $IDs$ of its endpoints, 
with the smaller of the two $IDs$ appearing first.
Upon arrival of an update to some edge $(u,v)$ on the stream with weight $\omega(u,v) \in (\epsilon/n) \cdot (2^{k-1}, 2^{k}]$
for some $k \in [k_0, k_\lambda]$, the slots array of both $u$ and $v$ for the scale $k$ are updated.
Note that for some $v \in V$ and $k \in [k_0, k_\lambda]$, 
the slots array $\emph{Slots}(v, k)$ enables us to sample with constant probability (See Section~\ref{sec:GuessDistance}.)
  an edge $e = (u,v) \in E$ with $\omega(u,v) \in (\epsilon/n) \cdot (2^{k-1}, 2^{k}]$, if such an edge exists.
 We maintain $O(\log n)$ copies of $\emph{Slots}(v, k)$ for each $k \in [k_0, k_\lambda]$  in our data structure $\emph{XORSlots}(v)$ 
 to increase the success probability of our sampling procedure. 
 Each of these $O(\log n)$ copies is generated using a different hash function.

 For a set $U \subseteq V$, an edge $e = (u,v)$ is called an \emph{outgoing edge} of $U$, if $u \in U\text{and~} v \notin U$.
 An important property of the $\emph{XORSlots}(v)$ data structures is that 
 for a given set $ U = \{u_1, u_2, \ldots, u_m\} \subseteq V$, and a scale index $k \in [k_0, k_\lambda]$,
 we can build an array $\emph{Slots}(U, k)$ (similar to the $\emph{Slots} (v,k)$ array for a single vertex $v$)
  from $\{\emph{Slots} (u_j, k)\}_{j \in [1,m]}$, provided all these arrays are generated using the same hash function.
 In sketching literature, this property is called linearity of sketches. (See~\cite{AhnGuha12a, CormodeFirmani} for more details.)
 The $\emph{Slots}$ array $\emph{Slots}(U, k)$ for set $U$ enables us to sample an outgoing edge of $U$ with weight in the range 
 $(\epsilon/n) \cdot (2^{k-1}, 2^{k}]$.
The value of the element at index $i$  of $\emph{Slots}(U, k)$ is given by \\
 $$\emph{Slots}(U, k)[i] = \emph{Slots}(u_1, k)[i] \bigoplus  \emph{Slots}(u_2, k)[i] \bigoplus \ldots \bigoplus  \emph{Slots}(u_m, k)[i].$$
Recall that $\bigoplus$ stands for bitwise XOR and all the $\emph{Slots}$ arrays in $\{\emph{Slots} (u_j, k)\}_{j \in [1,m]}$ 
 are generated using the same hash function. 
 A given edge $e =(x,y)$
 with $\omega(x,y) \in (\epsilon/n) \cdot (2^{k-1}, 2^{k}]$
 will be added to the same index elements in both $\emph{Slots} (x, k)$ and $\emph{Slots} (y, k)$.
 This ensures that every element of $\emph{Slots}(U, k)$ contains a bitwise XOR of $\emph{only}$ the outgoing edges of $U$.
 As with the $\emph{Slots}$ arrays for individual vertices, for some appropriate index $i$, the element  $\emph{Slots}(U, k)[i]$
 will contain only one edge with at least a constant probability.
 If we fail to sample an outgoing edge for $U$, we use a different set of $\emph{Slots}$ arrays generated using a different hash function.
 Recall that for each vertex $v \in V$ and $ k \in K$, we maintain $O(\log n)$ copies of $\emph{Slots}(v,k)$, each generated using a different
 hash function.
   
After the first pass, we go offline to compute the nodes of graphs $\{\mathcal{G}_k\}_{k \in K}$.
We compute the nodes of graphs $\mathcal{G}_k$ for $k = k_0,  k_0 + 1, k_0 +2, \ldots$, sequentially, in that order.
In the following, we describe a procedure called $ComputeCC$ which computes the nodes of  $\mathcal{G}_{k}$
from the nodes of $\{\mathcal{G}_j\}_{j < k}$.
For $k =k_0$, the procedure $ComputeCC$ computes the nodes of $\mathcal{G}_{k_0}$ from the vertex set $V$ of the input graph $G$.
We make $O(\log n)$ iterations. Each iteration starts with a set $CC$ consisting of all the connected components 
of $\mathcal{G}_{k}$ identified so far. Note that each connected component $C \in CC$ is a subset of the vertex set $V$.
Initially, set $CC$ contains all the nodes of $\{\mathcal{G}_j\}_{j < k}$.
For $k = k_0$, $CC$ contains singleton sets $\{v \}$, for every $v \in V$.
In every iteration, for every component $C \in CC$, we find an outgoing edge with weight in the range $(\epsilon/n) \cdot (2^{k-1}, 2^k]$,
for $k > k_0$ and $(0, (\epsilon/n) \cdot 2^{k_0}]$, for $k = k_0$ (if there exists such an edge).
 For this, we can compute (offline) and use the $\emph{Slots}$ array of every component, 
as described in the last paragraph. For a component $C_1 \in CC$ with an outgoing edge in the appropriate weight range
 to another component $C_2 \in CC$,
we merge the two components, and add $C_1 \cup C_2$ to $CC$, and remove $C_1$ and $C_2$ from $CC$.
We stop when none of the connected components in $CC$ have an outgoing edge with weight in the appropriate range.
Note that each of these connected components corresponds to a node of $\mathcal{G}_{k}$.
For computing the nodes of $\mathcal{G}_k$ from those of  $\{\mathcal{G}_j\}_{j < k}$, we only need to consider
edges $e =(x,y)$ with $\omega(x,y) \in (\epsilon/n) \cdot (2^{k-1}, 2^{k}]$ such that $x \in X$ and $y \notin X$, 
for every node $X$.
For every such node $X$, we merge $X$ with the node $Y$ containing $y$.
Note that while merging the nodes $X$ and $Y$, the edge $(x,y)$ that causes the merge and its 
exact weight do not matter as long as such an edge exists.
Therefore, we do not need to consider every edge incident on every vertex of 
a node $X$ as long as we can find an outgoing edge (if exists) with weight in the range $ (\epsilon/n) \cdot (2^{k-1}, 2^{k}]$.
This can be easily done by considering only the $\emph{Slots}$ array  $\emph{Slots}(X, k)$,
 which can be computed offline from the $\emph{Slots}$ arrays of the vertices contained in $X$.

At the end of the execution of procedure $ComputeCC$, we assign centers to the newly formed nodes and add star edges
to the set $S$ as follows.
For $k = k_0$, we assign an arbitrary vertex of each node as its center and add its star edges 
with weight given by  equation~(\ref{eq:aspWeightInterrnal})
 to the set $S$.
We also add a pair $(k_0, u^{*})$ to the list $L(v)$ of every vertex $v$ in the node centred at $u^{*}$.
For $k > k_0$, let $U \in CC$ be a node of $\mathcal G_k$  and let
$U$ be formed by merging nodes $X_1,X_2,\ldots, X_t$ (each from $\{\mathcal{G}_j\}_{j < k}$).
Let further $X_1$ be the largest node among $X_1,X_2,\ldots, X_t$, we assign the center $x^*$ of 
$X_1$ as the center of new node $U$.
 We update the list $L(x)$ for every $x \in U \setminus X_1$ with a pair $(k, x^{*})$, and add
a star edge $(x^{*},  x)$ with weight given by equation~(\ref{eq:aspWeightInterrnal}) to the set $S$.

Having computed the nodes of graphs $\{\mathcal{G}_{k}\}_{k \in K}$, we turn to the stream again to compute our hopsets.
We compute a separate hopset $\mathcal{H}_k$ for all $k \in K$ in parallel. 
For each $k \in K$, we run our hopset construction algorithm
from Section~\ref{sec:constructH_k}. 
Initially, the vertices of $\mathcal{G}_k$ can be 
derived from the $Lists$ data structure that we populated after the first pass.
We identify each vertex of $\mathcal{G}_k$ by its center. Whenever some update to an edge $(x,y)$ of weight $\omega(x,y)$ is read from the stream, we
know that it is active in at most $O(\log n/\epsilon)$ scales. For each scale $k \in K$ such that $(\epsilon/n) \cdot 2^k \le \omega(x,y) \le 2^{k+1}$,
we use the $Lists$ data structure to find the centres of nodes containing $x, y$ in $\mathcal{G}_k$, 
and execute the algorithm from Section~\ref{sec:constructH_k}
 as if an edge between these centres (of weight given by (\ref{eq:aspWeightExternal})) was just read from the stream. 
 Recall that the construction of a single-scale hopset involves performing approximate Bellman-Ford explorations of the input graph from a set 
of starting vertices. 
In our approximate Bellman-Ford exploration algorithm (Section~\ref{sec:WeightedForest}), 
we use a procedure $\emph{GuessDistance}$ (Section~\ref{sec:GuessDistance}) to estimate the distance of a vertex to the set of starting vertices.
We need to slightly tweak this procedure to take care of the fact that the graph $\mathcal{G}_k$ is a multigraph. We can do this by
maintaining a running sum of the CIS-based encoding (see Section~\ref{sec:Encodings} for more details) 
of the sampled edges instead of bitwise XOR of their binary names in the procedure $\emph{GuessDistance}$. 
Note that this change does not affect the space usage of the procedure $\emph{GuessDistance}$.

Elkin and Neiman~\cite{ElkinNeimanHopsets} provide a detailed stretch analysis of the reduction in the centralized model 
which also applies to their insertion-only implementation. (See Section 4 of ~\cite{ElkinNeimanHopsets}.) 
In particular, they show that the ultimate hopset $\mathcal{H}$ produced by this reduction 
 is a $(1+ 6 \epsilon, 6\beta +5)$-hopset of the input graph $G$. See Lemma 4.3 of~\cite{ElkinNeimanHopsets}.
 By Lemma~\ref{lem:FinalH-K}, every hopset edge of our single-scale hopset $\mathcal{H}_k$ ($k \in K$) construction 
is stretched at most by a factor of $(1+ \chi)$ compared to the insertion-only algorithm of~\cite{ElkinNeimanHopsets}.
We set $\chi = 6 \epsilon$. It follows that the stretch of our dynamic streaming construction is $(1 + 6\epsilon)^2$. 
Let $\epsilon' > 0$ be the desired stretch of our overall construction.
We set $6\epsilon = \epsilon'/4$.
Therefore, the overall stretch of our ultimate hopset is 

\begin{align*}
\begin{aligned}
 ( 1+ \epsilon'/4 )^2 \le 1 + \epsilon'.
\end{aligned}
\end{align*}

Next we analyze the space usage of the reduction. The following lemma summarizes the space requirement of the first phase.
\begin{lemma}\label{lem:AspRatioSpace1}
The first phase of our dynamic streaming reduction requires $O(n \cdot \log^3 n \cdot \log \Lambda)$ bits of memory.

\end{lemma}
\begin{proof}
The memory usage of the first phase has two main components, 
the space required to maintain $\emph{XORSlots}(v)$ data structures, 
and the space required to store the $O(\log n)$ hash functions of $O(\log n)$ bits each.
The storage of $\emph{XORSlots}(v)$ data structures involves storing for every vertex $v \in V$, 
$O(\log n)$ copies of $O(\log \Lambda)$ $\emph{Slots}$ arrays, each of size $O(\log^2 n)$ bits.
The total space required is therefore $O(n \cdot \log^3 n \cdot \log \Lambda)$ bits.
\end{proof}

As an output of the first phase, we produce the $\emph{Lists}$ data structure 
which is used throughout the second phase of the reduction.
The $\emph{Lists}$ data structure requires $O(n \cdot \log^2 n)$ bits of memory.
By equation (\ref{eq:aspHopsetSize}), the size of hopset $\mathcal{H}$ is $O(n^{1 + 1/\kappa} \cdot \log n)$ 
which implies that the space required to store the hopset edges is $O(n^{1 + 1/\kappa} \cdot \log^2 n)$ bits. 
In the second phase, we invoke our dynamic streaming algorithm from Section~\ref{sec:constructH_k} to construct all the relevant hopsets in parallel.
Lemmas~\ref{lem:SuperFinalHop1} and~\ref{lem:SuperFinalHop2} summarize the resource requirements of the two main steps of the algorithm 
from Section~\ref{sec:constructH_k}. 
Using the fact that every graph in $\{\mathcal{G}_k\}_{k \in K}$ has aspect ratio $O(n/\epsilon)$, 
we can essentially replace $\log \Lambda$ by $\log (n/\epsilon)$ in the space requirements in Lemmas~\ref{lem:SuperFinalHop1} and~\ref{lem:SuperFinalHop2}.
We get the following lemma summarizing the space requirement of the second phase of our reduction

\begin{lemma}\label{lem:AspRatioSpace1}
The second phase of our dynamic streaming reduction requires $O(\frac{\beta'}{\epsilon'} \cdot \log^2 1/\epsilon' \cdot  n^{1 + \rho} \cdot \log^5 n)$ bits of space.
\end{lemma}
The first phase of our reduction requires only one pass through the stream.
Since we build all the relevant hopsets in parallel, 
we can drop the $\log \Lambda$ factor from the pass complexity of hopset 
construction given by Lemma~\ref{lem:HopsetPassComplexity}.

Formally, we get the following equivalent of Theorem~\ref{thm:HopSummary}.

\begin{theorem}\label{thm:AspRatioReduction}
For any $n$-vertex graph $G(V, E, \omega)$ with aspect ratio $\Lambda$, $2 \le \kappa \le (\log n)/4$,
$1/\kappa \le \rho \le 1/2$ and $0 < \epsilon' < 1$, our dynamic streaming algorithm computes 
a $(1+\epsilon', \beta')$ hopset $H$ with expected size $O(n^{1+1/\kappa} \cdot \log n)$  and the hopbound $\beta'$ given by
\begin{equation}
\begin{aligned}
\beta' =  O\left (\frac{(\log \kappa \rho + 1/\rho )\log n}{\epsilon'}  \right)^{\log \kappa \rho + 1/\rho}
\end{aligned}
\end{equation}
whp.\\
It does so by making $O(\beta' \cdot (\log \kappa \rho + 1/\rho))$ passes through the stream and
using $O(n \cdot \log^3 n \cdot \log \Lambda)$ bits of space in the first pass and 
$O(\frac{\beta'}{\epsilon'} \cdot \log^2 1/\epsilon'  \cdot  n^{ 1 +\rho} \cdot \log^5 n)$ 
bits of space (respectively $O(\frac{\beta'^2}{\epsilon'} \cdot \log^2 1/\epsilon'  \cdot  n^{ 1 +\rho} \cdot \log^5 n)$ bits of space for path-reporting hopset)
  in each of the subsequent passes.
\end{theorem}

%% file: HopApplications.tex
\section{$(1+\epsilon)$-Approximate Shortest Paths in Weighted Graphs}\label{sec:hopApplications}
Consider the problem of computing $(1+\epsilon)$-approximate shortest paths (henceforth $(1+ \epsilon)$-ASP) for all pairs in $S \times V$, 
for a subset $S$, $|S| = s$, of designated source vertices, in a weighted undirected $n$-vertex
graph $G = (V,E,\omega)$ with aspect ratio $\Lambda$.

Let $\epsilon, \rho >0$ be parameters, and assume that $s = O(n^{\rho})$.
Our dynamic streaming algorithm for this problem computes a path-reporting $(1+\epsilon, \beta)$-hopset $H$ of $G$
with $\beta = O(\frac{\log n}{\epsilon \rho})^{1/\rho}$ using the algorithm described in Section~\ref{sec:aspectRatio},
with $\kappa = 1/\rho$.
By Theorem~\ref{thm:AspRatioReduction}, $|H| = O(\log n \cdot n^{1 + \rho})$,
the space complexity of this computation is $O(n \cdot \log^3 n \cdot \log \Lambda)$ for the first pass and
 $O(n^{1+\rho})\cdot \log^{O(1)} n$ for subsequent passes,
and the number of passes is $O(\beta) = poly(\log n)$.

Once the hopset $H$ has been computed, we conduct  $(1+\epsilon)$-approximate 
Bellman-Ford explorations in $G\cup H$ to depth $\beta$ from all the sources of $S$.
(See the algorithm from Section~\ref{sec:WeightedForest}.) By Theorem~\ref{thm: SingleBFE},
this requires $O(\beta)$ passes of the stream, and space $O(|S|\cdot n \cdot poly(\log n, \log \Lambda))$,
and results in $(1+\epsilon)$-approximate distances $d^{(\beta)}_{G \cup H}(s,v)$, for all
$(s,v) \in S \times V$. (Note that following every pass over $G$, we do an iteration of Bellman-Ford over the hopset $H$ \emph{offline},
as $H$ is stored by the algorithm.)
In addition, for every pair $(s,v) \in S \times V$, we also get the parent
of $v$ on the exploration rooted at source $s$. We compute the path $\pi_{G \cup H}(s,v)$ between $s$ and $v$ in graph $G \cup H$ from these parent pointers.
As described in Section~\ref{sec:PathReporting}, the path-reporting property of our hopset $H$
enables us to replace any hopset edge $e = (x,y) \in H$ on the path $\pi_{G \cup H}(s,v)$ with a corresponding path $\pi_G(x,y)$ in $G$.
By definition of the hopset, we have 
$$d_G(s,v) \le d^{(\beta)}_{G\cup H}(s,v) \le (1+\epsilon)\cdot d_G(s,v),$$
and the estimates $\hat{d}(s,v)$ computed by our approximate Bellman-Ford
algorithm satisfy 
$$d^{(\beta)}_{G \cup H}(s,v) \le \hat{d}(s,v) \le (1+\epsilon) \cdot d^{(\beta)}_{G \cup H}(s,v).$$
Thus, we have 
$$d_G(s,v) \le \hat{d}(s,v) \le (1+\epsilon)^2 \cdot d_G(s,v).$$

By rescaling $\epsilon' = 3 \epsilon$, we obtain $(1+\epsilon)$-approximate $S\times V$ paths, the total
space complexity of the algorithm is $O(n^{1 + \rho}\cdot poly(\log n, \log \Lambda))$,
and the number of passes is $poly(\log n)$.
We derive the following theorem:
\begin{theorem}\label{thm:wdDistances}
For any parameters $\epsilon, \rho >0$, and any $n$-vertex undirected weighted graph 
$G = (V,E,\omega)$ with polynomial in $n$ aspect ratio, and any set $S \subseteq V$
of $n^{\rho}$ distinguished sources, $(1+\epsilon)$-ASP for $S \times V$ can be 
computed in dynamic streaming setting in $\tilde{O}(n^{1+\rho})$ space and $\log^{\frac{1}{\rho} + O(1)} n = polylog(n)$ passes.
\end{theorem}

%% file: Appendix.tex

\begin{appendices}
\section{Hash Functions}\label{sec:hashH}
Algorithms for sampling from a dynamic stream are inherently randomized and often use hash functions as a source of randomness.
 A hash function $h$ maps elements from a given input domain to an output domain of bounded size. 
 Ideally, we would like to draw our hash function randomly from the space of all possible functions on the given input/output domain. 
 However, since we are concerned about the space used by our algorithm, we will rely on hash functions with limited independence. 
A family of functions $ H = \{h : \mathcal{U} \rightarrow [m] \}$, from a universe $\mathcal{U}$ to $[m]$, for some positive integer $m$, 
is said to be \emph{$k$-wise independent},
if it holds that, when $h$ is chosen uniformly at random from $H$ then for any $k$ distinct elements $x_1,x_2,\cdots, x_k \in \mathcal{U}$, 
 and any $k$ elements $z_1,z_2,\cdots, z_k \in [m]$, $x_1,x_2,\cdots, x_k$ and mapped by $h$ to $z_1,z_2,\cdots, z_k$ with probability $1/m^k$,
i.e., as if they were perfectly random. Such functions can be described more compactly, but are sufficiently random to allow formal guarantees to be proven.

The following lemma summarizes the space requirement of limited independence hash functions:
\begin{lemma}[\cite{CARTER1979}]\label{lem:2wiseSpace}
A function drawn from a family of $k$-wise independent hash functions can be encoded in $O(k \log n)$ bits.
\end{lemma}

Specifically, we will be using \emph{pairwise independent} hash functions.

The following lemma, a variant of which has also been proved in~~\cite{GibbKKT15, KingKT15} in a different context, is proved here for the sake of completeness.
 
 \begin{lemma}\label{lem:Pairwise}
	Let $h: \mathcal{U}  \rightarrow [2^{\lambda}]$ be a hash function sampled uniformly at random from a family of pairwise independent hash functions $\mathcal{H}$.
	 If we use $h$ to hash elements of a given set $\mathcal{S} \subseteq \mathcal{U}$ such that $|\mathcal{S}| = s$, then a specific 
	 element $d \in \mathcal{S}$ hashes to the set $[ 2^{\lambda - \ceil*{\log s } - 1}]$ and no other element of  $\mathcal{S}$ does so with probability at least $\frac{1}{8s}$.
\end{lemma}

\begin{proof}
Denote $t = \lambda - \ceil*{\log s } - 1$. Let $d^{Only}$ be the event that only the element $d \in \mathcal{S}$ and no other element $d' \in \mathcal{S}$  hashes to the set $[ 2^{\lambda - \ceil*{\log s} - 1}] =  [2^t]$.
Note that $\frac{1}{4s} \le\frac{2^t}{2^{\lambda}} \le \frac{1}{2s}$. It follows that

\begin{equation*}
	\begin{aligned}
		\underset{h\sim \mathcal{H}}{Pr}[d^{Only}] &= \underset{h\sim \mathcal{H}}{Pr}\bigg[ h(d) \in [ 2^t] \bigwedge_{d' \in \mathcal{S} \setminus \{d\}} h(d') \notin[ 2^t] \bigg] \\
		&=  \underset{h\sim \mathcal{H}}{Pr}\bigg[h(d) \in [ 2^t]\bigg]\cdot \underset{h\sim \mathcal{H}}{ Pr}\bigg[\bigwedge_{d' \in \mathcal{S} \setminus \{d\}} h(d') \notin[ 2^t] \mid h(d) \in [ 2^t]        \bigg ] \\
		  &\ge  \underset{h\sim \mathcal{H}}{Pr}\bigg[h(d) \in [ 2^t] \bigg] \cdot \bigg( 1 - \sum_{d' \in \mathcal{S} \setminus \{d\}} \underset{h\sim \mathcal{H}}{Pr}\bigg[ h(d') \in[ 2^t] \mid h(d) \in [ 2^t] \bigg]  \bigg) 
	\end{aligned}
\end{equation*}

By pairwise independence,
\begin{equation*}
	\begin{aligned}
		 \underset{h\sim \mathcal{H}}{Pr}\bigg[h(d') \in[ 2^t]~|~h(d) \in [ 2^t] \bigg]  &=  \underset{h\sim \mathcal{H}}{Pr}\bigg[ h(d') \in[ 2^t] \bigg] \\
		\text{Hence,~~} \underset{h\sim \mathcal{H}} {Pr}[d^{Only}]  &\ge \underset{h\sim \mathcal{H}}{Pr}\bigg[h(d) \in [ 2^t] \bigg] \cdot \bigg( 1 - \sum_{d' \in \mathcal{S} \setminus \{d\}} \underset{h\sim \mathcal{H}}{Pr}	\bigg[h(d') \in[ 2^t] \bigg]   \bigg) \\
 		&= \frac{2^t}{2^{\lambda}} . \bigg(  1- \sum_{d' \in \mathcal{S} \setminus \{d \}} \frac{2^t}{2^{\lambda}}  \bigg) 
 		\ge  \frac{1}{4s} . \bigg( 1- \sum_{d' \in \mathcal{S} \setminus \{d\} }  \frac{1}{2s}  \bigg) \\
 		 &= \frac{1}{4s} . \big( 1 -  (s -1) \frac{1}{2s} \big) > \frac{1}{4s} \cdot \frac{1}{2} = \frac{1}{8s} 
	\end{aligned}
\end{equation*}

\end{proof}

Lemma~\ref{lem:Pairwise} implies the following corollary:

\begin{corollary}\label{cor:pairwise}
	Let $h: \mathcal{U} \rightarrow [2^{\lambda}]$ be a hash function sampled uniformly at random from a family of pairwise independent hash functions $\mathcal{H}$.
	 If we use $h$ to hash elements of a given set $\mathcal{S} \subseteq \mathcal{U}$ with $|\mathcal{S}| = s$, 
	 then exactly one element in $\mathcal{S}$ hashes to the set $[ 2^{t}]$, $t = \lambda - \ceil*{\log s} -1$, with probability at least $\frac{1}{8}$.
\end{corollary}

\begin{proof}
	Let $OneElement$ be the event that exactly one of the $s$ elements in the set $\mathcal{S}$ hashes to the set $[ 2^{t}]$.
	The event $OneElement$ can be described as the event $d^{Only}$ from Lemma~\ref{lem:Pairwise} occurring for one of the elements $d \in \mathcal{S}$, i.e., 
\begin{equation*}
	\begin{aligned}
		\underset{h\sim \mathcal{H}}{Pr}[OneElement]  &=  \sum_{d \in \mathcal{S}} \underset{h\sim \mathcal{H}}{Pr}\big[d^{Only} \big] \\
	            &\ge \sum_{d \in \mathcal{S}} \frac{1}{8s} = 1/8
	\end{aligned}
\end{equation*}

\end{proof}

\section{New Sparse Recovery and $\ell_0$-Sampling Algorithms}\label{sec:firstAppedix}
In this appendix, we show that our sampler \emph{FindNewVisitor} (See Algorithm~\ref{alg:findNewVisitor}) in the dynamic streaming setting
can also be used to provide a general purpose $1$-sparse recovery and $\ell_0$-sampler in the strict turnstile model.
(Recall that a dynamic streaming setting is called \emph{strict turnstile model}, if ultimate values of all elements at the
end of the stream are non-negative, even though individual updates may be both positive or negative.)
We consider a vector $\arrow{a} = (a_1,a_2, \ldots, a_n)$, which comes in the form of a stream
of updates. Each update is of the form $\langle i, \Delta a_i\rangle$, and it means that one needs to add the quantity $\Delta a_i$
to the $i^{th}$ coordinate of the vector $\arrow{a}$. As was mentioned above, we assume that for each $i$, the ultimate sum of all
the update values $\Delta a_i$ that refer to the $i^{th}$ coordinate is non-negative.

We say that the vector $\arrow{a}$ is \emph{$1$-sparse}, if it contains exactly one element in its support.
The support of $\arrow{a}$ denoted $supp(\arrow{a})$ is the set of coordinates $a_i \neq 0$.

In the \emph{$1$-sparse recovery} problem, if the input vector $\arrow{a}$ is $1$-sparse, the algorithm needs
to return the (only) coordinate $i$ in the support of $\arrow{a}$ and its ultimate value $a_i$. Otherwise, the algorithm returns 
$\perp$ (indicating a failure). Ganguly~\cite{Ganguly2007CountingDI} devised an algorithm for this problem in the strict turnstile 
setting, which occupies space $O(\log M + \log n)$, where $M$ is the maximum value of any coordinate $a_j$ for any $j \in [n]$
during the stream. Cormode and Firmani~\cite{CormodeFirmani} devised an algorithm with the same space complexity which applies for integer update 
values in general turnstile model (in which ultimate negative multiplicities of the coordinates, also known as frequencies, are allowed).
(See Section~\ref{apx:1Sparse}.) We show an alternative solution to that of Ganguly~\cite{Ganguly2007CountingD} with the same space complexity.

\subsection{$1$-Sparse Recovery}\label{apx:1Sparse}
The basic idea is to use CIS-based encodings $\nu$ described in Section~\ref{sec:Encodings}.
Throughout the execution of our algorithm, we maintain a sketch $\mathcal{L}$ which is a two-dimensional vector in $\mathbb{R}^2$
and a counter $ctr$. Initially, $\mathcal{L} = \arrow{0}$ and $ctr = 0$.
Every time we receive an update $\langle i, \Delta a_i\rangle$, 
we update $\mathcal{L}$ as $\mathcal{L}  = \mathcal{L} + \nu(i) \cdot \Delta a_i $
and update $ctr$ as $ctr = ctr + \Delta a_i$.
At the end of the stream, if $ctr \neq 0$, we compute $\mathcal{L}' = \frac{\mathcal{L}}{ctr}$. (If $ctr =0$ , we return $\phi$, indicating that the input vector is empty.)
The algorithm then tests if $\mathcal{L}' \in \{\nu(1), \nu(2),\ldots, \nu(n) \}$, and if it is the case, i.e., $\mathcal{L}'  = \nu(i)$ for some $i \in [n]$, 
then it returns $(i, ctr )$, and $\perp$ otherwise.

For the analysis, observe that $\mathcal{L} = \sum_{i =1}^{n} \nu(i)\cdot a_i$ and $ctr = \sum_{i=1}^{n} a_i$.
If $|supp (\arrow{a}) | = 1$, then let $\{i\} = supp (\arrow{a})$. In this case, $\mathcal{L} = \nu(i)\cdot a_i$ and $ctr = a_i$
and thus $\mathcal{L}' = \frac{\mathcal{L}'}{ctr} = \nu(i)$. We can therefore retrieve $i$ from $\nu(i)$.
On the other hand, if $|supp (\arrow{a}) |  =0$,  then the algorithm obviously returns $\perp$. 
Finally, by Lemma~\ref{lem:convex}, if $|supp (\arrow{a}) | \ge 2$, then $\mathcal{L}' \notin \{\nu(1), \nu(2),\ldots, \nu(n) \}$, and
in this case algorithm returns a message \emph{too dense}.

In the context of our application of the above algorithm to computing near-additive spanners, one can just keep an 
encoding table which records $\nu(i)$ for every $i \in [n]$.

However, for a general-purpose $1$-sparse recovery, one needs to be able to compute $\nu(i)$ (given an
index $i \in [n]$) using $polylog(n)$ space. One also needs to compute $i$ from $\nu(i)$ using small space.
Recall that we define $R = \Theta(n^{3/2})$ and $\nu(1), \nu(2),\ldots, \nu(n)$, $n = \Theta(R^{2/3})$ are the $n$ vertices of 
the convex hull of the set of integer points within a radius-$R$ disc, centered at the origin, ordered clockwise.
These vectors can be computed by Jarn\'{i}k's constriction (See~\cite{JarnkberDG, ElkinCSmith}).
The latter can be computed in $O(\log^2 n)$ space, but the fastest $\log$-space algorithms that we know
for this task retrieve all vertices one after another and thus require time at least linear in $n$.

To speed up this computation, we next describe another encoding $\sigma$ which maps $[n]$ into
$\mathbb{Z}^5$. As a result, each encoding $\sigma(i)$ uses by constant factor more space 
than $\nu(i)$. On the other hand, we argue below that $\sigma(i)$ and $\sigma^{-1} (\mathcal{L})$
can be efficiently computed using $\log$-space, for any $i \in [n]$ and any feasible vector $\mathcal{L} \in \mathbb{Z}^5$.
(By a  \emph{feasible vector}, we mean here that $\mathcal{L}$ is in the range of the mapping defined by $\sigma$.)

Let $R = n$ and consider a $5$-dimensional sphere $\mathbb{S}$, centered at origin. The sphere contains 
$\Theta(R^3)$ integer points, but we will use just $R$ of them. Specifically, for any $i \in [n]$, let $(p_i, q_i, r_i, s_i)$
be a fixed four-square representation of $R^2 - i^2$, i.e., $R^2 - i^2 = p_i^2 + q_i^2 + r_i^2 + s_i^2$, 
where $p_i, q_i, r_i, s_i \in \mathbb{N}$. Then we define $\sigma(i) = (p_i, q_i, r_i, s_i)$.
(Such a representation exists for every natural number by Lagrange's four-square theorem, see, e.g.,~\cite{PollackEnrique}.)

There exist a number of efficient randomized (Las Vegas) algorithms~\cite{PollackEnrique, RabinShalit} 
for computing a four-square representation of a given integer.
One of these algorithms is deterministic. It is known to require time polynomial in $O(\log n)$,
assuming Heath-Brown's conjecture~\cite{heath-brown_1978} that the least prime congruent to $a$ $(\mod q)$,
when $\gcd(a,q) = 1$, is at most $q \cdot (\log q)^2$. (See Section~\cite{PollackEnrique})

Another alternative is to use a randomized algorithm of Rabin and Shalit~\cite{RabinShalit} which 
has been recently improved by Pollack and Trevi{\~n}o~\cite{PollackEnrique} and requires 
expected time $O(\log^2n/\log\log n)$.

The problem with it is, however, that it may return different representations $\sigma(i)$, when invoked several times
on the same number $R^2 - i^2$, for some $i \in [n]$. To resolve this issue, one may use Nisan's
pseudorandom generator~\cite{Nisan92pseudorandom} to generate the random string used by all the invocations
of Pollack and Enrique's algorithm~\cite{PollackEnrique} from a seed of polylogarithmic ($O(\log^2 n)$) length.
The latter seed can be stored by our algorithm. This ensures consistent computations of four-square representations
of different integers by our algorithm.The resulting  random string (produced by Nisan's generator) is 
indistinguishable from a truly random one from the perspective of any $polylog(n)$-space bounded algorithm.
Since both our algorithm and that of Pollack and Trevi{\~n}o~\cite{PollackEnrique} are $polylog(n)$-space bounded,
this guarantees the correctness of the overall computation.

\section{$\ell_0$-sampling}\label{apx:l0Sampling}
To demonstrate the utility of our new $1$-sparse recovery algorithm, we point out that this routine directly
gives rise to an $s$-sparse recovery algorithm, for an arbitrarily large $s$. (For example, see the 
description of the first pass of sub-phase $j$ of interconnection step in Section~\ref{sec:interSubPhase}.)

A vector $\arrow{a}$ is said to be \emph{$s$-sparse} if $|supp (\arrow{a})| \le s$. In the $s$-sparse 
recovery problem, the algorithm accepts as input a vector $\arrow{a}$. If the vector $\arrow{a}$ is not $s$-sparse
or $\arrow{a} = \arrow{0}$, the algorithm needs to report $\perp$. Otherwise, with probability at least $\delta > 0$,
for a parameter $\delta >0$, the algorithm needs to return the original vector $\arrow{a}$.
A direct approach to $s$-sparse recovery via $1$-sparse recovery is described in~\cite{Ganguly2007CountingDI}
and in Section $2.3.2$ of~\cite{CormodeFirmani}. 
It produces an algorithm whose space is $O(s \log \frac{1}{\delta})$ times the space
of the $1$-sparse recovery algorithm. One can use our $1$-sparse recovery algorithm instead of those of~\cite{Ganguly2007CountingDI}
or~\cite{CormodeFirmani} in it. 

Yet another application of our $1$-sparse recovery algorithm is $\ell_0$-samplers.
An $\ell_0$-sampler may return a $\perp$ (a failure) with probability at most $\delta$.
But if it succeeds, it returns a uniform (up to an additive error of $n^{-c}$, for a sufficiently
large $c$) coordinate $i$ and the corresponding value $a_i$ in the support of the 
input vector $\arrow{a}$. The scheme we describe next is close to Jowhari et al~\cite{Jowhari11},
and has a similar space complexity to it. It however uses $1$-sparse recovery directly,
while the scheme of~\cite{Jowhari11} employs $s$-sparse recovery (which, in turn, invokes 
$1$-sparse recovery). Like Jowhari et al~\cite{Jowhari11}, we first describe the algorithm 
assuming a truly random bit string of length $O(m \log n)$, where $m$ is the length of the stream
and $n$ is the length of the input vector $\arrow{a}$. We then replace it by string produced by Nisan's
 pseudorandom generator out of a short random seed. 
 This seed is stored by the algorithm. (Its length is $O(\log^2 n)$ like in~\cite{Jowhari11}.)
 
 The algorithm tries $\log n$ scales $j = 1,2,\ldots, \log n$, and each scale $j$ corresponds to
 a guess of $s = |supp (\arrow{a})|$ being in the range $2^{j-1} \le s \le 2^j$.
 On scale $j$ each coordinate $i$ is consistently sampled with probability $2^{-j}$,
 and a $1$-sparse recovery algorithm attempts to recover the subsampled vector.
 
 For a fixed coordinate $i$, and for $j$ such that $2^{j-1} < s \le 2^j$,
 the probability that only $i$ will be sampled is $\frac{1}{2^j} \cdot (1 - \frac{1}{2^j})^{j-1} \ge \frac{1}{2s} (1 - \frac{1}{s})^{s-1} 
 \ge \frac{e^{-1}}{2s}$.
 
 Since the event of two fixed distinct coordinates to be discovered are disjoint, it follows that the probability of  the sampler
 to recover \emph{some} coordinate is at least $\frac{e^{-1}}{2}$.
 Conditioned on its success to retrieve an element, by symmetry, it follows that the probabilities of different coordinates in 
 $supp (\arrow{a})$ to be recovered are equal. 
 Once the truly random source is replaced by the string produced by Nisan's pseudorandom
 number generator, the probabilities, however, will be skewed by an additive term of $n^{-c}$, for a sufficiently large constant $c > 0$.
 
 Similarly to the argument in~\cite{Jowhari11}, no $polylog(n)$-space tester is able to distinguish between the truly random string 
 and the one produced by Nisan's pseudorandom generator. Thus, in particular, they are indistinguishable for our ($polylog(n)$-space bounded)
 algorithm. 
 
 Viewed as a tester, our algorithm may be fed with a specific set of non-zero coordinates in the support of its input vector and any 
 specific coordinate $i$ in the support that the algorithm can test whether it is returned. 
 (This tester is $polylog(n)$-space bounded.)
 
 The overall space requirement of the algorithm in $O(\log n)$ times the space requirement of 
 the $1$-sparse recovery routine. The latter is $O(\log n)$ as well.
 In addition to this space of $O(\log^2 n)$, the algorithm also needs to remember the random seed
 of Nisan' generator which is of length $O(\log^2 n)$ as well. 
 
 The failure probability of the algorithm is, as was shown above $e^{-1}/2$. If we want to decrease it to $\delta$, 
 we can run $O(log 1/\delta$ copies of this algorithm in parallel, and pick an 
 arbitrary copy in which the algorithm succeeded.
 (If there exists such a copy, otherwise the algorithm returns a failure.)
 The overall space of the resulting algorithm becomes $O(\log^2 n \log 1/\delta)$,
  To summarize:
  
  \begin{theorem}
  Our algorithm provides an $L_0$-sampler with failure probability at most $\delta > 0$,
  for a parameter $\delta$, and additive error $n^{-c}$, for an arbitrarily large 
  constant $c$ which affects the constant hidden in the $O$-notation of space.
  Its space requirement is $O(\log^2n \cdot \log 1/\delta)$.
  \end{theorem}

\end{appendices}